\documentclass[a4paper,UKenglish,cleveref,autoref,thm-restate]{lipics-v2021}

\newboolean{long}

\setboolean{long}{false} %
\ifthenelse{\boolean{long}}{
	\NewDocumentEnvironment{prooflater}{m}{\begin{proof}}{\end{proof}\ignorespacesafterend}
	\NewDocumentEnvironment{proofsketch}{o +b}{}{\ignorespacesafterend}
	\newcommand{\restateref}[1]{}
	\NewDocumentEnvironment{statelater}{m}{}{}

	\NewDocumentCommand{\onlyShort}{+m}{}
	\NewDocumentCommand{\onlyLong}{+m}{#1}
	
}{
	\NewDocumentEnvironment{prooflater}{m +b}{%
		\expandafter\global\expandafter\def\csname#1\endcsname{\begin{proof}#2\end{proof}}%
	}{\ignorespacesafterend}
	\NewDocumentEnvironment{proofsketch}{O{Proof sketch.}}{\begin{proof}[#1]}{\end{proof}\ignorespacesafterend}
	\usepackage{apptools}
    \newcommand{\restateref}[1]{}

	\NewDocumentEnvironment{statelater}{m +b}{%
		\expandafter\global\expandafter\def\csname#1\endcsname{#2}%
	}{\ignorespacesafterend}

	\NewDocumentCommand{\onlyShort}{+m}{#1}
	\NewDocumentCommand{\onlyLong}{+m}{}
	
}

\pdfoutput=1 %
\hideLIPIcs  %

\graphicspath{{./graphics/}}%

\title{The (Parameterized) Complexity of Ordering a Graph While Avoiding a Forbidden Pattern}

\titlerunning{The Complexity of Ordering a Graph While Avoiding a Forbidden Pattern} %

\author{Thomas Depian}{TU Wien, Vienna, Austria}{tdepian@ac.tuwien.ac.at}{https://orcid.org/0009-0003-7498-6271}{Supported by the Vienna Science and Technology Fund (WWTF) [10.47379/ICT22029].}
\author{Simon D.~Fink}{TU Wien, Vienna, Austria}{sfink@ac.tuwien.ac.at}{https://orcid.org/0000-0002-2754-1195}{Supported by the Vienna Science and Technology Fund (WWTF) [10.47379/ICT22029].}
\author{Alexander Firbas}{TU Wien, Vienna, Austria}{afirbas@ac.tuwien.ac.at}{https://orcid.org/0009-0007-2049-2144}{Supported by the Vienna Science and Technology Fund (WWTF) [10.47379/ICT22029] and the Austrian Science Fund (FWF) [10.55776/Y1329].}
\author{Robert Ganian}{TU Wien, Vienna, Austria}{rganian@ac.tuwien.ac.at}{https://orcid.org/0000-0002-7762-8045}{Supported by the Vienna Science and Technology Fund (WWTF) [10.47379/ICT22029] and the Austrian Science Fund (FWF) [10.55776/Y1329].}
\author{Martin Nöllenburg}{TU Wien, Vienna, Austria}{noellenburg@ac.tuwien.ac.at}{https://orcid.org/0000-0003-0454-3937}{Supported by the Vienna Science and Technology Fund (WWTF) [10.47379/ICT22029].}
\author{Marie Diana Sieper}{Universität Würzburg, Germany}{marie.sieper@uni-wuerzburg.de}{https://orcid.org/0009-0003-7491-2811}{}

\authorrunning{T.~Depian, S.~D.~Fink, A.~Firbas, R.~Ganian, M.~Nöllenburg, and M.~D.~Sieper} %

\Copyright{Thomas Depian, Simon D.~Fink, Alexander Firbas, Robert Ganian, Martin Nöllenburg, and Marie Diana Sieper} %



\nolinenumbers %

\EventEditors{John Q. Open and Joan R. Access}
\EventNoEds{2}
\EventLongTitle{42nd Conference on Very Important Topics (CVIT 2016)}
\EventShortTitle{CVIT 2016}
\EventAcronym{CVIT}
\EventYear{2016}
\EventDate{December 24--27, 2016}
\EventLocation{Little Whinging, United Kingdom}
\EventLogo{}
\SeriesVolume{42}
\ArticleNo{23}

\usepackage{hyperref}
\usepackage{complexity}
\usepackage{mathtools}
\usepackage{mdframed}
\usepackage{xspace}
\usepackage{pifont}
\usepackage{cite}
\usepackage[normalem]{ulem}
\usepackage{booktabs}

\Crefname{claim}{Claim}{Claims}
\Crefname{observation}{Observation}{Observations}

\let\oldrestatable\restatable
\def\restatable{\expandafter\oldrestatable}

\newcommand{\vi}{\textsf{\textup{vi}}}
\newcommand{\vc}{\textsf{\textup{vc}}}
\newcommand{\nd}{\textsf{\textup{nd}}}

\newcommand{\XNLP}{\textsf{\textup{XNLP}}}
\DeclareMathOperator{\Image}{Im}

\newcommand{\Size}[1]{\ensuremath{\left\vert #1 \right\vert}}
\newcommand{\BigO}[1]{\ensuremath{\mathcal{O}(#1)}}
\newcommand{\probname}[1]{{\normalfont\textsc{#1}}}
\newcommand{\probnameHeader}[1]{{\textsc{#1}}}

\newcommand{\probdef}[3]{%
    \begin{mdframed}
		\probname{#1}
		\begin{description}
            \item[Given] #2
            \item[Question] #3
        \end{description}
	\end{mdframed}
}%
\newcommand{\PA}{\probname{Pattern Avoidance}\xspace}
\newcommand{\PAHeader}{\probnameHeader{Pattern Avoidance}\xspace}
\newcommand{\PAOneF}{\probname{1FF \PA}\xspace}

\newcommand{\PAOneFLong}{\probname{1-Forced Edge \PA on Forests}\xspace}
\newcommand{\PAOneFLongUnderline}{\probname{\underline{1}-\underline{F}orced Edge \PA on \underline{F}orests}\xspace}

\newcommand{\PANormalized}{\probname{Normalized}\xspace}
\newcommand{\PAOneFOneS}{\probname{1S\PAOneF}\xspace}
\newcommand{\PAOneFOneSLong}{\probname{1-Side Fixed \PAOneFLong}\xspace}
\newcommand{\PAOneFOneSLongUnderline}{\probname{\underline{1}-\underline{S}ide Fixed \PAOneFLongUnderline}\xspace}
\newcommand{\PAOneFTwoS}{\probname{2S \PA}\xspace}
\newcommand{\PAOneFTwoSLong}{\probname{2-Sides Fixed \PA}\xspace}
\newcommand{\PAOneFTwoSLongUnderline}{\probname{\underline{2}-\underline{S}ides Fixed \PA}\xspace}

\newcommand{\Instance}{\ensuremath{\mathcal{I}}\xspace}
\newcommand{\InstanceLong}{\ensuremath{(G,P)}\xspace}
\newcommand{\InstanceLongOneS}{\ensuremath{(G,P,U,\prec_U)}\xspace}
\newcommand{\PatternLong}{\ensuremath{(V(P), \prec_P, E^+(P), E^-(P))}\xspace}

\NewDocumentCommand{\Pred}{o m}{\ensuremath{\text{pred}(\IfNoValueF{#1}{#1, }#2)}\xspace}
\NewDocumentCommand{\Succ}{o m}{\ensuremath{\text{succ}(\IfNoValueF{#1}{#1, }#2)}\xspace}
\newcommand{\compress}[1]{\ensuremath{c\left(#1\right)}}
\newcommand{\multiset}[1]{\ensuremath{\mathcal{T}\left(#1\right)}}
\newcommand{\multisetsize}[1]{\ensuremath{\Size{#1}_{>0}}}
\newcommand{\eqclasses}[2]{\ensuremath{#1/{#2}}\xspace}

\Crefname{instprop}{Instance Property}{Instance Properties}
\creflabelformat{instprop}{#2{N#1}#3}
\Crefname{solprop}{Solution Property}{Solution Properties}
\creflabelformat{solprop}{#2{S#1}#3}
\Crefname{valprop}{Validity Constraint}{Validity Constraints}
\creflabelformat{valprop}{#2{V#1}#3}

\newcommand{\NewText}[1]{{\color{black}#1}}

\begin{document}

\maketitle

\begin{abstract}
In this paper, we study the \PA\ problem of determining whether a given graph~$G$ admits a linear vertex order which avoids a given {pattern $P$, i.e., a vertex sequence with some forced and forbidden edges, on every suborder}. Such patterns form a natural ordered counterpart to induced subgraphs in the order-invariant setting, and it is known that \PA\ captures a broad variety of graph problems including \textsc{Bandwidth}, \textsc{Vertex Coloring}, \textsc{Queue Number}, {and extends to vertex-deletion problems such as \textsc{Odd Cycle Transversal}}. 

We show that \PA\ is $\Sigma_2^{\P}$-complete and {furthermore} remains intractable (in both the classical and parameterized sense) even under a variety of severe restrictions to both the pattern $P$ and the graph $G$. 
{As our main contributions, we complement these lower bounds with the following tractability results, which provide a unifying framework for recognizing pattern-definable graph classes:}
\begin{itemize}
\item a fixed-parameter algorithm w.r.t.\ the vertex integrity of $G$ plus $|V(P)|$, 
\item a fixed-parameter algorithm w.r.t.\ the neighborhood diversity of $G$ plus $|E(P)|$, and
\item a polynomial algorithm for \PA\ {on forests} for almost all constant-sized patterns. 
\end{itemize}

\end{abstract}

\section{Introduction}

Characterizing graph classes by forbidden structures is a well-studied research topic in graph theory. 
Among the most prominent examples are forbidden-minor characterizations, e.g., for planar graphs (no $K_5$ or $K_{3,3}$ minor), series-parallel graphs (no $K_4$ minor), or outerplanar graphs (no $K_4$ or $K_{2,3}$ minor). 
More generally, Robertson and Seymour's famous graph minor theorem~\cite{DBLP:journals/jct/RobertsonS04} states that every minor-closed graph family has a characterization by a finite set of forbidden minors. 
But minors are not the only type of forbidden structures. 
For instance, chordal graphs are the graphs without induced cycles of length at least four, bipartite graphs are the graphs without odd cycles, %
and cographs are the graphs without an induced path $P_4$.
The strong perfect graph theorem~\cite{crst-spgt-06} %
shows that perfect graphs are exactly the graphs without induced odd holes or odd antiholes. 

In this paper, we are interested in forbidden \emph{patterns} in ordered graphs, which are ordered configurations $P$ of vertices with some vertex pairs forced to have an edge, while others are forbidden to be connected by an edge. 
A graph $G$ \emph{avoids} the pattern $P$ if it permits a linear order of its vertices $V(G)$ such that $P$ does not occur as an induced ordered subgraph in that order (for the formal definition see \Cref{sec:preliminaries}). 
The study of forbidden patterns started more than 40 years ago with characterizations of graphs with forbidden 3-vertex patterns~\cite{s-rbtgcgpigpcgnig-82,d-fos-90} which give rise to polynomial-time and often linear-time recognition algorithms, even for sets of forbidden 3-vertex patterns~\cite{hmr-owfp-14}.
Feuilloley and Habib~\cite{FH.GCF.2021} gave a characterization of graph classes that can be defined via forbidden patterns on $3$ vertices, and already for such simple patterns these classes include forests, bipartite graphs, interval graphs, chordal graphs, comparability graphs and others; for example, bipartite graphs can be characterized by excluding the ordered path on $3$ vertices~\cite{FH.GCF.2021}.
In contrast to the 3-vertex pattern case, Duffus, Ginn, and Rödl~\cite{dgr-ccosr-95} showed \NP-completeness of the recognition problem for almost all biconnected patterns.
Moreover, forbidden patterns on $4$ vertices have been studied in several works to date, albeit predominantly not in an algorithmic setting~\cite{froese_persistent_2021,davies_coloring_2023,hell_monotone_2012,feuilloley_classifying_2021}.

In the context of so-called \emph{linear graph layouts}, larger and not necessarily connected forbidden patterns also play an important role. 
For instance, it is well known that the graphs with stack number 1 or queue number 1~\cite{DujmovicW04} are exactly the graphs avoiding the 2-twist pattern of two crossing edges or the 2-rainbow pattern of two nested edges, respectively. 
For graphs with larger stack number $k$, the corresponding $k$-twist is obviously a forbidden pattern, but its absence is not a sufficient characterization of stack-number-$k$ graphs; it does provide an $\BigO{k \log k}$ upper bound, though~\cite{d-ibccg-22}. 
In contrast, graphs of queue number $k$ can be precisely characterized by forbidding the $k$-rainbow~\cite{DBLP:journals/siamcomp/HeathR92}---see also the recent work on the topic~\cite{haun_et_al:LIPIcs.STACS.2025.45}.
In terms of the computational complexity, recognizing graph classes with a forbidden pattern is generally \NP-complete
even for constant size patterns, e.g., recognizing graphs with queue number 1 (and thus avoiding a 4-vertex, 2-edge pattern) is \NP-complete~\cite{DBLP:journals/siamcomp/HeathR92}, as well as recognizing graphs of bandwidth $k$~\cite{GareyJohnson}, which are characterized by avoiding a $(k+2)$-vertex, 1-edge pattern. 

In this paper, we take a more fine-grained perspective on the complexity of \PA, i.e., the problem of recognizing graphs defined by a single forbidden pattern $P$. 
While a first natural direction to parameterize this problem is bounding the size of $P$, either in its number $|V(P)|$ of vertices or in its number $|E(P)|$ of forced and forbidden edges, the above observations immediately prevent the tractability of \PA when parameterized by the pattern size alone: even for $|V(P)| = 4$ or $|E(P)| = 1$, \PA is \NP-complete. Note that this is in contrast to recognizing graph classes defined by a bounded-size forbidden induced subgraph, which can be done in polynomial time by enumerating all induced subgraphs of the given size. In fact, as the starting point for our considerations, we list several complexity-theoretic lower bounds: %
\PA is
\begin{itemize}
\item $\Sigma_2^P$-complete (\cref{thm:sigma-completeness}, via a new reduction);
\item \XNLP-complete w.r.t.\ $|V(P)|$ even on trees of constant pathwidth (\cref{thm:bandwidth-hardness}(a)\NewText{, exploiting the known \XNLP-completeness of \probname{Bandwidth}}); and
\item \W[1]-hard w.r.t.\ the treedepth of the graph $G$, even on patterns with $|E(P)|=1$ (\cref{thm:bandwidth-hardness}(b)\NewText{, since \PA generalizes \probname{Bandwidth}}).
\end{itemize}

These lower-bound results \NewText{highlight, on the one hand, that the complexity with respect to the treedepth of $G$ plus $\Size{V(P)}$ remains open but suggest, on the other hand,} %
that %
\NewText{(more-restricted)} structural properties of the graph $G$ are needed in order to obtain fixed-parameter tractability of \PA, depending on whether one additionally bounds $|V(P)|$ or $|E(P)|$. %
Our first positive result (\cref{thm:vifpt}) is a fixed-parameter algorithm for \PA, parameterized by the vertex integrity of $G$ plus the number of vertices $|V(P)|$ in the pattern $P$.
As there are many application where the pattern contains many vertices but only few or even constantly-many edges (see above, e.g., for \textsc{Bandwidth}), we further consider the number of edges $|E(P)|$ in the pattern as parameter.
In this setting, we provide a fixed-parameter algorithm for \PA (\cref{thm:fpt-nd-mp}), parameterized by the neighborhood diversity of $G$ plus $|E(P)|$. We remark that \NewText{for} neither of these results %
\NewText{we can instead use} treedepth \NewText{as structural parameter} without a major breakthrough in our current understanding: for the latter this is because \textsc{Bandwidth} is known to remain \W[1]-hard w.r.t.\ the treedepth of $G$, while for the former this would immediately yield a fixed-parameter algorithm for \textsc{Queue Number}---a well-established open problem in the field~\cite{bgmn-paql-22}. %

Lastly, when restricting the graph $G$ to be a forest, we obtain an \XP\ algorithm for \PA\ for almost all patterns (\Cref{thm:non-mixed})---in particular, for all patterns except those with a single forced edge both of whose endpoints overlap with forbidden edges.
For this remaining open case, we provide a hardness reduction which rules out \XP-tractability w.r.t.\ $|V(P)|$ via the techniques developed towards our proof of \Cref{thm:non-mixed}---in particular, we show that testing bandwidth $2$ for a forest when the first and last vertex in the ordering are fixed, is \NP-hard; see~\Cref{thm:one-sided-hard}. 
\Cref{tab:results} summarizes the results presented in this paper.
\begin{table}
	\caption{\NewText{Overview of our results. We indicate with \emph{Par.} that we consider this as parameter.}}
	\label{tab:results}
	\centering
	\begin{tabular}{cccc}
		\toprule
		Graph $G$ & $\Size{V(P)}$ & $\Size{E(P)}$ & Result \\
		\midrule
		-- & -- & -- & $\Sigma_2^\P$-complete (\Cref{thm:sigma-completeness})\\
		Tree, constant pathwidth & Par.\  & 1 & \XNLP-complete (\Cref{thm:bandwidth-hardness})\\
		Treedepth (Par.) & -- & 1 & \W[1]-hard (\Cref{thm:bandwidth-hardness})\\
		Vertex integrity (Par.) & Par.\ & Bnd.\ by $\Size{V(P)}$ & \FPT\ (\Cref{thm:vifpt})\\
		Neighborhood diversity (Par.) & -- & Par.\  & \FPT\ (\Cref{thm:fpt-nd-mp})\\
		Forest & Par.\  & Bnd.\ by $\Size{V(P)}$  & \XP, most patterns (\Cref{thm:non-mixed})\\
		\bottomrule
	\end{tabular}
\end{table}

The above results directly imply and give a unified explanation for, among others:
\begin{itemize}
\item the fixed-parameter tractability of \textsc{Queue Number} w.r.t.\ the vertex integrity of $G$~\cite{DFGS.LLR.2025} \NewText{(pattern in \Cref{fig:pattern-results}a)}.
\item the fixed-parameter tractability of \textsc{Bandwidth} parameterized by the neighborhood diversity $G$~\cite{GKK+.BPC.2025}, and by the target bandwidth plus the vertex integrity of $G$ \NewText{(\Cref{fig:pattern-results}b)}.
\item fixed-parameter algorithms for Vertex Deletion generalizations (cf. Lemma~\ref{lem:deletion_to_pattern}) for all graph classes captured by a forbidden pattern, %
such as:
\begin{itemize}
\item \textsc{Odd Cycle Transversal}~\cite{KratschW14}, when parameterized by the solution size plus vertex integrity, \textbf{or} by the neighborhood diversity alone \NewText{(\Cref{fig:pattern-results}c; recall that after removing the vertices from the odd cycle transversal the resulting graph is bipartite)}.
\item Whether it is possible to delete $x$ vertices in order to obtain a co-comparability graph~\cite{BozykDK0O20},
parameterized by $x$ plus the vertex integrity of $G$ \textbf{or} by the neighborhood diversity of $G$ alone \NewText{(\Cref{fig:pattern-results}d)}.
\item Whether it is possible to delete $x$ vertices in order to achieve queue number $z$, parameterized by $x+z$ plus the vertex integrity of $G$.
\item Whether it is possible to delete $x$ vertices in order to achieve bandwidth $z$, parameterized by the neighborhood diversity of $G$.
\end{itemize}
\end{itemize}

\begin{figure}
	\centering
	\includegraphics[page=2]{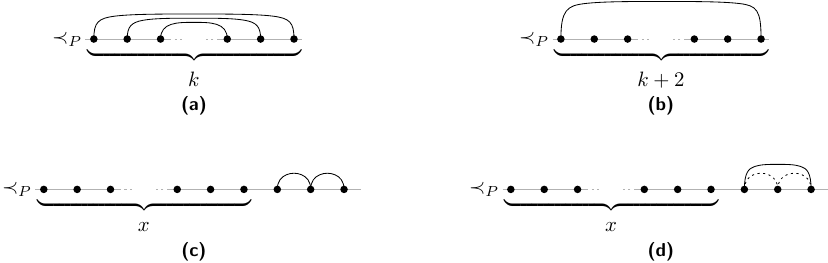}
	\caption{\NewText{The pattern $P$ that characterizes graphs 
    \textbf{\textsf{(a)}} with queue number at most $k$, 
    \textbf{\textsf{(b)}} with bandwith at most $k$, 
    \textbf{\textsf{(c)}} having odd-cycle transversal at most $x$, and 
	\textbf{\textsf{(d)}} where we can delete $x$ vertices to obtain a co-comparability graph. Solid and dashed lines represent forced and forbidden edges, respectively.
    The sub-patterns on the right for bipartite and co-comparability graphs in (c) and (d), respectively, are taken from Feuilloley and Habib~\cite{FH.GCF.2021}.
    }}
	\label{fig:pattern-results}
\end{figure}

\subparagraph*{Organization of the Paper.}
The remainder of this paper is organized as follows.
After short preliminaries in \Cref{sec:preliminaries}, we give the full details on our complexity-theoretic lower bounds in \Cref{sec:lower-bounds} (\Cref{thm:bandwidth-hardness,thm:sigma-completeness}).
For our algorithmic results, we provide a technical overview in \Cref{sec:overview}.
The details are then provided in %
the subsequent sections.
\Cref{sec:conclusion} contains concluding remarks and open questions.

\section{Preliminaries}
\label{sec:preliminaries}
For an integer $p \geq 1$, we let $[p]$ denote the set $\{1, 2, \ldots, p\}$.
Let $G$ be a graph with vertex set $V(G)$ and edge set $E(G)$.
All considered graphs are simple and undirected.
For two graphs $H$ and $G$, we write $H \subseteq G$ to denote that $H$ is a \emph{subgraph} of $G$.
For a set $X \subseteq V(G)$, we let $G[X]$ denote the subgraph of $G$ induced on $X$. %
We use the `$+$' and `$-$'-operator to add or remove vertices and edges from a graph, respectively.
For a vertex $v \in V(G)$ we let $N_G(v) \coloneqq \{u \mid uv \in E(G)\}$ denote the set of \emph{neighbors} of $v$ and $d_G(v) \coloneqq \Size{N_G(v)}$ its \emph{degree}.
Let $\prec_G$ be a linear order of the vertices $V(G)$ of $G$.
For a set $X \subseteq V(G)$ we let $\prec_G\mid_X$ denote the restriction of $\prec_G$ to the vertices in $X$.
Similarly, for an integer $i$, we let $\prec_G\mid_i$ denote the order restricted to the first $i$ vertices.
For an order $\prec$, we use $V(\prec)$ to access the vertices in $\prec$ and say that a vertex $v \in V(G)$ is \emph{spanned} by an edge $uw \in E(G)$ if $u \prec_G v \prec_G w$ holds.
We call $v \in V(G)$ the \emph{first} (\emph{last}) \emph{vertex} in $\prec_G$ if there is no vertex $u \in V(G)$ such that $u \prec_G v$ ($v \prec_G u$).
The function $\Pred[\prec_G]{v}$ returns the \emph{predecessor} of $v \in V(G)$ in $\prec_G$, which is the vertex $u$ that immediately precedes $v$ in $\prec_G$ (if it exists).
We define the \emph{successor} function $\Succ[\prec_G]{v}$ analogously.
Note that the first and last vertex have no predecessor and successor, respectively.
For the remainder of the paper, we omit the reference to the graph $G$ if it is clear from the context.
For the linear orders $\prec_A$ and $\prec_B$ on two disjoint sets $A$ and $B$, we denote with $\prec_A \oplus \prec_B$ the \emph{concatenation} of the two linear orders and extend this operation to pairs of sequences as well as sequences and linear orders.

To express some of our bounds, we will occasionally use the Knuth notation $\uparrow\uparrow$ where for an integer $z$, $2\uparrow\uparrow z$ represents an exponential tower of $2$'s of height $z$.

\subparagraph*{Vertex Integrity and Neighborhood Diversity.}
Let $G$ be a graph.
The \emph{vertex integrity} $\vi(G)$ of $G$ is the smallest integer $k$ such that $G$ contains a vertex set $S$ with the property that for every connected component $C$ of $G - S$ we have $\Size{V(C)} + \Size{S} \leq k$.
Computing the vertex integrity of $G$ alongside $S$ is fixed-parameter tractable in $\vi(G)$~\cite{DDvH.CCV.2016}.
To define the neighborhood diversity of $G$, we must first introduce the concept of neighborhood types.
Two vertices $u,v\in V(G)$ have the same \emph{neighborhood type}, indicated as $u \sim_N v$, if and only if $N_G(u) \setminus \{v\} = N_G(v) \setminus \{u\}$.
Observe that $\sim_N$ defines an equivalence relation on $V(G)$.
The \emph{neighborhood diversity} $\nd(G)$ of $G$ is the smallest integer $k$ such that there exists a partition of $V(G)$ into $k$ sets $V_1, V_2, \ldots, V_{k}$ and $u \sim_N v$ holds for all $u,v \in V_i$ and all $i \in [k]$.
Note that this implies that each $V_i$, $i \in [k]$, is either a clique or an independent set.
It is known that we can compute $\nd(G)$ and the partition of $V(G)$ into neighborhood types in polynomial time~\cite{Lam.AMt.2011}.

\subparagraph*{Patterns.}
Our definition of patterns leans on those of earlier work on pattern recognition and detection~\cite{FH.GCF.2021,DFHP.Pdo.2023}.
A \emph{pattern} $P = \PatternLong$ is a four-tuple where $\prec_P$ is a total order of $V(P)$, $E^+(P), E^-(P) \subseteq \binom{V(P)}{2}$, and $E^+(P) \cap E^-(P) = \emptyset$.
We call $V(P)$ the vertices, $E^+(P)$ the \emph{forced edges}, and $E^-(P)$ the \emph{forbidden edges} of the pattern.
For a pattern $P = \PatternLong$, we define $n_P \coloneqq \Size{V(P)}$, $m_P^+ \coloneqq \Size{E^+(P)}$, $m^-_P \coloneqq \Size{E^-(P)}$, $E(P) \coloneqq E^+(P) \cup E^-(P)$, and $m_P \coloneqq \Size{E(P)} = m^+_P + m^-_P$.
Let $H$ be a graph on $n_P$ vertices and $\prec_H$ be a linear order of $V(H)$.
We say that $(H, \prec_H)$ \emph{matches} the pattern $P$ if there exists a bijection $\delta\colon V(H) \to V(P)$ such that 
\begin{itemize}
	\item $u \prec_H v$ if and only if $\delta(u) \prec_P \delta(v)$ for every $u,v\in V(H)$,
	\item for every $uv \in E^+(P)$ we have $\delta^{-1}(u)\delta^{-1}(v) \in E(H)$, and
	\item for every $uv \in E^-(P)$ we have $\delta^{-1}(u)\delta^{-1}(v) \notin E(H)$.
\end{itemize}
For a graph $G$ (not necessarily on $n_P$ vertices) and a linear order $\prec_G$ of $V(G)$, we say that $(G, \prec_G)$ \emph{contains} the pattern $P$ if there exists a set of vertices $X \subseteq V(G)$ such that $(G[X], \prec_G\mid_X)$ matches $P$.
We call $(\prec_G, X)$ a \emph{realization} of $P$ (for the graph $G$).
Note that for a given pattern $P$, the bijection $\delta$ is implicitly defined via $\prec_G$ and $X$.
Therefore, we omit it in the following unless necessary, in which case we use $\delta{(\prec_G, X)}$, or simply $\delta$ if $(\prec_G, X)$ is clear from the context, to refer to it.
We use $\Image(\delta) = X$ to denote the image of $\delta$.
A linear order $\prec_G$ \emph{avoids} $P$ if $(G, \prec_G)$ does not contain $P$.
We let $\mathcal{C}_P$ denote the class of all graphs~$G$ that admit a linear order $\prec_G$ that avoids $P$.
Note that in contrast to earlier work~\cite{FH.GCF.2021}, we allow $G$ to be disconnected.

We study the parameterized complexity of deciding membership in $\mathcal C_P$, that is:

\probdef{\PA}{A graph $G$ and a pattern $P$.}{Does there exist a total order $\prec_G$ of $V(G)$ that avoids the pattern $P$?}
We call a total order $\prec_G$ satisfying the above condition a \emph{solution}.

\section{Complexity-Theoretic Lower Bounds}
\label{sec:lower-bounds}
Already fairly simple patterns are powerful enough to encode that a graph $G$ is, for example, 3-colorable or has bandwidth at most $k$.
This expressive power comes at a cost and in this section, we discuss complexity-theoretic lower bounds of \PA.
First, we show in \Cref{sec:sigma-2-p-complete} that \PA is $\Sigma_2^\P$-complete.
In \Cref{sec:bandwidth-hardness}, we focus on the parameterized complexity of \PA.
We exploit the known (parameterized) hardness of \probname{Bandwidth} to rule out \NewText{(efficient) parameterized algorithms for many natural parameter combinations.
In particular, we argue that the problem remains \NP-hard even on trees with constant pathwidth and $m_P = 1$.
}

\subsection{\PAHeader is $\boldsymbol{\Sigma_2^\P}$-Complete}
\label{sec:sigma-2-p-complete}
We start with showing that \PA is $\Sigma_2^\P$-complete.
(See Papadimitriou~\cite{DBLP:books/daglib/0072413} for the formal definition of $\Sigma_2^\P$).
Intuitively, while \NP\ problems correspond to an existential quantifier---where a positive instance has a polynomial-size certificate verifiable in polynomial time---$\Sigma_2^\P$ problems extend this by one quantifier alternation: membership can be expressed by an ``$\exists\forall$'' sentence with a polynomial-time predicate, which means checking the validity of a fixed certificate is a \coNP\ problem.
\PA is a natural $\Sigma_2^\P$ problem, since for a fixed vertex ordering, deciding whether the ordering avoids a given pattern is in \coNP, as the existence of a realization of the pattern in that order certifies non-avoidance.

To show hardness, we use the following lemma, which allows us to express membership in a graph class up to $s$ vertex deletions within our pattern-avoidance framework.
In particular, it lets us model the task of deleting at most $s$ vertices so that the remaining graph is $K_t$-free, i.e., has clique number less than~$t$.
This vertex-deletion problem is known as the \textsc{Generalized Node Deletion Problem} (\textsc{GNDP}) and is $\Sigma_2^\P$-complete~\cite{DBLP:conf/stacs/Rutenburg91}.

\begin{lemma}\label{lem:deletion_to_pattern}
    Let $P$ be a pattern and let $P^\star$ be obtained from $P$ by adding $d$ isolated vertices at the beginning of $P$'s total order.
    Then $\mathcal{C}_{P^\star}$ consists of the graphs that are in $\mathcal{C}_P$ after at most $d$ vertex-deletions.
\end{lemma}
\begin{proof}
    $(\Rightarrow):$
    Let $G$ be a graph.
    We first show that $G \in \mathcal{C}_{P^\star}$ implies that there is $X \subseteq V(G)$ of size at most $d$ with $G-X \in \mathcal{C}_{P}$.
    We derive the contrapositive.
    Assume that for every deletion set $X \subseteq V(G)$ with $|X| \leq d$ we have $G-X \not\in \mathcal{C}_P$, i.e., for every total order $\prec_{G-X}$ of $V(G)\setminus X$, $(G-X, \prec_{G-X})$ contains $P$.
    Towards deriving $G \not\in \mathcal{C}_{P^\star}$,
    consider any total order $\prec_G$ of $V(G)$.
    By assumption, there is a realization of $P$ in $(G, \prec_G)$ using none of the first $d$ vertices of $\prec_G$.
    Hence, the first $d$ vertices of $\prec_G$ (acting as the added isolated vertices), together with the realization of $P$ realize $P^\star$ in $(G, \prec_G)$.
    Thus $G \not\in \mathcal{C}_{P^\star}$.

    $(\Leftarrow):$
    Let $G$ be a graph.
    Assume there is $X \subseteq V(G)$ of size at most $d$ with $G-X \in \mathcal{C}_{P}$, that is, there is a total order $\prec_{G-X}$ of $V(G)\setminus X$ that avoids $P$.
    Towards deriving $G \in \mathcal{C}_{P^\star}$,
    let the total order $\prec_G$ of $V(G)$ be obtained by appending the vertices $X$ in any order to the front of $\prec_{G-X}$.
    Towards a contradiction, suppose $(G, \prec_G)$ contains $P^\star$.
    Fix such a realization.
    Suppose the realization maps a vertex $v$ of $V(P)\cap V(P^\star)$ to a vertex in $X$.
    Then, the realization needs to map all $d$ added isolated vertices, i.e., $V(P^\star)\setminus V(P)$ to a spot before $v$ in $\prec_G$.
    But there are at most $d-1$ such spots, yielding a contradiction.
    Hence, restricting the realization of $P^\star$ to subpattern $P$ witnesses that $(G-X, \prec_{G-X})$ contains $P$, contrary to our assumption.
\end{proof}

\begin{theorem}
\label{thm:sigma-completeness}
\PA is $\Sigma_2^\P$-complete.
\end{theorem}
\begin{proof}
\emph{Membership.}
Given an instance $(G,P)$ of \PA, we guess a total order $\prec$ of $V(G)$.
It remains to verify that $(G,\prec)$ avoids~$P$.
For a fixed order $\prec$, non-avoidance of $P$ is in $\NP$, as a realization of $P$ in $(G,\prec)$ is a polynomial-size witness that can be checked in polynomial time. Thus, avoidance of $P$ in $(G, \prec)$ is in \coNP\ and $\PA\in\Sigma_2^\P$.

\emph{Hardness.}
We reduce from the $\Sigma_2^\P$-complete \textsc{Generalized Node Deletion Problem} (\textsc{GNDP})
\cite{DBLP:conf/stacs/Rutenburg91}.
An instance of \textsc{GNDP} consists of a graph $G$ and integers $s,t$, and asks
whether there exists a vertex set $X\subseteq V(G)$ with $|X|\le s$ such that $G-X$ is $K_t$-free.

Given such an instance, we construct in polynomial time an equivalent instance $(G,Q_t^\star)$ of \PA.
Here, $Q_t$ is the pattern obtained from the clique $K_t$ with its vertices ordered arbitrarily.
Observe that $\mathcal{C}_{Q_t}$ is exactly the class of $K_t$-free graphs.
Let $Q_t^\star$ be obtained from $Q_t$ by adding $s$ isolated vertices at the beginning of the order.
By \cref{lem:deletion_to_pattern}, the graph $G$ admits a total order avoiding $Q_t^\star$
if and only if there exists a set $X\subseteq V(G)$ with $|X|\le s$ such that
$G-X\in\mathcal{C}_{Q_t}$, i.e., such that $G-X$ contains no copy of $K_t$.
Hence, $(G,s,t)$ is a positive instance of \textsc{GNDP} if and only if $(G,Q_t^\star)$ is a positive instance
of \PA.
\end{proof}

\subsection{Parameterized Lower Bounds for \PAHeader}
\label{sec:bandwidth-hardness}

Our main ingredient is the pattern from \Cref{fig:pattern-results}b that can be used to represent the problem \probname{Bandwidth} as an instance of \PA.
In \probname{Bandwidth}, we are given a graph $G$ and an integer $k$ and want to know whether there exists a linear order $\prec_G$ such that for every edge $uv \in E(G)$ (with $u \prec_G v$) there are at most $k$ vertices $w$ with $u \prec_G w \preceq_G v$; i.e., every edge spans at most $k - 1$ vertices~\cite{GareyJohnson}.
Next, we \NewText{transfer} known lower bounds with respect to the parameterized complexity of \probname{Bandwidth} to our problem.
To this end, observe that the pattern has $k+2$ vertices and one edge, i.e., $n_P = k+2$ and $m_P = m_P^+ = 1$.

\probname{Bandwidth} is known to be \XNLP-complete (and thus \W[$t$]-hard for every $t \geq 1$) parameterized by the bandwidth even on caterpillar graphs of hairlength $3$~\cite{BGNS.PPC.2022},%
and hence in particular on \NewText{trees} with constant pathwidth and feedback edge number. Moreover, it is \W[1]-hard parameterized by the treedepth of $G$~\cite{GHK+.Egt.2022}; see also the overview of the parameterized complexity of \probname{Bandwidth} by Gima et al.~\cite{GKK+.BPC.2025}. %
We summarize the above discussion in the following theorem.
\begin{theorem}
	\label{thm:bandwidth-hardness}
	\PA\ is (a) \XNLP-complete w.r.t.\ $n_P$ even on trees with constant pathwidth, and (b) \W\textup{[1]}-hard w.r.t.\ the treedepth of the input graph.
	\NewText{Both results even apply} %
	to instances with $m_P=1$.
\end{theorem}

\section{Technical Overview of Algorithmic Results}
\label{sec:overview}
Due to space constraints, we can only provide in the main body a technical overview of the core ideas behind our algorithmic results.
Therefore, we defer the full details to \cref{sec:vi,sec:nd-plus-mp,sec:poly-trees}.
In \Cref{sec:to:vi}, we give an outline of the Ramsey Pruning technique we use for our \FPT\ result parameterized by $\vi(G) + n_P$ (\Cref{thm:vifpt}), which we describe in full detail in \Cref{sec:vi}.
In \Cref{sec:to:nd-plus-mp}, we sketch the integer linear programming approach that is based on several insights into pattern compression and solution normalization and yields our  \FPT\ algorithm parameterized by $\nd(G) + m_P$ (\Cref{thm:fpt-nd-mp}), which is laid out in full formality in \Cref{sec:nd-plus-mp}.
Finally, we present in \Cref{sec:to:poly-trees} our dynamic programming approach for the \XP-algorithm on forests parameterized by $n_P$ (\Cref{thm:non-mixed}) and outline the one class of patterns for which this approach fails (\Cref{thm:one-sided-hard}).
Full details on both aspects are given in \Cref{sec:poly-trees}.
Short concluding remarks and open questions are given at the end of this paper in \Cref{sec:conclusion}.

\subsection{A Fixed-Parameter Algorithm w.r.t.\ the Vertex Integrity}
\label{sec:to:vi}
Our first algorithmic result concerns a parameterization using the structural vertex integrity parameter together with the number of vertices in the pattern.

\begin{restatable}\restateref{thm:vifpt}{theorem}{thmvifpt}
\label{thm:vifpt}
\PA\ is fixed-parameter tractable when parameterized by $\vi(G)+n_P$.
\end{restatable}

To establish Theorem~\ref{thm:vifpt}, we extend the recently introduced \emph{Ramsey Pruning} technique, which was previously employed to establish the fixed-parameter tractability of \textsc{Twin-Width}~\cite{GanianRocton26}, \textsc{Queue Number}, and \textsc{Stack Number}~\cite{DFGS.LLR.2025} w.r.t.\ $\vi(G)$. We note that in the previous works, it was not necessary to combine a bound on the vertex integrity with problem-specific parameters since the vertex integrity already suffices to upper-bound the twin-width, queue number, and stack number of a graph---however, the same argument cannot be applied for \PA. Case in point: the parameterized complexity of \textsc{Bandwidth} (a special case of \PA\ with $m_P=1$)  parameterized by $\vi(G)$ is an open problem~\cite{GKK+.BPC.2025}. Before we can discuss how the technique needs to be adapted to cover \PA, we first provide a general overview of Ramsey Pruning.

A commonly used approach to establish fixed-parameter tractability for problems w.r.t.\ $\vi(G)$ is to employ kernelization, that is, show that the input graph $G$ can be reduced to have size upper-bounded by a function of the parameter, via gradually removing redundant subgraphs without impacting the (non-)existence of a solution. Typically, such direct kernelization arguments rely on proving that removing one out of sufficiently many ``identical'' copies of a subgraph $H$ is safe in the following sense: \begin{enumerate}[(a)]
\item removing $H$ from $G$ preserves the existence of a solution (which is usually trivial), and
\item adding a copy of $H$ back into $G'=G-H$ preserves the existence of a solution on $G'$ (which is usually the less trivial part).
\end{enumerate}
However, for \PA\ (as well as for \textsc{Queue Number} and \textsc{Stack Number}~\cite{DFGS.LLR.2025}) it is entirely non-obvious how one could establish a proof for point (b) directly: a hypothetical solution on $G'$ can be chosen in a way which makes reinserting $H$ into it impossible, seemingly necessitating the recomputation of an entirely different solution on $G'$.

\begin{figure}
    \centering
    \includegraphics{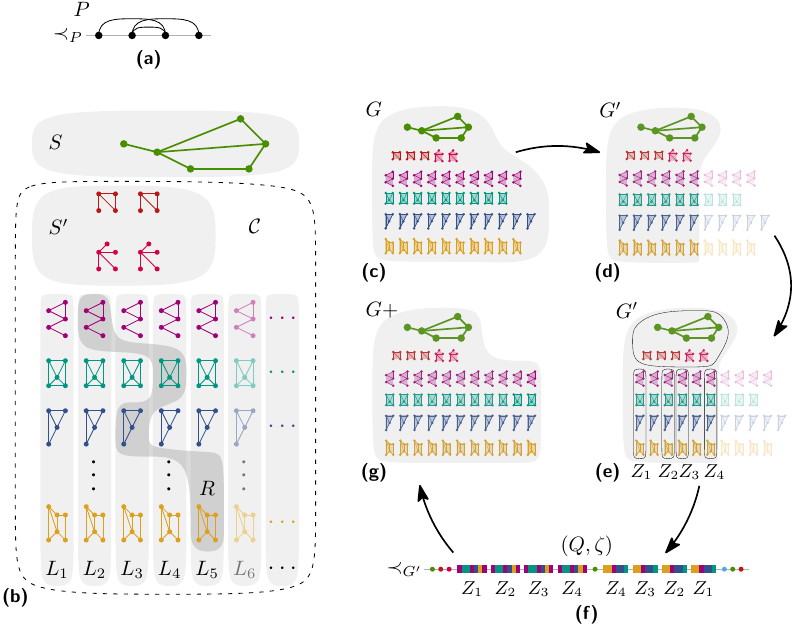}
    \caption{Schematic visualization of the Ramsey pruning technique used to establish \Cref{thm:vifpt}. \textbf{\textsf{(a)}} Shows a forbidden pattern $P$ with $n_P = 4$ and \textbf{\textsf{(b)}} illustrates the used notation.
    \textbf{\textsf{(c)}}--\textbf{\textsf{(g)}} visualize the proof steps, where $G$ is the input graph, $G'$ the reduced graph after applying \Cref{lem:reducedcomputation}.
    In \textbf{\textsf{(e)}} and \textbf{\textsf{(f)}}, $Z_1$ to $Z_4$ are the large groups whose layout we extend to a solution for $G^+$, a supergraph of $G$, and, consequently, also to a solution for $G$.    
    }
    \label{fig:ramsey-overview}
\end{figure}

The core idea underlying Ramsey Pruning is that instead of directly attacking point (b), we use Ramsey-type arguments to argue that every hypothetical solution for $G'$ must contain a carefully selected subsolution which can guide the reconstruction of a ``well-structured'' solution, one whose properties make it amenable to reinserting an arbitrary number of copies of $H$. Thus, while we do not directly prove the safeness of performing any single deletion operation, the proof technique guarantees that the existence of a solution is preserved between the first and final graph in the sequence of reduction rules. For a schematic illustration of the technique (albeit with some additional technical notation), see also \Cref{fig:ramsey-overview}.

Let $S$ be a deletion set witnessing a bound on the vertex integrity of $G$ (cf. Section~\ref{sec:preliminaries}), computed via known algorithms~\cite{DDvH.CCV.2016}.
The initial steps for applying Ramsey Pruning on \PA\ closely follow the previous works~\cite{DFGS.LLR.2025,GanianRocton26}, including the iterative extension of the deletion set $S$ to $S\cup S'$ in order to guarantee that all connected components in $G-(S\cup S')$ occur sufficiently frequently~(\Cref{lem:reducedcomputation}). The main complication occurs when attempting to prove the actual safeness of removing copies of $H$, which is the problem-specific part that is detailed in Subsection~\ref{subsec:VI-Kernel-proof}. Unlike in previous applications of the technique, for \PA\ it is not sufficient to have a constant-sized subsolution to guide the reconstruction of $G$---the generality of \PA\ seems to necessitate that we have a subsolution whose size also depends on a function of $n_P$. Overcoming this obstacle requires the use of a stronger variant of Ramsey's Theorem in our proof and a new line of arguments that exclude the creation of forbidden patterns if we start from a guiding subsolution that is not only well-structured, but also satisfies additional size lower bounds.

\subsection{Parameterization by the Neighborhood Diversity + $\boldsymbol{m_P}$}\label{sec:overview-nd}
\label{sec:to:nd-plus-mp}
Our second main result (see \cref{sec:nd-plus-mp}) establishes the fixed-parameter tractability of \PA when parameterized by the neighborhood diversity of the graph $G$ and the number of edges of the pattern $P$. Recall that with this parameterization, we consider a partition of the vertices of $G$ into a bounded number $\nd(G)$ of neighborhood types, i.e., equivalence classes of vertices w.r.t.\ their neighborhood in $G$ (where, within each class the induced subgraph is either a clique or an independent set). 

\begin{restatable}\restateref{thm:fpt-nd-mp}{theorem}{thmfptndmp}
\label{thm:fpt-nd-mp}
    \PA is fixed-parameter tractable when parameterized by $\nd(G) + m_P$.
\end{restatable}

Our algorithm exploits that the vertices of each neighborhood type are somewhat interchangeable in terms of representing forced and forbidden edges of the pattern, and, additionally, that we can restrict the relevant complexity of the pattern to its set of edges and their incident vertices. 

We proceed in two main steps to prove \cref{thm:fpt-nd-mp}. 
On a high level, in the first step we show that it is sufficient to focus on well-structured solutions of \PA, which bound the number of alternations between maximal consecutive subsequences of vertices of the same neighborhood type (called \emph{segments}) by a function of $\nd(G)$ and $m_P$. 
Once this is established, the second step uses integer linear programming (ILP) to distribute the vertices of all neighborhood types suitably among their respective maximal consecutive subsequences in such a way that the pattern is avoided. 
Since the number of variables of the ILP remains bounded by a function in our parameter, it can be solved via a fixed-parameter algorithm using the classical algorithm of Lenstra~\cite{Len.IPF.1983} and its subsequent improvements~\cite{Kan.MCB.1987,DBLP:journals/siamcomp/HeathR92}. 
Since the number of possible assignments of segments to neighborhood types is upper-bounded by a function of our parameters, we can afford to branch over all such assignments. For each, we then construct and solve an ILP instance that determines whether it is possible to assign the neighborhood types into the segments in a way which avoids the target pattern $P$.

\begin{figure}
	\centering
	\includegraphics[page=1]{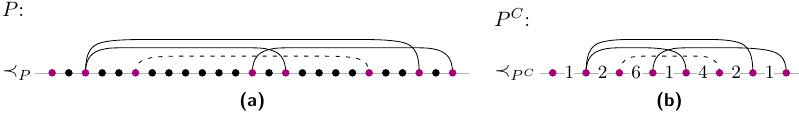}
	\caption{\textbf{\textsf{(a)}} The pattern $P$ and \textbf{\textsf{(b)}} the corresponding compressed pattern $P^C$. Critical vertices are colored purple.}
	\label{fig:compressed-pattern}
\end{figure}

\subparagraph*{Step 1.} We observe that the original pattern $P$ can be \emph{compressed} to a pattern with just $\BigO{m_P}$ \emph{critical} vertices by keeping only the first and last vertices of $P$ as well as all the endpoints of $E(P)$. Note that $P$ is fully characterized by these critical vertices along with the number of non-critical vertices occurring between each consecutive pair of critical vertices; see \Cref{fig:compressed-pattern} for an illustration.
Moreover, to determine whether a subsequence in a total order forms the pattern $P$, it is only necessary to check the neighborhood types of the critical vertices of $P$---the neighborhood types of non-critical vertices are entirely irrelevant.
We now
enumerate all $\nd(G)^{\BigO{m_P}}$ possible assignments of neighborhood types to critical vertices of $P$, and further restrict our attention to so-called \emph{provoking} assignments which would realize $P$ if the number of non-critical vertices in the respective intervals is sufficiently large.

Next, we decompose a hypothetical solution $\prec_G$ into a bounded number of \emph{safe intervals} by introducing and superimposing $\BigO{m_P}$ cuts for each provoking assignment (Lemma~\ref{lem:nd-plus-mp-number-intervals-number}). We show that these safe intervals have the useful property that we can re-order the vertices within them arbitrarily while retaining an equivalent solution to $\prec_G$; this lets us put these vertices into a normalized order sorted by neighborhood types. As an immediate consequence, we obtain that the number of neighborhood-type alternations within each safe interval can be restricted to at most $\nd(G)$, which, together with the bound on the number of safe intervals, yields the desired result of Step 1: every positive instance of \PA admits a solution with $f(m_P,\nd(G))$ neighborhood-type alternations (formalized in \Cref{cor:nd-plus-mp-alternations-bound}). 

\subparagraph{Step 2.} The second step exploits the property established above by first branching over all possibilities of assigning the at most $f(m_P,\nd(G))$ non-empty segments to the $\nd(G)$ different neighborhood types. 
The number of these branches is still bounded by a function in $m_P$ and $\nd(G)$.
It remains to determine whether it is possible to assign vertices to segments in a way which prevents 
any provoking assignment from realizing $P$.
We solve this last aspect in Subsection~\ref{sec:nd-plus-mp-ilp} by constructing an ILP whose constraints ensure that 
\begin{enumerate}[(i)]
\item the cardinality of all segments of the same neighborhood type add up to the cardinality of that neighborhood type in $G$, and 
\item for any possible provoking assignment $\beta$ of the critical vertices to these segments, at least one pair of critical vertices is ``too close'' to each other to realize the pattern $P$, i.e., there are too few vertices between them to be consistent with $P$.
\end{enumerate}
The final difficulty lies in proving that property (ii) is both a necessary and sufficient condition for avoiding $P$. There, we employ an inductive 
construction strategy that finds an occurrence of $P$ given any provoking assignment $\beta$ where no pair of critical vertices is ``too close''.

\subsection{An Algorithm for Almost All Fixed Patterns on Forests}
\label{sub:forests}\label{sec:to:poly-trees}
As our third and final main result (see \cref{sec:poly-trees}), we provide an \XP\ algorithm for \PA\ on forests when parameterized by $n_P$, for \emph{almost all} possible patterns $P$. To formalize our result, let a pattern be \emph{mixed} if, intuitively, 
it (1) contains a single forced edge $e$, and furthermore (2) each of the endpoints of $e$ lies between the endpoints of some forbidden edge; see also \Cref{fig:mixed-pattern}(b).\footnote{A precise definition of the notion, and in particular of what ``between'' means here, is provided in \Cref{def:mixed-pattern}.}

\begin{figure}
	\begin{center}
		\includegraphics[page=1]{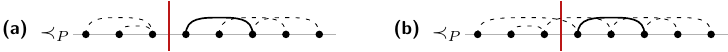}
	\end{center}
	\caption{A \textbf{\textsf{(a)}} left-separated and \textbf{\textsf{(b)}} mixed pattern. The forced edge is solid, forbidden edges are dashed. Note that in (a) no forbidden edge crossed the red line.}
	\label{fig:mixed-pattern}
\end{figure}

\begin{restatable}\restateref{thm:non-mixed}{theorem}{thmforest}
\label{thm:non-mixed}
For every constant-size non-mixed pattern $P$, \PA\ is polynomial-time solvable on forests.
\end{restatable}

\begin{figure}
	\begin{center}
		\includegraphics[page=1]{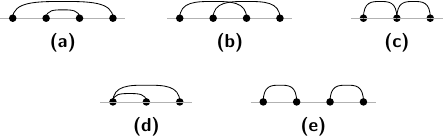}
	\end{center}
	\caption{All patterns (up to reversal of the total order) with exactly two forced edges, no forbidden edges, and no isolated vertices.}
	\label{fig:patterns-two-forced-edges}
\end{figure}

\subparagraph*{Easy Cases.}
We first show that patterns with zero or at least two forced edges admit a simple classification. 
If $P$ contains at least two forced edges, then we consider two of them and all possible ``subpatterns'' that they can form; see \Cref{fig:patterns-two-forced-edges}.
A pattern-avoiding total order exists for each subpattern. 
If $P$ contains no
forced edges, then either the input graph $G$ has size bounded by a function of $|V(P)|$
— in which case the problem can be solved by brute force — or no pattern-avoiding order
exists. Consequently, the only nontrivial case is when $P$ has exactly one forced edge.

\subparagraph*{Left-Separated Patterns.}
For constant-size patterns with a single forced edge, we further distinguish whether $P$ is
\emph{left-separated} (see also \Cref{fig:mixed-pattern}(a)). Up to symmetry, left-separated patterns precisely correspond to the remaining non-mixed cases: they are patterns where the left endpoint of the forced edge separates everything to its left from the right side. We hence proceed with our arguments for the left-separated case.

\subparagraph*{Ignoring the Left Side.}
Our first preprocessing step shows that, for left-separated patterns, we may branch on a
constant number of choices such that everything strictly to the left of the forced edge
can be safely ignored. This reduces the problem to reasoning about the structure on the
right side of the forced edge in $P$, up to a bounded amount of nondeterminism.

\subparagraph*{Hydra Vertices.}
We then turn to the degree structure of the forest. A key insight is that vertices with many
non-leaf neighbors impose strong constraints on any pattern-avoiding order. We call such
vertices \emph{hydra vertices} (see Definition~\ref{def:poly-trees-preprocessing-hydra-vertex}). We prove that, when their own leaves
are ignored, for every hydra vertex there is only a linear number (in $n_P$) of vertices right to it, in any valid total order. This further implies %
that the total number of hydra vertices is bounded by a function of $n_P$ (i.e., a constant).

This allows us to guess a bounded set of vertices at the right end of the order that is
guaranteed to contain all hydra vertices. We refer to this guessed region as the
\emph{fixed side}. Combined with a leaf-compression step, this yields a partial order in
which the rightmost segment is fixed up to permutation, while the remainder of the graph
is significantly simplified.

\subparagraph*{Further Preprocessing.}
With the fixed side in place, we apply additional reductions. First, we remove and solve separately all connected
components that do not intersect the fixed side, as they can be ordered independently.
Second, we show that almost all neighbors of hydra vertices can be removed without affecting
feasibility. After these steps, we obtain an equivalent instance in which every vertex has
bounded ``interesting'' degree, that is, bounded degree after ignoring leaves.

\begin{figure}
	\begin{center}
		\includegraphics[page=1]{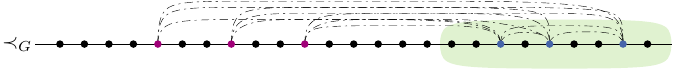}
	\end{center}
	\caption{We ensure that for every set $A$ (purple) of $n_P$ free (i.e., non-fixed) vertices, there is an interdependent set $I_A$ (blue) of $n_P$ vertices in the green fixed side disconnected from $A$.}
	\label{fig:fixed-side}
\end{figure}

\subparagraph*{Expanding the Fixed Side.}
At this point, bounded degree alone is still insufficient for a direct dynamic programming
approach. We therefore further expand the fixed side to ensure a crucial independence
property: for every set of $n_P$ vertices in the non-fixed part, there exists a sufficiently
large independent set in the fixed side that is non-adjacent to all of them; see also \Cref{fig:fixed-side}. This property
guarantees that the behavior of the forbidden pattern can be analyzed locally around each
edge in the non-fixed part.

\subparagraph*{Dynamic Programming.}
With all preprocessing completed, we can apply 
careful dynamic programming to solve the instance. This in particular relies on gradually constructing a total order while keeping track of only a bounded number of vertices surrounding our current position in the order---i.e., we employ a right-to-left ``sliding window'' approach. The dynamic programming procedure maintains information about the vertices to the left and to the right of the sliding window. 
Finally, all of the pruning steps and structural insights come together:
\begin{enumerate}
\item In order to achieve the desired running time bounds, we show that there are only a limited number of ways the vertices outside of the current constant-size sliding window can be partitioned between the ``right'' processed part of the instance, and the ``left'' yet-to-be-processed part. 
\item At the same time, we show that the information we have about the fixed side and the current sliding window guarantees that we will be able to detect a forbidden pattern in the graph as soon as it occurs, even if part of it lies outside of the window. 
\end{enumerate}

\subparagraph*{The Hard Customers: Mixed Patterns.}
The preprocessing steps above intrinsically rely on $P$ being left-separated. For mixed patterns, it is no longer true that hydra vertices must all be on the right---some might be placed on the left instead. At first glance, this complication would seem to admit a simple solution: why not simply extend the notion of having a fixed side to both the leftmost and rightmost ends of our hypothetical total order? 

Unfortunately, in \Cref{thm:one-sided-hard} we show that such a line of attack is doomed to fail: \PA\ turns out to be \NP-hard on forests if the first and last vertex of the total order is fixed on the input. We formalize this below.
\enlargethispage{.5em}

\probdef{\PAOneFTwoSLongUnderline~(\PAOneFTwoS)}{A graph $G$, a pattern $P$, two sets $U_L, U_R \subseteq V(G)$, $U_L \cap U_R = \emptyset$, and two total orders $\prec_{U_L}$ on $U_L$ and $\prec_{U_R}$ on $U_R$.}{Does there exist a total order $\prec_G$ of $V(G)$ that starts with $\prec_{U_L}$, ends with $\prec_{U_R}$, and avoids the pattern $P$?}

\begin{restatable}\restateref{thm:one-sided-hard}{theorem}{thmonesidedhard}
\label{thm:one-sided-hard}
    \PAOneFTwoSLong is \NP-hard, even if $G$ is a forest, $n_P=4$, $m_P = m_P^+ = 1$, and $\Size{U_L} = \Size{U_R} = 1$.
\end{restatable}

\begin{proof}[Proof Sketch.]
We reduce from the strongly \NP-hard \textsc{Bin Packing} problem \cite{GareyJohnson}, which asks whether, given a finite set $U$ of items with sizes $s(u) \in \mathbb{Z}^+$ for each $u \in U$, a bin capacity $B \in \mathbb{Z}^+$, and an integer $K > 0$, the set $U$ can be partitioned into disjoint subsets $U_1, U_2, \ldots, U_K$ such that $\sum_{u \in U_i} s(u) \le B$ for all $i \in [K]$.

Fix an instance of \textsc{Bin Packing}.
We construct an instance of \PAOneFTwoSLong 
with input graph $G$,
pattern $P$,
first vertex $U_L = \{s\}$ (with the trivial total order $\prec_{U_L}$),
and final vertex $U_R = \{t\}$ (with the trivial total order $\prec_{U_R}$),
as follows.
Let the forbidden pattern $P \coloneqq ( \{a,b,c,d\}, a \prec b \prec c \prec d, \{ad\}, \{\})$, i.e., $P$ is the pattern characterizing the graphs of bandwidth at most two displayed in \cref{fig:one-sided-hardness}a (in \Cref{sec:poly-trees-hardness}).
To construct the input graph $G$, 
we begin with a path of length $B \cdot K + 2\cdot(K-1)$, called the \emph{spine}, with first vertex $s$ and last vertex $t$.
For $1 \leq i \leq K-1$, call the $i \cdot(B+2)$'th vertex of the path the \emph{separator vertex $s_i$}.
Attach two pendant vertices to each separator vertex.
We call the $K$ subpaths between consecutive vertices in the sequence $s, s_1, s_2, \ldots, s_{K-1}, t$ (including both endpoints of each subpath) the \emph{bin paths}. Here, we call the subpath from $s$ to $s_1$ the first bin path, and number the remaining bin paths in sequence.
Finally, we construct the \emph{item gadgets}. For each item $u \in U$, add a disjoint path of length $s(u)-1$ (i.e., a path with $s(u)$ vertices) to $G$.
See \cref{fig:one-sided-hardness} (located in \Cref{sec:poly-trees-hardness}) for an example of the reduction.
As \textsc{Bin Packing} is strongly \NP-hard, the construction is feasible in polynomial time. To establish \cref{thm:one-sided-hard}, it remains to show that the reduction is correct (cf.~\Cref{lem:one-sided-hardness-correctness}).
\end{proof}

As an immediate implication, we obtain:

\begin{restatable}\restateref{cor:bandwidth-hard}{corollary}{corbandwidthhard}
    \label{cor:bandwidth-hard}
    Deciding if a forest admits a linear order of bandwidth two with prescribed first and last vertex is \NP-hard.
\end{restatable}

Moreover, in \Cref{thm:one-sided-hard-on-trees} we show that the lower bound also translates to trees, albeit at the cost of a slightly more complex pattern and requiring two fixed vertices on the left side. 
These lower bounds are highly surprising, as testing for constant bandwidth is known to be polynomial-time solvable on trees (and, equivalently, forests)~\cite{Saxe80}. Essentially, while \Cref{thm:one-sided-hard} does not yet directly rule out polynomial-time tractability for fixed mixed patterns, it implies that ``standard'' dynamic-programming style approaches to the problem cannot lead to a solution; in particular, the separation steps described in the \textbf{Further Preprocessing} paragraph fail and thus any conventional dynamic program would need to account for an exponential number of ways of partitioning vertices between the processed and yet-to-be-processed parts.

\NewText{
\subparagraph*{Outlook.}
We leave the above-stated question as well as the complexity of \PA when parameterized by $\vi(G) + m_P$ or the treedepth of $G$ plus $n_P$ open for future work and discuss ways to tackle it in the conclusion; see  \Cref{sec:conclusion}.

This concludes the technical overview of our algorithmic results.
We proceed by providing the technical details of all results discussed in the overview.

}

\section{A Fixed-Parameter Algorithm w.r.t.\ the Vertex Integrity}
\label{sec:vi}
We now give the full details on our first algorithmic result.
We assume throughout this section that $n_P\geq 3$; this is without loss of generality, as the case of $n_P\leq 2$ is trivially solvable in polynomial time.
It will also be useful to recall the Knuth notation at this point (cf. Section~\ref{sec:preliminaries}).

\thmvifpt*

We establish Theorem~\ref{thm:vifpt} by proving the following lemma, which directly implies the claimed result. 
Let $p \coloneqq \vi(G)$.
\begin{lemma}
	\label{lem:kernel}
	There is an 
	$\mathcal{O}(n_P^{n_P+1} \cdot n)$ 
	time algorithm that takes an instance $(G,P)$ of \PA\ and outputs a subgraph $G'$ of $G$ (a \emph{kernel}) of size at most 
	$\big(2\uparrow\uparrow (p\cdot 2^{2p^2})\cdot 2 + 6\big)^{n_P \cdot p}$ 
	 with the following property: $(G',P)$ is a YES-instance of \PA\ if and only if so is $(G,P)$.   
Moreover, a solution for $(G,P)$ can be computed from a solution for $(G',P)$ in polynomial time.
\end{lemma}

The remainder of this section is hence devoted to proving Lemma~\ref{lem:kernel}. Our proof employs the recently developed \emph{Ramsey Pruning} technique~\cite[Section~3]{DFGS.LLR.2025}, and the structuring of this section closely follows the presentation in that paper. Apart from problem-specific deviations, the fundamental differences occur in Lemmas~\ref{lem:fewlarge} and~\ref{lem:structure} (i.e., in the very core of the proof), which require us to establish stronger guarantees on the kernel.

For the following, let us fix an instance $(G,P)$ and a deletion set $S \subseteq V(G)$ such that each connected component $C$ of $G - S$ satisfies $\Size{V(C)} + \Size{S} \leq \vi(G) = p$; recall from Section~\ref{sec:preliminaries} that such a set can be computed in fixed-parameter time. Further let $\mathcal{C}$ be the set of connected components of $G-S$. 
First, we define a notion of ``component-types'' which groups components in $\mathcal{C}$ that exhibit the same outside connections and internal structure.
\begin{definition} We say two graphs $H_0, H_1 \in \mathcal{C}$ are \emph{twins}, denoted $H_0 \sim H_1$, if there exists a canonical isomorphism $\alpha$ from $H_0$ to $H_1$ such that for each vertex $u \in V(H_0)$ and each $v \in S$, $uv \in E(G)$ if and only if $\alpha(u)v \in E(G)$. 
\end{definition}

\begin{lemma}
\label{obs:size-equiv-class}
Each graph $H\in \mathcal{C}$ has at most $p$ vertices, $\sim$ is an equivalence relation and the number of equivalence classes in $[\sim]$ is upper-bounded by $\frac{p \cdot 2^{2p^2}}{p!}$. Moreover, a partition of connected components into $[\sim]$ can be computed in time at most $\mathcal{O}(\frac{p \cdot 2^{2p^2}}{p!}\cdot n)$.
\end{lemma}
\begin{proof}
By definition of $\mathcal{C}$, each $H \in \mathcal{C}$ has at most $p$ vertices. Since the number of non-isomorphic graphs on exactly $p$ vertices is upper-bounded by $2^{p^2}$, we can upper-bound the total number of possible non-isomorphic graphs $H$ by $\frac{p \cdot 2^{p^2}}{p!}$.
Since $|S| \leq p$, there are at most $p^2$ possible edges between $S$ and each $H \in \mathcal{C}$. Hence, we have $[\sim] \leq \frac{p \cdot 2^{2p^2}}{p!}$. For the running time, it suffices to process the connected components of $G-S$ in an arbitrary order and use exhaustive branching over $\alpha$ to determine (in time at most $\mathcal{O}(\frac{p \cdot 2^{2p^2}}{p!})$) which of the equivalence classes in $[\sim]$ it belongs to.
\end{proof}

We now introduce the notion of large equivalence classes based on their size. We then use this to define what we call a large group of vertices---one that contains exactly one representative
from each large equivalence class. Our kernel will later keep a bounded number of these large groups.
\begin{definition}
\label{defn:large-eq-class}
Let $k$ be a positive integer, an equivalence class $[H]$ of $\sim$ is said to be {\em $k$-large} if $|[H]|\geq k$. Further a vertex set $L \subseteq V(G)$ is called a {\em $k$-large group} if the induced subgraph $G[L]$ is a disjoint union of exactly one graph from each $k$-large equivalence class of~$\sim$.
\end{definition}

Next we define a special induced subgraph that will serve as our kernel. The definition is based on carefully choosing a $k$ that is bounded by a computable function of $p$ and $n_P$ so that keeping only $k$ many $k$-large groups along with the small parts of the graph suffices to capture the necessary structure. 
Towards this, let us first fix 
$f(n_P,p,x):=2^{3n_P\cdot ((x+2\cdot p^2\cdot 2^{2p^2})!)^2}+2$ 
to be a computable function that is large enough to apply our Ramsey-type arguments later on; here, $x$ will be an integer that represents the size of a deletion set (initially $S$, but this will be updated iteratively in the proof of Lemma~\ref{lem:reducedcomputation}). Moreover, let $g(n_P,p)$ be a computable function of $n_P$ and $p$ that will upper-bound our nested application of the function $f$ in that same proof; to provide a concrete bound, we set $g(n_P, p):= \big(2\uparrow\uparrow (p\cdot 2^{2p^2})\cdot 2 + 4\big)^{n_P \cdot p}$.

\begin{definition}
\label{def:red-graph}
An induced subgraph $G'$ of $G$ is said to be a {\em reduced graph} of $G$ if there exists a positive integer $x\leq g(n_P,p)$ and a partition of $V(G') = S \uplus S' \uplus Y$ satisfying:
\begin{enumerate}
	\item $|S\cup S'|= x$ and \( S' \) is the set of all vertices in graphs in \( \mathcal{C} \) that do not belong to an $f(n_P,p,x)$-large equivalence class of $\sim$.
	\item If there are no $f(n_P,p,x)$-large equivalence classes of $\sim$, $Y=\emptyset$ and $V(G')=S\uplus S' = V(G)$. Otherwise, it holds that $Y=L_1 \uplus \cdots \uplus L_{f(n_P,p,x)}$ where \( L_i \) 
	is an $f(n_P,p,x)$-large group for each $i\in\{1,\cdots,f(n_P,p,x)\}$, and each pair of $L_i$ and $L_j$, $i \neq j$, is vertex-disjoint.
\end{enumerate}
\end{definition}

Intuitively, the reduced graphs we will be dealing with will consist of $S$, a set $S'$ of all equivalence classes of $\sim$ which are too small to fully saturate our large groups (these will later be treated essentially in the same way as $S$), and a sufficient number of large groups; equivalence classes of size larger than $f(n_P,p,x)$ are ``pruned'' to have size exactly $f(n_P,p,x)$. A schematic overview of this intuition can be found in 
Figures~\ref{fig:ramsey-overview}c and~d later on.

The core of our result is the following lemma, which we prove separately in \Cref{subsec:VI-Kernel-proof}.
\begin{lemma}
\label{lem:VI-kernel-proof}
If $G'$ is a reduced graph of $G$ then $(G,P)$ admits a solution if and only if so does $(G',P)$. Moreover, a solution for the former can be constructed from a solution for the latter in polynomial time.
\end{lemma}

Before proceeding to that proof, we show how to construct a reduced graph having size bounded by a computable function of $p$ and $n_P$ in polynomial time. 

\begin{lemma}
\label{lem:reducedcomputation}
There exists a reduced graph $G'$ of $G$, and given $S$ and $\sim$ such a graph can be computed in polynomial time. Further, the number of vertices in $G'$ is upper-bounded by $\big(2\uparrow\uparrow (p\cdot 2^{2p^2})\cdot 2 + 6\big)^{n_P \cdot p}$.
\end{lemma}
\begin{proof}
We present an algorithm to compute a reduced graph $G'$ below:
\begin{enumerate}
	\item Initialize $X=S$ and $\mathcal{C}':=\mathcal{C}$
	\item As long as there exists an equivalence class $[H]$, $H\in \mathcal{C}'$ that is not $f(n_P,p,|X|)$-large:
	\begin{enumerate}
		\item Set $\mathcal{C}'=\mathcal{C}'\setminus [H]$ and $X=X\cup \{v:v\in H', H'\in [H]\}$
	\end{enumerate}
	\item If $\mathcal{C}'\neq \emptyset$, then set $O:=X\uplus L_1\uplus \cdots \uplus L_{f(n_P,p,|X|)}$ where each $L_{1\leq i\leq f(n_P,p,|X|)}$ is a $f(n_P,p,|X|)$-large group. Else set $O:=X$.
	\item Output $G'=G[O]$
\end{enumerate}

We first bound the size of $V(G')=O$. In each iteration of step~$2$, the vertices of all graphs in exactly one equivalence class $[H]$, $H\in \mathcal{C'}$ are added to the set $X$ and the graphs in $[H]$ are removed from $\mathcal{C'}$. Let $X_i$, $i\in \{1,\cdots,2^{2p^2}\}$ be the set $X$ at the end of the $i^{th}$ iteration of step $2$ and let there be $t$ iterations of step $2$. Note that $t=0$ if there is no successful iteration of step 2. Since there are at most $\frac{p\cdot 2^{2p^2}}{p!}$ equivalence classes in $\sim$, the algorithm repeats step~2 at most $\frac{p \cdot 2^{2p^2}}{p!}$ times, thus $t\leq \frac{p \cdot 2^{2p^2}}{p!}$.

Let $X_0:=S$, observe that $|X_0|=|S|\leq p$. If $t\geq 1$, by construction and by the fact that each graph in $\mathcal{C}$ has at most $p$ vertices, for each $i\in \{1,\cdots,t\}$, $|X_i|\leq |X_{i-1}|+p\cdot f(n_P,p,|X_i|)$. For ease of analysis, let us note that $t\leq p\cdot 2^{2p^2}$, and hence $|X_t|\leq g(n_P,p)\leq \big(2\uparrow\uparrow (p\cdot 2^{2p^2})\cdot 2 + 4\big)^{n_P\cdot p}$.

Further each $f(n_P,p,|X_t|)$-large group has at most $\frac{p \cdot 2^{2p^2}}{p!} \cdot p \leq p\cdot2^{2p^2}$ vertices and thus $f(n_P,p,|X_t|)$ many $f(n_P,p,|X_t|)$-large groups have at most $p\cdot 2^{2p^2}\cdot f(n_P,p,|X_t|)$ vertices. Therefore we have $|O|\leq |X_t|+p\cdot 2^{2p^2}\cdot f(n_P,p,|X_t|)\leq g(n_P,p)+p\cdot 2^{2p^2}\cdot f(n_P,p,g(n_P,p)) \leq \big(2\uparrow\uparrow (p\cdot 2^{2p^2})\cdot 2 + 6\big)^{n_P\cdot p}$.

Next observe that when the algorithm stops all vertices in graphs in $\mathcal{C}$ that do not belong to an $f(n_P,p,|X_t|)$-large equivalence class belongs to $X_t$. Further $S\subseteq X_t$. For the $f(n_P,p,|X_t|)$-large equivalence classes, $f(n_P,p,|X_t|)$ many $f(n_P,p,|X_t|)$-large groups are added to $O$. Thus by construction, the induced graph $G'=G[O]$ output by the algorithm is a reduced subgraph by Definition~\ref{def:red-graph}. This completes the proof.
\end{proof}

Lemma~\ref{lem:reducedcomputation} combined with Lemma~\ref{lem:VI-kernel-proof}
establishes Lemma~\ref{lem:kernel} by constructing a reduced graph which is a kernel of the desired size.
Thus, what remains is to establish Lemma~\ref{lem:VI-kernel-proof}; this is where the core ideas of Ramsey pruning come into play.

\subsection{Proof of Lemma~\ref{lem:VI-kernel-proof}}
\label{subsec:VI-Kernel-proof}
For this subsection, let us fix $G'$ to be a reduced graph of $G$ computed by Lemma~\ref{lem:reducedcomputation}.
Since $G'$ is an induced subgraph of $G$, it is clear that if $(G,P)$ has a solution then so does $(G',P)$. To complete the proof we show that the reverse is true: If there exists a total order $\prec_{G'}$ of $V(G')$ that avoids the pattern $P$, then there also exists a total order $\prec_{G}$ of $V(G)$ that avoids the same pattern.

Recall that $f(n_P,p,x):=2^{3n_P\cdot ((x+2\cdot p^2\cdot 2^{2p^2})!)^2}+2$.
Let $\prec_{G'}$ be a fixed (but hypothetical) solution for $(G',P)$. %
Throughout this section, for any subgraph $H\subseteq G'$, we denote by $\prec_{G'}[H]$ the total order of $H$ obtained by restricting $\prec_{G'}$ to $H$.
If $V(G')=V(G)$ we are done. So let $V(G') = S \uplus S' \uplus L_1 \uplus \cdots \uplus L_{f(n_P, p, |S\cup S'|)}$ be a partition witnessing that $G'$ is a reduced graph. Here each $L_i$ is a $f(n_P, p, |S\cup S'|)$-large group and $S'$ is the set of all vertices in components that do not belong to an $f(n_P, p, |S\cup S'|)$-large equivalence class in $\sim$. For brevity, we will hereinafter use \emph{large} as shorthand for $f(n_P, p, |S\cup S'|)$-large. 

Let $\mathcal{L}=\{L_1,\cdots,L_{f(n_P,p,|S\cup S'|)}\}$. Our key idea is to identify $n_P$ large groups with a special structure (a ``guiding sublayout'') in the solution $\prec_{G'}$ using Ramsey theory. We will then use this pattern to insert the remaining parts of the graph. Before proceeding, we define some notations to be able to identify these groups, starting with a ``template'' $R$. 

\begin{definition}
Let $k$ be the number of large equivalence classes in $\sim$ and let $R_i\in \mathcal{C}$
be a canonical representative of the $i^{th}$ large  equivalence class $[R_i]$. Let $R:=V(R_1)\uplus \cdots \uplus V(R_k)$.
\end{definition}

Recall that by Definition~\ref{defn:large-eq-class}, if $X$ is a large group then $G[X]$ is a disjoint union of exactly one graph from each large equivalence class of $\sim$. 
We now define a natural isomorphism from $G[X]$ to $G[R]$. This will allow us to map consistently between different large groups via $R$.
\begin{definition}
For a large group $X$ with $G[X]=H_1\uplus \cdots \uplus H_k$, $H_i\in [R_i]$ for each $i\in [k]$, let $\alpha_X$ be an isomorphism from $G[X]$ to $G[R]=R_1\uplus \cdots \uplus R_k$ such that for each $i\in [k]$ and vertex $u\in H_i$, $\alpha_X(u)\in V(R_i)$, and for each $v\in S$, $uv\in E(G)$ if and only if $\alpha_X(u)v\in E(G)$. 
\end{definition}

It will later be used to employ $\alpha_X$ (and, more specifically, its inverse) on an ordered set of vertices; towards this, for a totally ordered set $W=(w_1,\dots,w_\circ)$ of vertices from $R$, we let $\alpha^{-1}_X(W)$ denote $\alpha^{-1}_X(w_1),\dots,\alpha^{-1}_X(w_\circ)$.

For any large group $A$ and for each $u\in R$ we will sometimes use $u_A$ as shorthand for $\alpha_A^{-1}(u)$.
For distinct large groups $X,Y\in \mathcal{L}$, we say $X\prec_{G'} Y$ if in $\prec_{G'}$, the first vertex in $X\cup Y$ is from $X$. We now fix two distinct large groups $L$ and $L'$ in $\mathcal{L}$ such that $L\prec_{G'}L'$.

\begin{definition}
For $X\prec_{G'} Y\in \mathcal{L}\setminus \{L,L'\}$, we define $\phi_{X,Y}:S\cup S'\cup X\cup Y \rightarrow S\cup S'\cup L\cup L'$
\begin{itemize}
	\item $\phi_{X,Y}(s)=s$ for each $s\in S\cup S'$
	\item $\phi_{X,Y}(x)=\alpha^{-1}_{L}(\alpha_X(x))$ for each $x\in X$
	\item $\phi_{X,Y}(y)=\alpha^{-1}_{L'}(\alpha_Y(y))$ for each $y\in Y$
\end{itemize}
Note that $\phi_{X,Y}$ is an isomorphism from $G[S\cup S'\cup X\cup Y]$ to $G[S\cup S'\cup L\cup L']$.
\end{definition}
\begin{definition}
Let $X,Y \in \mathcal{L}\setminus\{L,L'\}$ such that $X\prec_{G'} Y$. 
Define $\textsf{info}_{\prec_{G'}}(X, Y)$ as the total order of $G[S\cup S'\cup L\cup L']$ obtained from $\prec_{G'}[G'[S\cup S'\cup X\cup Y]]$ %
using the isomorphism $\phi_{X,Y}$.%
\end{definition}
We are now ready to show the existence of $n_P$ many distinct large groups $Z_1,\dots,Z_{n_P}$ with a useful consistent pattern between them and $S\cup S'$ based on $\textsf{info}$. Essentially, these groups will form the aforementioned guiding sub-order used to argue the correctness of the pruning step, i.e., establish Lemma~\ref{lem:VI-kernel-proof} by showing that we can reinsert all the removed vertices by building on $\prec_{G'}[G'[S\cup S' \cup Z_1 \cup \dots \cup Z_{n_P}]]$. Note that---perhaps counterintuitively---we will do so by discarding all other information about $\prec_{G'}$. We refer to \Cref{fig:ramsey-overview} for an overview of the entire approach and to \Cref{fig:ramsey-blocks}(a) for a visualization of $\textsf{info}$.
\newcommand{\asc}{\ensuremath{\nearrow}}
\newcommand{\desc}{\ensuremath{\searrow}}
\begin{figure}
    \centering
    \includegraphics{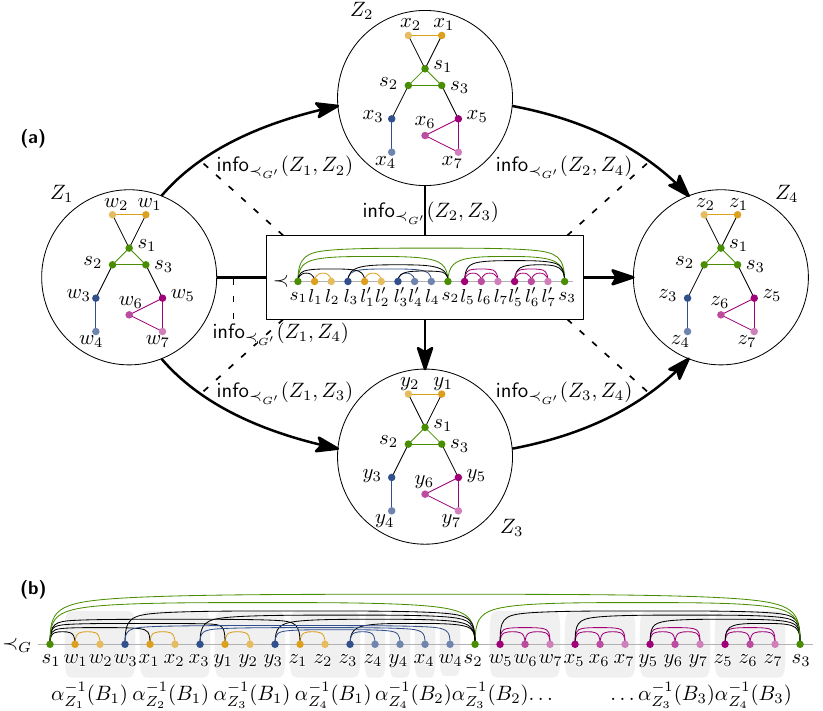}
    \caption{For the pattern from \Cref{fig:ramsey-overview}a, \textbf{\textsf{(a)}} the four large blocks $Z_1 \prec_{G'} Z_2 \prec_{G'} Z_3 \prec_{G'} Z_4 \in \mathcal{L} \setminus \{L, L'\}$ form a monochromatic $K_4$ in $\mathcal{H}$. 
    \textbf{\textsf{(b)}}~Moreover, $(Q, \zeta)$ partitions $R$ into three blocks $B_1 = \{r_1, r_2, r_3\}$, $B_2 = \{r_4\}$, and $B_3 = \{r_5, r_6, r_7\}$ with $\zeta(B_1) = \zeta(B_3) = \asc$ and $\zeta(B_2) = \desc$.}
    \label{fig:ramsey-blocks}
\end{figure}

\begin{lemma}
\label{lem:fewlarge}
There exists distinct large groups $Z_1,\dots,Z_{n_P} \in \mathcal{L} \setminus \{L,L'\} $ with $Z_i \prec_{G'} Z_j$ for each $1\leq i\leq j \leq n_P$ such that the following holds. For each $1\leq i\leq j\leq n_P$ and $1\leq i'\leq j' \leq n_P$, $\textsf{info}_{\prec_{G'}}(Z_i, Z_j)=\textsf{info}_{\prec_{G'}}(Z_{i'}, Z_{j'})$. %
\end{lemma}
\begin{proof}
Construct a complete edge colored auxiliary graph $\mathcal{H}$ 
with $V(\mathcal{H})=\mathcal{L} \setminus \{L,L'\} $. For $X,Y\in \mathcal{L} \setminus \{L,L'\} $ with $X\prec Y$, we set $\mathsf{color}((X,Y))=\mathsf{info}_{\prec_{G'}}(X,Y)$. Let $x:=|S\cup S'|$.

\begin{claim}
	\label{claim:edge-colors}
	The number of possible distinct edge colors of $\mathcal{H}$ is at most $(x+2\cdot p^2\cdot 2^{2p^2})!$.
\end{claim}
	\begin{claimproof}
		Recall that $\prec_{G'}$ is a solution for $(G',P)$. By definition, $\mathsf{info}_{\prec{G'}}(X,Y)$ for $X\prec Y\in \mathcal{L}\setminus \{L,L'\}$ is a total order of $G[S\cup S'\cup L\cup L']$. 
		The number of possible total orders of $G[S\cup S'\cup L\cup L']$ is at most $(|S\cup S'\cup L\cup L'|)!$. There are 
		at most $p\cdot 2^{2p^2}$ equivalence classes and each graph in any equivalence class has size at most $p$. Thus we can bound the size of large groups $|L|=|L'|\leq p^2\cdot 2^{2p^2}$. Then using $|S\cup S'|=x$
		and $|L|+|L'|\leq 2p^2\cdot 2^{2p^2}$, we bound the number of possible total orders of $G[S\cup S'\cup L\cup L']$ by $(x+2\cdot p^2\cdot 2^{2p^2})!$, as desired.
	\end{claimproof}

We now use a well-known fact (from Ramsey theory~\cite{Alon_asymptoticallytight}) that any edge-colored clique on $n$ vertices colored with $t$ colors has a monochromatic clique of size at least $\log_t(n)/t$. 

We have $|V(\mathcal{H})|=|\mathcal{L}|-2= f(n_P,p,x)-2\geq 2^{3n_P\cdot ((x+2\cdot p^2\cdot 2^{2p^2})!)^2}$.
Thus, it holds $n\geq 2^{3n_P\cdot ((x+2\cdot p^2\cdot 2^{2p^2})!)^2}$ and $t\leq (x+2\cdot p^2\cdot 2^{2p^2})!$ and 
 therefore there is a monochromatic clique of size at least $\log_t(n)/t\geq n_P$. %
The elements of a monochromatic clique of size $n_P$ in $\mathcal{H}$ yield  distinct large groups $Z_1,\dots,Z_{n_P}\in \mathcal{L}\setminus \{L,L'\}$ with the desired properties. 
\end{proof}

Let $Z_1,\dots,Z_{n_P}\in \mathcal{L}$ be a tuple of large groups witnessing the previous lemma. We now show that we can partition the vertices of $Z_1,\dots,Z_{n_P}$ in a nice way that will guide us on how to insert the remaining parts of the graph. 

First, notice that for each $1\leq i\leq j\leq n_P$, each $\prec_{G'}[Z_i\cup S\cup S']$ and $\prec_{G'}[Z_{j}\cup S\cup S']$ must behave in the same way when viewed as isomorphic copies with respect to $R$ due to definition of $\textsf{info}_{prec_{G'}}$; we now formalize this.

\begin{observation}
	\label{obs:internalorder}
	There is an ordering $\prec'$ of vertices of $R\cup S\cup S'$ such that:
	\begin{itemize}
		\item For each $s\in S\cup S'$ and $r\in R$, if $s\prec' r$ then $s\prec_{G'} r_{Z_i}$ for each $i\in [n_P]$;
		\item For each $s\in S\cup S'$ and $r\in R$, if $r\prec' s$ then $r_{Z_i} \prec_{G'} s$ for each $i\in [n_P]$;
		\item For each $u,v\in S\cup S'$, if $u\prec' v$ then $u\prec_{G'} v$;
		\item For each $u,v\in R$ if $u\prec'v$ then $u_{Z_i}\prec_{G'} v_{Z_i}$ for each $i\in [n_P]$.
	\end{itemize} 
\end{observation}

Note that Observation~\ref{obs:internalorder} does not provide any information about how the vertices of $Z_i$ interact with those of $Z_j$. Obtaining some structure in that regard will be our next task.
For the following, let us set $\Upsilon=S\cup S' \cup Z_1 \cup \dots \cup Z_{n_P}$. 
Let a \emph{block} be a consecutive subsequence of vertices from $R$ w.r.t.\ $\prec'$, and a \emph{solution-block} of $\prec_{G'}[G'[\Upsilon]]$ be a consecutive subsequence of vertices from $\Upsilon$; see also \Cref{fig:ramsey-blocks-example}(b).
\begin{figure}[h]
    \centering
	\includegraphics{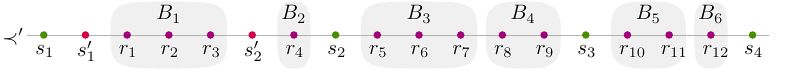}
    \caption{An ordering $\prec'$ of $R \cup S \cup S'$ and a partition of $R$ into the blocks $B_1$ to $B_6$, indicated in gray.
	Note that vertices of $S \cup S'$, colored green and red, respectively, always separate two blocks. However, not every pair of blocks is separated by such a vertex, see, for example, blocks $B_3$ and $B_4$.}
    \label{fig:ramsey-blocks-example}
\end{figure}

We are now ready to prove the structural result which provides strong guardrails on $\prec_{G'}[G'[\Upsilon]]$; in particular, it guarantees that all vertices from each $Z_i$ must occur in consecutive solution-blocks and following a strict ``ascending'' or ``descending'' order. We remark that this structural result is only made possible by the extraction of at least $n_P$ copies of large groups via Lemma~\ref{lem:fewlarge}; see also 
\Cref{fig:ramsey-blocks} for an illustration.

\begin{lemma}
\label{lem:structure}
There exists a pair $(Q,\zeta)$, where $Q$ is a partitioning of $R$ into blocks and $\zeta\colon Q\rightarrow\{\asc, \desc\}$, satisfying the following.
For every block $B\in Q$ with $\zeta(B) = \asc$, the sequence 
$\alpha^{-1}_{Z_1}(B)\alpha^{-1}_{Z_2}(B)\dots \alpha^{-1}_{Z_{n_P}}(B)$ 
is a solution-block of $\prec_{G'}[G'[\Upsilon]]$.
Moreover, for every block $B\in Q$ with $\zeta(B) = \desc$, the sequence $\alpha^{-1}_{Z_{n_P}}(B)\alpha^{-1}_{Z_{n_P-1}}(B)\dots \alpha^{-1}_{Z_1}(B)$ is a solution-block of $\prec_{G'}[G'[\Upsilon]]$.
\end{lemma}
\begin{proof}
By definition of $Z_1,\dots,Z_{n_P}$, we infer the following direct properties for each $1\leq i< j< q \leq n_P$, where we recall that for a large group $A$ and $u \in R$ the notation $u_A$ is a shorthand for $\alpha_A^{-1}(u)$:

\begin{itemize}
	\item For each $u\in R$, either $u_{Z_i}\prec u_{Z_j}\prec u_{Z_q}$ or $u_{Z_q}\prec u_{Z_j}\prec u_{Z_i}$.
	\item For any block $B$ of vertices in R, if all vertices in $\alpha^{-1}_{Z_i}(B)$ occur before all vertices in $\alpha^{-1}_{{Z_q}}(B)$ then all vertices in $\alpha^{-1}_{{Z_j}}(B)$ also occur before all vertices in $\alpha^{-1}_{{Z_q}}(B)$.
	\item For any block $B$ of vertices in R, if all vertices in $\alpha^{-1}_{{Z_q}}(B)$ occur before all vertices in $\alpha^{-1}_{{Z_i}}(B)$ then all vertices in $\alpha^{-1}_{{Z_j}}(B)$ also occur before all vertices in $\alpha^{-1}_{{Z_i}}(B)$.
\end{itemize}

On a high level we traverse the linear order $\prec_{G'}$ of $G'$ from left to right. We find the first vertex $u \in Z_1 \cup \dots \cup Z_{n_P}$. If $u$ belongs to $Z_1$, we then traverse $\prec$ and build a solution block $\alpha^{-1}_{Z_1}(B)$ containing only $Z_1$ vertices. This corresponds to block $B$ of $R$. We show below that $\alpha^{-1}_{Z_1}(B)\alpha^{-1}_{Z_2}(B) \dots \alpha^{-1}_{Z_{n_P}}(B)$ is a solution block. We then add $B$ to $Q$ and set $\zeta(B)=\asc$ and continue the traversal. If instead $u$ was a vertex in $Z_{n_P}$ then we would have found a descending block. We now formalize this argument.

Let $R':=R$ and let $u$ be the first vertex in $R'$ in $\prec'$. Further let $\alpha^{-1}_{Z_1}(u)\prec_{G'} \alpha^{-1}_{Z_2}(u) \prec_{G'} \dots \prec_{G'} \alpha^{-1}_{Z_{n_P}}(u)$. 
Choose $B=u\prec'\cdots \prec'w$ to be the largest block of $R$ starting with $u$ such that $\alpha^{-1}_{Z_1}(B)$ is a solution block, and let $d$ be the direct successor of $w$ according to $\prec'$ in $R$. That is we traverse $\prec_{G'}$ starting at $\alpha^{-1}_{Z_1}(u)$ until we find a vertex not in $X$. Note $u$ could be the same as $w$. We now show that $\alpha^{-1}_{Z_1}(B)\alpha^{-1}_{Z_2}(B)\dots \alpha^{-1}_{Z_{n_P}}(B)$ forms a solution block.

First we can infer that all vertices in $\alpha^{-1}_{{Z_1}}(B)$ occur before all vertices in $\alpha^{-1}_{{Z_{n_P}}}(B)$. This implies all vertices in $\alpha^{-1}_{{Z_j}}(B)$ for each $1<j<n_P$ also occur before all vertices in $\alpha^{-1}_{{Z_{n_P}}}(B)$. Therefore $\alpha^{-1}_{Z_1}(B)\alpha^{-1}_{Z_2}(u)$ is also a solution bloc; repeating the same argument yields that $\alpha^{-1}_{Z_1}(B)\alpha^{-1}_{Z_2}(B)\dots \alpha^{-1}_{Z_{n_P}}(B)$ forms a solution block which ends with $w_{Z_{n_P}}$. Next we show that $w_{Z_{n_P}}\prec_{G'} d_{Z_1}$.  

We know by the properties we observed earlier that either $d_{Z_{n_P}}\prec_{G'} d_{Z_2}\prec_{G'} d_{Z_1}$ or $d_{Z_1}\prec_{G'} d_{Z_2}\prec_{G'} d_{Z_{n_P}}$. If it is the former then $w_{Z_{n_P}}\prec_{G'} d_{Z_{n_P}}\prec_{G'} d_{Z_1}$ as desired. If it is the latter, then $d_{Z_1}$ is somewhere between $u_{Z_2}$ and $w_{Z_{n_P}}$ with $u_{Z_2}\prec_{G'} w_{Z_2}\prec_{G'} u_{Z_{n_P}} \prec_{G'} w_{Z_{n_P}}$.
We can divide this into two cases -- $u_{Z_2}\prec_{G'} d_{Z_1}\prec_{G'} u_{Z_{n_P}}$ or $w_{Z_2}\prec_{G'} d_{Z_1}\prec_{G'} w_{Z_{n_P}}$ -- but this would in either case imply that $\textsf{info}(Z_1,Z_2)=\textsf{info}(Z_2,Z_{n_P})=\textsf{info}(Z_1,Z_{n_P})$ is not true. This shows that $\alpha^{-1}_{Z_1}(B)\alpha^{-1}_{Z_2}(B)\alpha^{-1}_{Z_{n_P}}(B)$ is a solution block. We add $B$ to $P$ and set $\zeta(B)=\asc$. 

If we were in the case where $\alpha^{-1}_{Z_{n_P}}(u)\prec_{G'} \alpha^{-1}_{Z_2}(u) \prec_{G'} \dots \prec_{G'} \alpha^{-1}_{Z_1}(u)$, similar arguments will hold with $\alpha^{-1}_{Z_{n_P}}(B)\alpha^{-1}_{Z_2}(B)\dots\alpha^{-1}_{Z_1}(B)$ being a solution block and we set $\zeta(B)=\desc$. 
Finally, recursing with $R'=R'\setminus B$ completes the proof.
\end{proof}

Crucially, we can now use Lemma~\ref{lem:structure} to safely insert an arbitrary number of large groups into $G[\Upsilon]$ while avoiding $P$; this is handled by the following lemma.

\begin{lemma}
\label{lem:layout-supergraph}
Let $G^+$ be a supergraph of both $G$ and $G'$ such that each large equivalence class in $G^+$ has equal size. Then if $(G',P)$ admits a solution, so does $(G^+,P)$.
\end{lemma}

\begin{proof}
Let $t\geq 3$ be the maximum size of a large equivalence class of $\sim$. %
Hence, $G^+$ can be expressed as $G^+=G^+[S\cup S'\cup X_1\cup  \cdots \cup X_t]$ where for each $i\in[t]$ $X_i$, is a large group.  Recall that $\prec'$ is the ordering of $R\cup S\cup S'$ from Observation~\ref{obs:internalorder}. Moreover, let $(Q,\zeta)$ be a pair witnessing Lemma~\ref{lem:structure} where $Q$ is a partitioning of $R$ into blocks and $\zeta\colon P\rightarrow\{\asc, \desc\}$. %

We construct a total ordering  $\prec_{G^+}$ as follows: in  $\prec'$, we replace each block $B\in P$ with $\alpha_{X_1}^{-1}(B)\alpha_{X_2}^{-1}(B)\cdots\alpha_{X_t}^{-1}(B)$ if $\zeta(B)= \asc $ and with $\alpha_{X_t}^{-1}(B)\alpha_{X_{t-1}}^{-1}(B)\cdots\alpha_{X_1}^{-1}(B)$ if $\zeta(B)= \desc $.

To complete the proof, we show that $\prec_{G^+}$ avoids the pattern $P$. Towards a contradiction, assume that $\prec_{G^+}$ contains vertices $r_1 \prec_{G^+} \dots \prec_{G^+} r_{n_P}$ which match the pattern $P$. Observe that all but $n_P$ large groups from $G^+$ results in a subgraph isomorphic to $G^-$, and restricting $\prec_{G^+}$ to this subgraph yields $\prec_{G^-}$ (up to isomorphism). But since there are at most $n_P$ large groups intersecting $r_1 \prec_{G^+} \dots \prec_{G^+} r_{n_P}$, this yields a contradiction with $\prec_{G^-}$ being a solution for $(G^-,P)$.
\end{proof}

Lemma~\ref{lem:kernel} now follows directly from Lemma~\ref{lem:reducedcomputation}, the known algorithm for computing a witness $S$ for vertex integrity~\cite{DDvH.CCV.2016}, the trivial observation that YES-instances are closed under vertex deletion (as discussed at the beginning of Subsection~\ref{subsec:VI-Kernel-proof}) and Lemma~\ref{lem:layout-supergraph}.

\section{Parameterization by the Neighborhood Diversity + $\boldsymbol{m_P}$}
\label{sec:nd-plus-mp}
In this section, we establish \Cref{thm:fpt-nd-mp}, i.e., that \PA is \FPT\ parameterized by the neighborhood diversity $\nd(G)$ of $G$ and the number $m_P$ of edges in the pattern $P$.
For the remainder of this section, let us fix an instance $\Instance = \InstanceLong$ of \PA.
Recall from \cref{sec:preliminaries} that two vertices $u$ and $v$ have the same neighborhood type, denoted as $u \sim_N v$, if and only if they have the same neighborhood after ignoring $u$ and $v$, i.e., $N_G(u) \setminus \{v\} = N_G(v) \setminus \{u\}$.

Our proof strategy consists of two major steps.
First, we show in \Cref{sec:nd-plus-mp-alternations} that we can assume without loss of generality that a solution $\prec_G$ to \Instance, if it exists, is well-structured.
In particular, we derive a bound on the number of neighborhood type alternations along $\prec_G$.
This allows us to branch over all possible alternations that we might encounter in a solution.
Moreover, each branch represents a bounded-size sequence $S$ of neighborhood types, encoding intervals in the vertex order consisting of vertices of that particular type.
In our second step, we devise an integer linear programming (ILP) formulation for each branch with a bounded number of variables to check if we can distribute the vertices of $G$ over $S$ while avoiding the pattern $P$; see \Cref{sec:nd-plus-mp-ilp}.
For the remainder of this section, we fix a hypothetical solution $\prec_G$ to \Instance.

\subsection{Bounding Neighborhood-Type Alternations Within Solutions}
\label{sec:nd-plus-mp-alternations}
We call a vertex $v \in V(P)$ \emph{critical} in $P$ if $v$ is the first or last vertex of $\prec_P$ or if there exists an edge $uv \in E(P)$.
We let $C$ denote the set of critical vertices and make the following observation:
\begin{observation}
    \label{obs:nd-plus-mp-number-critical}
    There are at most $2(m_P + 1)$ critical vertices.
\end{observation}
Next, we define the concept of the \emph{compressed} pattern $P^C$ with respect to the critical vertices $C$.
To construct $P^C$ from $P$, we replace every maximal and consecutive sequence of non-critical vertices in $P$ with its length; see \Cref{fig:compressed-pattern} for an example.

Intuitively, $P^C$ corresponds to a (restricted) run-length encoding of $P$.
Observe that the size of $P^C$, i.e., the number of edges, (critical) vertices, and counts in $P^C$, is linear in~$m_P$.
Let $\alpha\colon C \to V(G)/{\sim_N}$ %
be an assignment of critical vertices to neighborhood types.
We call $\alpha$ an \emph{assignment function}.
Thanks to \Cref{obs:nd-plus-mp-number-critical}, the number of potential assignment functions is bounded in our parameter:
\begin{observation}
	\label{obs:nd-plus-mp-number-critical-assignment}
	There are at most ${\nd(G)}^{2(m_P + 1)} \in \nd(G)^{\mathcal{O}(m_P)}$ different assignment functions $\alpha$.
\end{observation}
We call an assignment function $\alpha$ \emph{provoking} if (i) for every forced edge $uv \in E^+(P)$ we have that vertices $a \in \alpha(u)$ and vertices $b \in \alpha(v)$ are adjacent in $G$ and for every forbidden edge $uv \in E^-(P)$ we have that vertices in $a \in \alpha(u)$ and $b \in \alpha(v)$ are not adjacent in $G$.
Intuitively, a provoking assignment function witnesses a potential way of introducing the pattern $P$ in a linear order $\prec_G'$ and conversely, every possible way of introducing the pattern corresponds to a provoking assignment.

We use provoking assignment functions to decompose the solution $\prec_G$ into a bounded number of \emph{safe intervals} within which the vertices can be reordered without introducing the pattern, i.e., the resulting order~$\prec_G'$ will still be a solution.
In particular, we will show that vertices within the same interval can be ordered according to their neighborhood type, which allows us to bound the number of neighborhood-type alternations within each interval.
To this end, consider an arbitrary provoking assignment function $\alpha$. 
We use $\alpha$ to introduce \emph{cuts} in $\prec_G$ along the left-to-right traversal $\tau\colon [n] \to V(G)$ that enumerates the vertices according to~$\prec_G$. %
These cuts define segments in $\tau$ that we later use to argue that our reordering operation is safe \emph{with respect to $\alpha$}.
Since we need to ensure safeness with respect to every provoking assignment function, we repeat the steps that we describe next for every possible provoking assignment function $\alpha$.
Eventually, we partition $\prec_G$ along the union of these cuts into the aforementioned safe intervals.

We place the cuts as follows.
We begin with the leftmost entry in the compressed pattern~$P^C$.
By definition, this entry corresponds to a critical vertex $v \in C$.
Let $u$ be the first vertex in $\tau$ such that $u \in \alpha(v)$.
\begin{figure}
    \centering
    \includegraphics{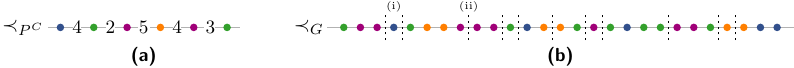}
    \caption{\textbf{\textsf{(a)}} A compressed pattern $P^C$ together with an assignment function $\alpha$ encoded via the colored vertices. \textbf{\textsf{(b)}} The cuts that we introduce in $\prec_G$ based on $\alpha$, indicated with the dashed lines. Colors encode the neighborhood types and we omit edges for clarity.}
    \label{fig:neighborhood-diversity-cuts}
\end{figure}
We place a cut immediately to the left and right of $u$, i.e., between $\Pred{u}$ and $u$, and between $u$ and $\Succ{u}$, if these vertices exist; see the cuts marked by (i) in \Cref{fig:neighborhood-diversity-cuts}b for an example.
We call $u$ the \emph{fixed point} for the critical vertex $v$, indicated as $\alpha^v$.
After placing the cuts, we move one position to the right in the compressed pattern $P^C$.
If the next entry of $P^C$ is a count $c$ (rather than a critical vertex), we advance $c$ steps in the traversal $\tau$ and place the next cut immediately after the vertex $\tau(\tau^{-1}(u) + c)$, if it exists; see also the cut marked by (ii) in \Cref{fig:neighborhood-diversity-cuts}b.
Observe that we traversed a distance of $c$ vertices in the linear order $\prec_G$.
We refer to these vertices as the \emph{fixed segment} of the count $c$, indicated as $\alpha^{[c]}$.
We then advance one position in both $P^C$ (to the next critical vertex $v'$) and $\tau$.
Otherwise, if the next entry of $P^C$ is another critical vertex $v' \in C$, we advance only one position in $\tau$, i.e. after $u\in\tau$.
Afterwards, i.e., once the above is handled, we repeat the procedure, i.e., find the next vertex $u'$ in $\tau$ such that $u' \in \alpha(v')$ and place the cuts.
We repeat this process until all entries of $P^C$ have been processed or the traversal $\tau$ reaches the end of $\prec_G$.

We perform the above steps for every provoking assignment function $\alpha$ %
and consider the union of all cuts introduced in that way.
The resulting maximal contiguous subsequences of $\prec_G$ between two consecutive cuts define the \emph{safe intervals} of $\prec_G$.
Below we show that their number is bounded.

\begin{lemma}
	\label{lem:nd-plus-mp-number-intervals-number}
	There are at most $\BigO{m_P} \cdot {\nd(G)}^{\BigO{{m_P}}}$ safe intervals.
\end{lemma}
\begin{proof}
	We first establish a bound on the number of cuts made for a single provoking assignment function~$\alpha$.
	With this, we can bound the overall number of cuts made and, consequently, the number of safe intervals.
	
	Let $P^C$ be the compressed pattern $P$ and let $\alpha\colon C \to V(G)/{\sim_N}$ be a provoking assignment function.
	Consider an entry $a$ of the compressed pattern $P^C$.
	If $a$ is a critical vertex $a \in C$, then we make at most one cut before and at most  one cut after the next vertex $u$ in $\prec_G$ with $u \in \alpha(a)$, i.e., in total at most two cuts.
	Since $P^C$ contains $\Size{C}$ critical vertices, this yields at most $2\Size{C}$ cuts so far.
	If $a$ is a count, we make at most one cut, namely after the next $a$ vertices in~$\prec_G$.
	Since the first and last vertex of $\prec_P$ are critical and since $P^C$ contains at most one count between two critical vertices, we make at most $\Size{C} - 1$ additional cuts.
	Overall, this yields at most $3\Size{C} - 1\leq 3\Size{C} \leq 6(m_P + 1)$ cuts that we make for a single provoking assignment function $\alpha$, where the last inequality follows from the bound on the number of critical vertices established in \Cref{obs:nd-plus-mp-number-critical}.
	
	The number of safe intervals corresponds to the overall number of cuts we make (plus one).
	We can bound the latter by combining the above bound for a single assignment function with the number of possible provoking assignment functions.
    Recall \Cref{obs:nd-plus-mp-number-critical-assignment}, which gives a bound on the overall number of assignment functions, and thus also provoking assignment functions.
	More precisely, we obtain ${6(m_P + 1)}\cdot{\nd(G)}^{2(m_P + 1)} \in \BigO{m_P} \cdot {\nd(G)}^{\BigO{m_P}}$ safe intervals.
\end{proof}

Our next task is to show that we can reorder the vertices within safe intervals based on their neighborhood type without invalidating a solution.
To formalize this goal, let $\iota\colon V(G)/{\sim_N}\to [\nd(G)]$ be an \emph{indexation} of the neighborhood types of $G$.
Moreover, let $s$ denote the number of safe intervals in $\prec_G$ and let $\xi \colon V \to [s]$ be a function that assigns each vertex $v \in V(G)$ to its safe interval in $\prec_G$, addressed via their implicit numbering from left to right.
We now transform the linear order $\prec_G$ into a \emph{normalized} linear order $\prec_G^N$ using the following rules, applied in this order; see also \Cref{fig:neighborhood-diversity-normalized-order} for an example.
\begin{figure}
    \centering
    \includegraphics{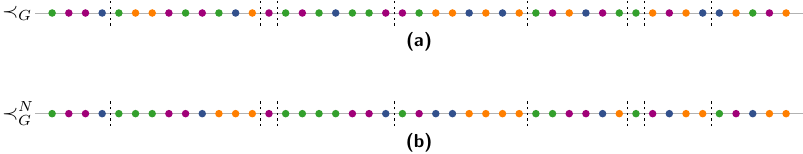}
    \caption{\textbf{\textsf{(a)}} A solution $\prec_G$ partitioned into its safe intervals (indicated via the dashed lines) and \textbf{\textsf{(b)}} its normalized variant $\prec_G^N$. Colors indicate neighborhood classes and we omit edges for clarity.}
    \label{fig:neighborhood-diversity-normalized-order}
\end{figure}
For every pair of vertices $u,v\in V(G)$, $u \neq v$, 
\begin{enumerate}
	\item if $\xi(u) \neq \xi(v)$, we set $u \prec_G^N v$ if $\xi(u) < \xi(v)$ and $v \prec_G^N u$ if $\xi(v) < \xi(u)$;
	\item else, if $u \not\sim_N v$, we set $u \prec_G^N v$ if $\iota(u) < \iota(v)$ and $v \prec_G^N u$ if $\iota(v) < \iota(u)$;
	\item else we set $u \prec_G^N v$ if $u \prec_G v$ and $v \prec_G^N u$ if $v \prec_G u$.
\end{enumerate}
On an intuitive level, we (1.) use the order of the safe intervals, (2.) break ties within the same safe interval based on the neighborhood type, and, if the tie persists, (3.) on the original solution $\prec_G$.
In the next lemma, we show that the above-introduced normalization process preserves solutions.
\begin{lemma}
	\label{lem:nd-plus-mp-reorder-safe-intervals}
	Let $\prec_G^N$ be the normalized linear order obtained from the solution $\prec_G$.
	The linear order~$\prec_G^N$ avoids the pattern $P$, i.e., it is a solution to \Instance.
\end{lemma}
\begin{proof}
	Towards a contradiction, assume that $\prec_G$ is a solution but $\prec_G^N$ is not.
	In particular, this means that $\prec_G^N$ contains $P$, i.e., there exists a set $X \subseteq V(G)$ such that $(\prec_G^N, X)$ realizes the pattern $P$.
    In the following, we use $\delta$ as a shorthand for $\delta(\prec_G^N, X)$.
	Let $\alpha_{X}$ be the assignment function that matches every critical vertex $v \in C$ of $P$ to the neighborhood type as witnessed in $\delta$, i.e., for every $v \in C$ we have $\delta^{-1}(v) \in \alpha_X(v)$.
    Observe that $\alpha_X$ is provoking, as witnessed by $(\prec_G^N, X)$.
    Moreover, when defining the safe intervals, we considered every provoking assignment function $\alpha$, and in particular $\alpha_{X}$.
	We now use the cuts introduced for $\alpha_{X}$ to show that there exists a set $X' \subseteq V(G)$ %
	such that $(\prec_G, X')$ also realizes $P$, i.e.,~$\prec_G$ is not a solution.
	Moreover, for every critical vertex $v \in C$ we ensure $\delta^{-1}(v) \sim_N \delta(\prec_G, X')^{-1}(v)$, i.e., they are from the same neighborhood type.
	This way, we preserve the ``existence'' of forced and forbidden edges, respectively, and ensure that $(\prec_G, X')$ is a realization of $P$.
	Note that non-critical vertices are isolated vertices in the pattern and their precise neighborhood type is, therefore, irrelevant.
	
	Let the critical vertices $C = \{v_1, v_2, \ldots, v_{\Size{C}}\}$ be ordered according to $\prec_P$, i.e., $v_1 \prec_P v_2 \prec_P \ldots \prec_P v_{\Size{C}}$.
	We now iteratively construct the set $X'$ %
	as follows.
	\begin{enumerate}
		\item Initialize $X' = \emptyset$. %
        \item Start before the leftmost vertex in $\prec_G$ and set $i=1$.
		\item Consider the critical vertex $v_i$:\label{lem:nd-plus-mp-reorder-safe-intervals-step-critical-vertex}
		\begin{enumerate}
            \item Let $u_i$ be the next vertex in $\prec_G$ of neighborhood type $\alpha_X(v_i)$.
            If $u_i$ exists, then
            \item Set $X' \gets X' \cup \{u_i\}$
		\end{enumerate}
		\item Consider the count $c$ between the critical vertices $v_i$ and $v_{i+1}$ (if it exists):\label{lem:nd-plus-mp-reorder-safe-intervals-step-count}
		\begin{enumerate}
            \item Let $U_{i,i+1}$ be the next $c$ vertices in $\prec_G$ after $u_i$.
            If $U_{i,i+1}$ exists, then
			\item Set $X' \gets X' \cup U_{i,i+1}$
		\end{enumerate}
		\item Repeat Steps~\ref{lem:nd-plus-mp-reorder-safe-intervals-step-critical-vertex} and~\ref{lem:nd-plus-mp-reorder-safe-intervals-step-count} for every $2 \leq i \leq \Size{C}$; recall that the critical vertex $v_{\Size{C}}$ is the last vertex in the pattern and hence we skip Step~\ref{lem:nd-plus-mp-reorder-safe-intervals-step-count} for $i = |C|$.
	\end{enumerate}
	To complete the proof, we still have to argue (i) that the above procedure terminates with a set $X'$ with $\Size{X'} =n_P$ %
	and (ii) that $(\prec_G, X')$ realizes $P$.
	We start with establishing Part~(i).
	\begin{claim}
		\label{claim:nd-plus-mp-reorder-safe-intervals-terminate}
		The set $X'$ has size $\Size{X'} = n_P$. %
	\end{claim}
	\begin{claimproof}    
        Observe that our strategy for identifying these vertices resembles a greedy strategy, i.e., for a fixed position in $\prec_G$, we picked the next vertex from the corresponding neighborhood type, and for a fixed segment, we picked the next $c$ vertices to fill up the count.
		In particular, we show that for every vertex $v \in X$ the corresponding vertex $v' \in X'$ exists and we have %
        $\xi(v') \leq \xi(v)$, i.e., $v'$ lies in the same interval as $v$ or to the left of $v$.

        To this end, consider the first critical vertex $v_1 \in C$.
        Since $(\prec_G^N, X)$ realizes $P$, the set $X$ (together with $\delta(\prec_G^N, X)$) witnesses that there exists a vertex $u \in V(G)$ of neighborhood type $\alpha_X(v_1)$.
        Let $u_1$ be the leftmost vertex in $\prec_G$ of neighborhood type $\alpha_X(v_1)$.
        By definition of $\alpha_X$, this implies that we have two cuts around $u_1$ in $\alpha_X$ and, therefore, $\xi(u_1) \leq \xi(\delta(\prec_G^N, X)^{-1}(v_1))$.
        Since we take in Step~\ref{lem:nd-plus-mp-reorder-safe-intervals-step-critical-vertex} the first vertex of neighborhood type $\alpha_X(v_1)$, we can assume $u_1 \in X'$.
        Next, consider the count $c$ in $P^C$ between the critical vertices $v_1$ and $v_2$.
        Let $v \in V(P)$ be the rightmost vertex of $P$ represented by $c$.
        The realization $(\prec_G^N, X)$ witnesses that there exist $\Size{\{w \in V(P) \mid v_1 \prec_P w \prec_P v_2\}}$ vertices after $u_1$ in $\prec_G$ since there is a cut after $u_1$ in $\prec_G$.
        Hence, the cut for $c$ and $\alpha_X^{[c]}$ exists.
        Let $u$ be the rightmost vertex in $\alpha_X^{[c]}$.
        Since we take in Step~\ref{lem:nd-plus-mp-reorder-safe-intervals-step-count} the next $c$ vertices after $u_1$, we have $\xi(u) \leq \xi(\delta(\prec_G^N, X)^{-1}(v))$ and we can assume that $X'$ contains $\alpha_X^{[c]}$.
        Moreover, as there is a cut in $\prec_G$ after $u$, we can only move $u$ to the left when we re-order the vertices within the safe intervals, i.e., $u \preceq_G \delta(\prec_G^N, X)^{-1}(v)$ holds.
        Furthermore, for the next critical vertex $v_2 \in C$ it holds $u \prec_G \delta(\prec_G^N, X)^{-1}(v_2)$, i.e., it is left from the cut after $u$.
        At this point, 
        analogous arguments can be made for the remaining critical vertices and counts in $P^C$.
        Note that while $(\prec_G^N, X)$ guarantees for every critical vertex $v' \in C$ the existence of a vertex $u'$ in $\prec_G$ of the same neighborhood type as $v'$, the fact that we only reorder within safe intervals, and that the next vertex of a corresponding neighborhood type defines up to three cuts, ensures that $u'$ never lies to the left of the next vertex in $\prec_G$ that we would pick in Steps~\ref{lem:nd-plus-mp-reorder-safe-intervals-step-critical-vertex} and~\ref{lem:nd-plus-mp-reorder-safe-intervals-step-count}.
        Combining all, we can conclude that $X'$ indeed contains one vertex for every vertex of the pattern, i.e., $\Size{X'} = n_P$ holds.
	\end{claimproof}
	By \Cref{claim:nd-plus-mp-reorder-safe-intervals-terminate}, the set $X'$ is well-defined in the sense that it contains one vertex from $G$ for every vertex of $P$.
	We next show that it also witnesses $P$ in $\prec_G$, i.e., establish Part~(ii).
	\begin{claim}
		\label{claim:nd-plus-mp-reorder-safe-intervals-witness}
		The tuple $(\prec_G, X')$ witnesses the pattern $P$. %
	\end{claim}
	\begin{claimproof}
		Towards a contradiction, assume that the tuple $(\prec_G, X')$ does not witness the pattern $P$.
		Recall that we use $\delta$ as a shorthand for $\delta(\prec_G^N, X)$.
		Moreover, we now use $\delta'$ as a shorthand for $\delta(\prec_G, X')$ to improve readability.
		
		There are three cases how $(\prec_G, X')$ could not witness $P$: (i) There is a forced edge $uv \in E^+(P)$ such that $\delta'^{-1}(u)\delta'^{-1}(v)\notin E(G)$, (ii) there is a forbidden edge $uv \in E^-(P)$ such that $\delta'^{-1}(u)\delta'^{-1}(v)\in E(G)$, or (iii) there are two critical vertices $u,v\in C$ such that 
        \begin{align*}
            \Size{\{w \in V(P) \mid u \prec_P w \prec_P v\}} \neq \Size{\{w \in V(G) \mid \delta'^{-1}(u) \prec_G\mid_{X'} w \prec_G\mid_{X'} \delta'^{-1}(v)\}}.    
        \end{align*} %
		In the following we argue that none of these cases is possible.
		
		For Cases~(i) and~(ii), we observe that the vertices $u,v \in V(P)$ are critical, i.e., $u,v\in C$.
		When constructing the set $X'$, we employed a neighborhood-type preserving replacement for critical vertices, i.e., we have $\delta^{-1}(u) \sim_N \delta'^{-1}(u)$ and $\delta^{-1}(v) \sim_N \delta'^{-1}(v)$.
		By the definition of neighborhood types, we, therefore, have $\delta'^{-1}(u)\delta'^{-1}(v) \in E(G)$ if and only if $\delta^{-1}(u)\delta^{-1}(v) \in E(G)$.
		As $(\prec_G^N, X)$ witnesses $P$, the forced edge $uv \in E^+(P)$ also exists in $G$, i.e., $\delta'^{-1}(u)\delta'^{-1}(v)\in E(G)$ and the forbidden edge $uv \in E^-(P)$ also does not exists in $G$, i.e., $\delta'^{-1}(u)\delta'^{-1}(v)\notin E(G)$.
		Thus, both cases are not possible and it remains to argue Case~(iii).
		For this case, we can assume that $u$ and $v$ are two neighboring critical vertices in $P$, i.e., there does not exist a critical vertex $w \in C$ such that $u \prec_P w \prec_P v$.
		This is without loss of generality since if $u$ and $v$ are not neighboring, we can always find another pair $u'$ and $v'$ of neighboring critical vertices for which Case~(iii) also applies.
		Observe that Case~(iii) now boils down to not satisfying the count $c$ between the two critical vertices $u$ and $v$ in $P^C$.
		Furthermore, recall that we added for every count, and in particular for the count $c$ between $u$ and $v$, the safe segment to $X'$.
		This safe segment consists of exactly $c$ vertices.
		Thus, Case~(iii) is not possible by construction.
		Since none of the cases is possible, we can conclude that $(\prec_G, X')$ witnesses the pattern~$P$.%
\end{claimproof}
	
	\Cref{claim:nd-plus-mp-reorder-safe-intervals-witness} is the last missing ingredient to complete the proof of the lemma.
	In particular, recall that we assumed, towards a contradiction, that $\prec_G$ is a solution but $\prec_G^N$ is not.
	\Cref{claim:nd-plus-mp-reorder-safe-intervals-terminate,claim:nd-plus-mp-reorder-safe-intervals-witness} together show how to construct a set $X' \subseteq V(G)$ such that $(\prec_G, X')$ witnesses the pattern $P$.
	This contradicts our assumption that $\prec_G$ was a solution.
	Thus, we can conclude that if $\prec_G$ is a solution, so is $\prec_G^N$, which completes the proof.
\end{proof}

\Cref{lem:nd-plus-mp-reorder-safe-intervals} allows us to bound the number of neighborhood-type alternations along a solution $\prec$ to \Instance, where an alternation occurs whenever we have two neighboring vertices $u \prec v$ with $u \not\sim_{N} v$.
Consider the linear order $\prec_G^N$, which is a solution by \Cref{lem:nd-plus-mp-reorder-safe-intervals}.
Within each safe interval, there are at most $\nd(G) - 1$ alternations.
Moreover, we have (at most) one additional alternation at the boundaries of a safe interval.
Since there are $\BigO{m_P} \cdot {{\nd(G)}}^{\BigO{m_P}}$ 
safe intervals by \Cref{lem:nd-plus-mp-number-intervals-number}, we obtain the following corollary.
\begin{corollary}
	\label{cor:nd-plus-mp-alternations-bound}
	If \Instance is a positive instance, then there exists a solution with $\BigO{m_P} \cdot {\nd(G)}^{\BigO{m_P}}\cdot \BigO{\nd(G)}$ neighborhood-type alternations.
\end{corollary}
\Cref{cor:nd-plus-mp-alternations-bound} is valuable since it allows us to decompose the linear order into a bounded number of maximal consecutive subsequences of vertices from the same neighborhood type, in the following called \emph{segments}.
Moreover, the exact order among vertices in the same segment is irrelevant, since they have the same neighborhood type.
All that matters is the number of vertices (of the correct type) that are within a particular segment.
In the next section, we devise an integer linear program for this step.

\subsection{An Integer Linear Program for Determining the Size of Segments}
\label{sec:nd-plus-mp-ilp}
For the remainder of this section, let us fix a tuple $(S, \gamma)$, where $S = (s_1, \ldots, s_{\ell})$ is a sequence of segments in a hypothetical solution $\prec_G'$ and $\gamma\colon S \to V(G)/{\sim_{N}}$ is a function that associates each segment $s_i \in S$ with a neighborhood type $\gamma(s_i) \in V(G)/{\sim_{N}}$.
By \Cref{cor:nd-plus-mp-alternations-bound}, we can bound $\Size{S} = \ell$ by a function in $\nd(G) + m_P$.
Moreover, we can assume without loss of generality that every segment $s_i \in S$ contains at least one vertex $v \in V(G) \cap \gamma(s_i)$ and no two neighboring segments are associated with the same neighborhood type, i.e., we have $\gamma(s_i) \neq \gamma(s_{i + 1})$ for every $i \in [\ell - 1]$.
However, note that the exact number of vertices inside a segment still need to be determined.
In the following, we construct an integer linear program (ILP) that determines how many vertices are placed in each segment $s_i \in S$.
The ILP consists of two components, each composed of a series of variables and constraints.
The first component determines how many vertices are assigned to each segment and ensures that every vertex is (implicitly) assigned to one segment.
The second component is devoted to ensuring the avoidance of the pattern $P$ in the eventually obtained linear order $\prec_G'$.
In the following, we say that a linear order $\prec_G'$ of $V(G)$ \emph{realizes} $(S, \gamma)$ if $\prec_G'$ decomposes into the sequence of segments together with their neighborhood type as specified in $(S, \gamma)$.

\subparagraph*{Filling the Segments.}
We introduce one integer variable $x_i \geq 1$ to represent the cardinality of every segment $s_i \in S$.
Let $t \in V(G)/{\sim_{N}}$ be a neighborhood type of $G$ and let $S(t)$ denote the set of segments $s_i$ with $\gamma(s_i) = t$.
Note that $t \subseteq V(G)$ is a subset of vertices.
To ensure that every vertex $v \in t$ is represented in exactly one segment, we introduce the constraint $\sum_{s_i \in S(t)} x_i = \Size{t}$ for every neighborhood type $t \in V(G)/{\sim_{N}}$.
Let $\mathcal{X}(G)$ denote the set of so-far introduced variables.
In the following, we say that a linear order $\prec_G'$ \emph{realizes} a variable assignment $\Psi \colon \mathcal{X}(G) \to \mathbb{Z}$ if it realizes $(S, \gamma)$ and every segment $s_i \in S$ in $\prec_G'$ consists of $\Psi(x_i)$ vertices.

\subparagraph*{Avoiding the Pattern.}
Recall that $C$ denotes the set of critical vertices of the pattern $P$.
Let $\beta\colon C\to S$ denote \emph{segment assignment} of the critical vertices of $P$ to segments in $S$.
Similar to \Cref{sec:nd-plus-mp-alternations}, we call $\beta$ a \emph{provoking} segment assignment if (i) for every pair $u,v \in C$ of critical vertices with $\beta(u) = s_i$ and $\beta(v) = s_j$ we have that $u \prec_P v$ implies $i \leq j$ and (ii) the function $\gamma \circ \beta$ is a provoking assignment function.
Intuitively, a provoking segment assignment tells us that a linear order $\prec_G'$ realizing $(S, \gamma)$ contains the pattern $P$ (where each critical vertex $v \in C$ is from the segment $\beta(v)$) as soon as we satisfy the count (of isolated vertices) between every pair of critical vertices.
Therefore, our task is now to ensure that at least one of these counts is not fulfilled.

To this end, let $\beta$ be a provoking segment assignment.
Moreover, let $u,v\in C$ be two critical vertices with $u\prec_P v$ and let $s_i = \beta(u),s_j=\beta(v) \in S$ be the two assigned segments.
Note that we have $i \leq j$ since $\beta$ is provoking.
We let $d_P(u,v)$ denote the distance of the two critical vertices $u$ and $v$ in the pattern $P$, i.e., the number of (not necessarily critical) vertices $w\in V(P)$ with $u \preceq_P w \preceq_P v$.
Similarly, we let $d_S(s_i,s_j) = d_S(\beta(u), \beta(v))$ denote the distance between (and including) the two segments $s_i$ and $s_j$, i.e., $\sum_{i \leq q \leq j} x_q$.

Towards ensuring that we avoid the pattern, we introduce one binary variable $x_{u,v}^{\beta} \in \{0,1\}$ for every pair of critical vertices $u,v \in C$, $u \prec_P v$.
We will use the variable $x_{u,v}^{\beta}$ to indicate whether the segments between $s_i$ and $s_j$ contain enough vertices to represent all vertices between (and including) $u$ and $v$ in $P$, i.e., whether $d_S(\beta(u), \beta(v)) \geq d_P(u,v)$.
More formally, we should have $x_{u,v}^{\beta} = 1$ whenever $d_S(\beta(u), \beta(v)) \geq d_P(u,v)$.
To this end, we define $D_{u,v}^{\beta} \coloneqq d_S(\beta(u),\beta(v)) - d_P(u,v)$.
Note that $d_S(\beta(u), \beta(v)) \geq d_P(u,v)$ whenever $D_{u,v}^{\beta} \geq 0$.
We now need to enforce $x_{u,v}^{\beta} = 1$ in this case and introduce the constraint
\begin{equation}
    (n + 1) x^{\beta}_{u,v} \geq D_{u,v}^{\beta} + 1,
    \label{eq:neighborhood-diversity-pattern-constraint}
\end{equation}
where we recall that $n = \Size{V(G)}$ denotes the number of vertices in $G$.
Observe that whenever $D_{u,v}^{\beta} \geq 0$, $x^{\beta}_{u,v}=1$ has to hold to fulfill this constraint.
On the other hand, if $D_{u,v}^{\beta} < 0$, then the value of $x^{\beta}_{u,v}$ is not constrained.
However, we will ensure in the end that it is ``favorable'' to set $x^{\beta}_{u,v} = 0$ in this case.

With the following two lemmas, we show that the value of $D_{u,v}^{\beta}$, and thus also the variables $x_{u,v}^{\beta}$, actually encode the (non-)existence of the pattern $P$, at least with respect to the provoking segment assignment $\beta$.
To make the latter more formal, let $\prec_G'$ be a linear order that realizes $(S, \gamma)$ and, for a segment $s \in S$, let $V(G)[\prec_G',s]$ denote the vertices contained in $s$ in $\prec_G'$.
If $\prec_G'$ contains $P$, then we call a witnessing tuple $(\prec_G', X)$ \emph{consistent} with $\beta$ if, for every critical vertex $u \in C$, we have that $\delta(\prec_G', X)^{-1}(u) \in V(G)[\prec_G',\beta(u)]$.

In the next lemma, we show that if for a variable assignment $\Psi$ we have $D_{u,v}^{\beta} < 0$ for some $u,v\in C$, then every linear order $\prec_G'$ that realizes $\Psi$ either avoids $P$ or has no tuple $(\prec_G', X)$ consistent with $\beta$.
\begin{lemma}
	\label{lem:nd-plus-mp-ilp-a-small}
    Let $\Psi \colon \mathcal{X}(G) \to \mathbb{Z}$ be a variable assignment such that $D_{u,v}^{\beta} < 0$ holds for some $u,v\in C$.
    Furthermore, let $\prec_G'$ be a linear order of $V(G)$ that realizes $\Psi$.
    If there exists a set $X \subseteq V(G)$ such that $(\prec_G', X)$ realizes $P$, then $(\prec_G', X)$ is inconsistent with $\beta$.
\end{lemma}
\begin{proof}
	Let $\Psi$ be a variable assignment as in the statement and let $\prec_G'$ be a linear order that realizes $\Psi$.
	For the sake of the proof, assume that there exists a set $X \subseteq V(G)$ such that $(\prec_G', X)$ realizes $P$, otherwise, there is nothing to show.
	Towards a contradiction, assume that $(\prec_G', X)$ is consistent with $\beta$.
	Let $u,v \in C$ be the two critical vertices with $u \prec_P v$ such that $D_{u,v}^{\beta} < 0$.
	Their existence is guaranteed by $\Psi$.
	Recall $D_{u,v}^{\beta} = d_S(\beta(u),\beta(v)) - d_P(u,v)$, from which we derive that $D_{u,v}^{\beta} < 0$ implies $d_P(u,v) > d_S(\beta(u),\beta(v))$.
	Hence, if $\beta(u) = s_i$ and $\beta(v) = s_j$, the segments $s_q$ with $i \leq q \leq j$ contain in sum strictly fewer vertices than required by $P$, i.e., less than the number of vertices in $V' = \{w \in V(P) \mid u \preceq_P w \preceq_P v\}$.
	Thus, the segments $s_q$ contain not enough vertices to represent all vertices $w \in V'$ with $u \preceq_P w \preceq_P v$, implying that $(\prec_G', X)$ cannot realize $P$.
	Thus, if $(\prec_G', X)$ exists, it must be inconsistent with $\beta$.
\end{proof}
\Cref{lem:nd-plus-mp-ilp-a-small} shows that having $D_{u,v}^{\beta} <0$ for some $u,v \in C$ is sufficient for avoiding the pattern, at least with respect to the provoking segment assignment $\beta$.
With the next lemma, we show that this property is also necessary.
\begin{lemma}
	\label{lem:nd-plus-mp-ilp-a-large}
	Let $\Psi \colon \mathcal{X}(G) \to \mathbb{Z}$ be a variable assignment such that $D_{u,v}^{\beta} \geq 0$ holds for all $u,v\in C$, $u \prec_P v$.
	Then, for every linear order $\prec_G'$ that realizes $\Psi$ there exists a set $X \subseteq V(G)$ such that $(\prec_G', X)$ realizes $P$.
    Moreover, $(\prec_G', X)$ is consistent with $\beta$.
\end{lemma}
\begin{proof}
	Let $\Psi \colon \mathcal{X}(G) \to \mathbb{Z}$ be a variable assignment such that $D_{u,v}^{\beta} \geq 0$ holds for all $u,v\in C$, $u \neq v$ and let $\prec_G'$ be an arbitrary linear order that realizes $\Psi$.
	In the following, we show how to construct a set $X \subseteq V(G)$ such that $(\prec_G', X)$ realizes $P$.
    Moreover, the way we construct $X$ will guarantee that $(\prec_G', X)$ is consistent with $\beta$.
	To this end, we use an ``inductive construction strategy'' and iteratively identify for the next vertex $v \in V(P)$ the corresponding next vertex $u \in V(G)$ (according to $\prec_G'$) and reduce thus the pattern length by at least one until all vertices of $P$ are processed.
	In particular, we will focus on the critical vertices $v \in C$ of $P$ and select the first vertex from the corresponding segment $\beta(v)$ whenever possible.
    More concretely, for a critical vertex $v \in C$, we say that $X$ fulfills property $L(v)$ if the identified vertex $u \in V(G)$ is leftmost in $\beta(v)$.
    Critical vertices $v$ for which $L(v)$ holds can be seen as the ``inductive steps'' in our inductive construction strategy.
	Observe that this resembles a greedy strategy that, intuitively, should ensure that we can find one representative for every vertex of the pattern before we arrive at the end of $\prec_G'$.
	
	Let $v_1, v_2, \ldots, v_{\Size{C}}$ be the critical vertices ordered as in $\prec_P$, i.e., $v_1 \prec_P v_2 \prec_P \ldots \prec_P v_{\Size{C}}$.
	We start with $v_1$ and consider the segment $\beta(v_1) = s_i \in S$ for some $i \in [\ell]$; recall $\ell = \Size{S}$.
	Recall that $V(G)[\prec_G',s_i]$ denotes the vertices in segment $s_i$ of the linear order $\prec_G'$.
	We have $\Size{V(G)[\prec_G',s_i]} = \Psi(x_i) \geq 1$ as $\prec_G'$ realizes $\Psi$.
	There exists a leftmost vertex $u_1 \in V(G)[\prec_G',s_i]$ with respect to $\prec_G'$, i.e., for every vertex $u'\neq u_1 \in V(G)[\prec_G',s_i]$ it holds $u_1 \prec_G' u'$.
	We add the vertex $u_1$ to the set $X$ and it will represent the first critical vertex $v_1 \in C$ of $P$.
    Note that this implies that we fulfill $L(v_1)$ (in fact, this can be seen as the ``base case'' of our inductive construction strategy).
	Before we turn our attention to the second critical vertex $v_2 \in C$, we first have to find representatives for the non-critical vertices $w \in V(P)$ with $v_1 \prec_P w \prec_P v_2$.
	Recall that these vertices are represented by a count $c$ in the compressed pattern $P^C$.
	We add the first $c$ vertices $u' \in V(G)$ right of $u_1$, i.e., with $u_1 \prec_G' u'$, to $X$.
	With this, we have found a representative for every vertex $w \in V(P)$ with $w \prec_P v_2$.
	
	Towards finding a representative $u_2$ for the second critical vertex $v_2$, let $\beta(v_2) = s_j \in S$ be the segment to which $v_2$ should belong to.
	Moreover, let $s_k$ be the segment to which the vertex $\Succ[\prec_G']{u'}$ belongs, where $u'$ is the so-far last vertex added to $X$ (which could also be $u_1$ if $v_2 = \Succ{\prec_P, v_1}$).
	We distinguish the following cases depending on the relation of $k$ to $j$.
	
	\proofsubparagraph*{Case 1: $\boldsymbol{k < j}$.}
	If $k < j$, then all vertices in $X$ belong to segments that lie left of the next segment $\beta(v_2)$ according to $\beta$.
	In this case, we can proceed as for $v_1$, i.e., advance to the beginning of segment $\beta(v_2)$ and pick its leftmost vertex $u_2$.
	
	\proofsubparagraph*{Case 2: $\boldsymbol{k > j}$.}
	If $k > j$, then the next vertex we could pick, i.e., the vertex $u_2$ that should represent $v_2$, belongs to a segment $s_k$ that is right of $s_j = \beta(v_2)$.
	This can only be the case if $d_S(\beta(v_1), \beta(v_2)) < d_P(v_1, v_2)$, i.e., the segments between and including $\beta(v_1)$ and $\beta(v_2)$ do not contain sufficiently many vertices to represent all vertices $w \in V(P)$ with $v_1 \preceq_P w \preceq_P v_2$; recall that we fulfill $L(v_1)$.
	However, since $D_{u,v}^{\beta} \geq 0$ holds for all pairs $u,v\in C$, this is not possible.
	In particular, observe that also $D_{v_1,v_2}^{\beta} \geq 0$ holds, which implies $d_S(\beta(v_1), \beta(v_2)) \geq d_P(v_1, v_2)$.
	
	\proofsubparagraph*{Case 3: $\boldsymbol{k = j}$.}
	If $k = j$, then the next vertex we are able to pick is already from the correct segment $\beta(v_2)$.
	While this last case seems easy at first glance, we still have to argue that over the course of creating the set $X$, there remain enough vertices to find a representative for every vertex in $P$.
	In particular, observe that until now, for every critical vertex of $P$, we always selected the first vertex from the corresponding segment; recall also Case~1 and our argument in Case~2.
	In Case~3, we can no longer ensure this. 
	Instead, we proceed as follows.
	Let $u_2 \in V(G)$ be the next vertex we are able to pick, which belongs to segment $s_k = s_j = \beta(v_2)$ by the assumptions that lead to this case.
	We add $u_2$ to $X$ and proceed as before, i.e., consider the count $c'$ between $v_2$ and $v_3$ in $P^C$ and also add the next $c'$-many vertices to $X$.
	Afterwards, we repeat this process for $v_3,v_4, \ldots$ until we are either at $v_{\Size{C}}$ or no longer in Case~3.
	In the former case, we stop, since we have found a representative for every vertex.
	We argue in the end why this means that $(\prec_G', X)$ realizes $P$.
	For the latter case, let $v_q \in C$ be the next critical vertex for which we should find a representative $u_q$.
    Moreover, let $\beta(v_q) = s_{j'} \in S$ be the segment from which we should pick the next vertex according to $\beta$, and let $s_{k'} \in S$ be the segment from which we could pick the next vertex.
	Since we are no longer in Case~3 (with respect to $s_{j'}$ and $s_{k'}$), we have either $k' < j'$ or $k' > j'$.
	Both cases can be handled similarly to Case~1 and~2, respectively.
	For the case $k' < j'$, we advance to the segment $s_{j'}$, which is right of the segment $s_{k'}$ and proceed as in Case~1, i.e., pick the leftmost vertex from the segment as our $u_q$.
    Observe that this implies that we fulfill $L(v_q)$.
	To see that the case $k' > j'$ is not possible, we compare $d_S(\beta(v_1), \beta(v_q))$ and $d_P(v_1, v_q)$.
	Observe that up until now, the vertices from $X$ correspond to a contiguous interval in $\prec_G'$ of length $d_P(v_1, v_q) - 1$.
	This is because we always remained in Case~3 by assumption, i.e., picked the next vertices from $\prec_G'$ as the pattern dictated us without omitting one.
	Moreover, since $D_{v_1, v_q}^\beta \geq 0$ by assumption, we know that $d_P(v_1, v_q) \leq d_S(\beta(v_1), \beta(v_q))$, i.e., the segments between and including $\beta(v_1)$ and $\beta(v_q)$ contain enough vertices to find a representative for every vertex $w \in V(P)$ with $v_1 \preceq_P w \preceq_P v_q$, in particular for $v_q$; recall that we fulfill $L(v_1)$.
	Thus, we conclude that $k'> j'$ is not possible; compare also to Case~2 above.
	
	Combining all, we derive that each case is either not possible or we can find representatives for critical vertices and counts in the compressed pattern $P^C$ while maintaining property $L(v_i)$ for a suitable $v_i \in C$.
	Thus, we can repeat the above process until we have found a representative $u_{\Size{C}}$ for the last critical vertex $v_{\Size{C}}$.
	Since, by definition, $v_{\Size{C}}$ is also the rightmost vertex in $\prec_P$, we eventually obtain a set $X$ of size $\Size{X} = n_P$, i.e., one that contains one representative for every vertex of the pattern.
	To see that $(\prec_G', X)$ realizes the pattern~$P$, it is sufficient to observe that $(\prec_G', X)$ is consistent with $\beta$.
    More formally, by analyzing Cases~1 and~3, we note that we selected for every critical vertex $v \in C$ a representative $u$ from the segment $\beta(v)$.
    This holds in particular for every vertex $v \in V(P)$ that is incident to a forced or forbidden edge of $P$.    
	Since $\beta$ is a provoking segment assignment, the respective neighborhood types ensure that $u$ is (not) adjacent to the respective other critical vertices as required in $P$.
	Combining all, the statement follows.
\end{proof}
\Cref{lem:nd-plus-mp-ilp-a-small,lem:nd-plus-mp-ilp-a-large} together show that we prevent the provoking segment assignment $\beta$ if and only if $D_{u,v}^{\beta} < 0$ holds for some pair of critical vertices $u,v\in C$, $u \neq v$.
Hence, we must ensure that at least one variable $x_{u,v}^{\beta}$ can be set to zero, i.e., is not forced to be one by the constraint from \Cref{eq:neighborhood-diversity-pattern-constraint}.
To this end, we add the following constraint, which ensures this:
$\sum_{u\neq v\in C} x_{u,v}^{\beta} \le \binom{\Size{C}}{2} - 1$.

Since we must prevent all possible ways how the pattern $P$ can occur, we repeat above steps for all possible provoking segment assignments $\beta$.
What remains for establishing \Cref{thm:fpt-nd-mp} is observing that the number of variables in our linear program can be bounded by a function in $m_P + \nd(G)$, i.e., our parameter, and combining this with the known result that checking feasibility for an ILP is \FPT\ in the number of variables~\cite{Len.IPF.1983,Kan.MCB.1987}.

\thmfptndmp*
\begin{proof}
    Let $\Instance = \InstanceLong$ be an instance of \PA.
    We first branch to fix the number $\ell$ of segments that we observe in a hypothetical solution and the corresponding tuple $(S, \gamma)$ with $\Size{S} = \ell$.
    By \Cref{cor:nd-plus-mp-alternations-bound}, we have $\ell\leq f(m_P, \nd(G))$.
    Overall, this yields at most $f(m_P, \nd(G))\cdot {f(m_P, \nd(G))}^{{\nd(G)}}$ branches.
    By \Cref{lem:nd-plus-mp-reorder-safe-intervals,cor:nd-plus-mp-alternations-bound}, this branching step is safe.
    
    For each branch, we construct the above-described ILP.
    It contains, on the one hand, one variable $x_i$ for each segment $s_i \in S$ and, on the other hand, $\binom{\Size{C}}{2}$ variables $x_{u,v}^{\beta}$, $u,v\in C$, for every provoking segment assignment $\beta$.
    Recall that $C$ denotes the critical vertices of the pattern $P$.
    There are at most $\ell^{\Size{C}}$ different segment assignments, which also bounds the number of provoking segment assignments.
    Thus, the ILP contains at most $f(m_P, \nd(G)) + \binom{\Size{C}}{2}\cdot f(m_P, \nd(G))^{\Size{C}}$ variables since $\ell \leq f(m_P, \nd(G))$.
    By \Cref{obs:nd-plus-mp-number-critical}, we have $\Size{C} \in \BigO{m_P}$.
    Consequently, the number of variables in the ILP is in $\BigO{\binom{m_P}{2}\cdot f(m_P, \nd(G))^{m_P}}$.
    
    Using \Cref{lem:nd-plus-mp-ilp-a-small,lem:nd-plus-mp-ilp-a-large}, we conclude that $\Instance$ is a positive instance of \PA if and only if there exists one branch where the ILP is feasible.
    Testing the feasibility of an ILP is fixed-parameter tractable in the number of its variables~\cite{Len.IPF.1983,Kan.MCB.1987}, i.e., in our case it is \FPT\ in $m_P + \nd(G)$.
    Together with the above-derived bound on the number of branches, the statement follows.
\end{proof}
It is known that the neighborhood diversity of a graph $G$ can be upper bounded by the vertex cover number $\vc(G)$ of $G$, i.e., we have $\nd(G) \leq 2^{\vc(G)} + \vc(G)$~\cite{Lam.AMt.2011}.
We remark:%
\begin{corollary}
\label{cor:fpt-vc-mp}
    \PA is fixed-parameter tractable with respect to the vertex cover number $\vc(G)$ of the input graph $G$ and the number of edges $m_P$ in the pattern $P$.
\end{corollary}

\section{An Algorithm for Almost All Constant-Sized Patterns on Forests}
\label{sec:poly-trees}
In this section, we establish \Cref{thm:non-mixed}, i.e., provide a polynomial time algorithm for \PA on forests for almost all constant-sized patterns.
At the core of our algorithm lies a dynamic program, targeted at solving structured instances of \PA.
On our quest to obtain these instances, we first consider in \Cref{sec:poly-trees-no-two-forced-edge} patterns with none or at least two forced edges.
So what remains to handle are patterns with precisely one forced edge $E^+(P) = \{uv\}$ with $u \prec_P v$.
We first focus on patterns where $u$ is the leftmost vertex in $\prec_P$.
In \Cref{sec:poly-trees-preprocessing}, we present a number of preprocessing steps that iteratively simplify the instance.
For the simplified instances we devise in \Cref{sec:poly-trees-dp} the aforementioned DP-algorithm.
In \Cref{sec:poly-trees-extension}, we show that the DP-algorithm can be extended to handle all non-mixed patterns.
Finally, we present in \Cref{sec:poly-trees-hardness} an \NP-hardness result that rules out an extension of this dynamic programming approach to mixed patterns.

We recall that for an instance \InstanceLong of \PA, $n_G$ and $n_P$ denote the number of vertices in $G$ and $P$, respectively.
Moreover, we recall $m_P^+ = \Size{E^+(P)}$ and $m_P^- = \Size{E^-(P)}$. 

\subsection{Patterns with Zero or At Least Two Forced Edges}
\label{sec:poly-trees-no-two-forced-edge}
The dynamic program that we devise in a later section implicitly assumes that our pattern contains exactly one forced edge.
In this section, we justify this assumption by exploiting known properties of forests to efficiently handle the remaining cases. 

\subparagraph*{Patterns $\boldsymbol{P}$ with $\boldsymbol{m_P^+} \geq 2$.}
We first consider the case where $P$ contains at least two forced edges, i.e., $m_P^+ = \Size{E^+(P)} \geq 2$.
Note that we do not make any assumptions on the size (or structure) of $E^-(P)$.

\begin{lemma}\label{lem:pattern_deletion_to_bipartite}
    Let the pattern $P$ be given by $( \{a,b,c,d\}, a \prec b \prec c \prec d, \{ab, cd\}, \{\})$, i.e., case (e) of \Cref{fig:patterns-two-forced-edges}.
    Then $\mathcal{C}_P$ is the class of graphs which can be made bipartite by deleting at most one vertex.
\end{lemma}
\begin{proof}
    $(\Rightarrow):$ Let $G \in \mathcal{C}_P$, and let $\prec$ be a linear order of $V(G)$ that certifies this. W.l.o.g., we assume $E(G)$ is non-empty. Let $st \in E(G)$ with $s \prec t$ such that $s$ is maximal with respect to~$\prec$, i.e., $st$ is the edge of $G$ whose left endpoint in $\prec$ is rightmost.
    We claim that $s$ certifies the desired property, i.e., $G - s$ is bipartite (with bipartition $(L, R)$, where $L$ are the vertices to the left of $s$ in $\prec$ and $R$ those to the right of $s$).
    Let $ab \in E(G - s)$ with $a \prec b$. If $s \prec a \prec b$, 
    this contradicts the choice of $s$. If $a \prec b \prec s$, the edges $ab,st$ induce the forbidden pattern~$P$. Thus we have $a \prec s \prec b$, which by construction means $a \in L$ and $b \in R$.

    $(\Leftarrow):$
    Let $G$ be a (w.l.o.g., non-empty) graph and $s \in V(G)$ such that $G - s$ is bipartite with bipartition $(A, B)$ (if $G$ is already bipartite, such $s$ can be chosen arbitrarily).
    Construct the linear order $\prec$ of $V(G)$ by first placing vertices $A$ in any order, then vertex $s$, then vertices $B$ in any order.
    Let the \emph{span} of an edge $ab \in E(G)$ with $a \prec b$ in $\prec$ be the set $\{ v \in V(G) \mid a \preceq v \preceq b \}$.
    Suppose $(G, \prec)$ contains pattern $P$, i.e., there are two edges $ab, cd \in E(G)$ with disjoint spans in $\prec$. 
    But this is a contradiction, as the span of $ab$ (resp.\ $cd$) always contains $s$: if $a,b \neq s$ (resp.\ $c,d \neq s$), each endpoint is in a different set of the bipartition, thus by construction the edge spans $s$, and if otherwise one of $a, b$ (resp.\ $c, d$) equals $s$, the edge spans $s$ trivially. 
\end{proof}

\begin{lemma}
    \label{lem:forests-two-forced-edge-positive-instance}
    Let $G$ be a forest and assume that
    $P$ contains at least two forced edges.
    Then the instance $(G, P)$ of \PA is positive.
\end{lemma}
\begin{proof}
Let $e_1, e_2$ be distinct forced edges of $P$.
Consider the subpattern $P'$ formed from $P$ by removing all edges except $e_1$ and $e_2$, and subsequently deleting all isolated vertices.
Clearly, if $G$ avoids $P'$, it avoids $P$ too.
That is, $\mathcal{C}_{P'} \subseteq \mathcal{C}_{P}$.%
Next, we consider all possible $P'$ (up to reversal of $\prec_{P'}$) and observe that in each case, the input forest $G$ already belongs to the more restrictive class $\mathcal{C}_{P'}$, and hence also to $\mathcal{C}_{P}$.
\Cref{fig:patterns-two-forced-edges} enumerates all possible $P'$. We treat each case separately.  

\begin{description}
    \item[(a)] $\mathcal{C}_{P'}$ is the class of graphs with queue number at most one which contains all forests \cite{DBLP:journals/siamcomp/HeathR92}.
    \item[(b)] $\mathcal{C}_{P'}$ is the class of graphs with book thickness at most one, which are precisely the outerplanar graphs \cite{DBLP:journals/jct/BernhartK79}; forests are outerplanar.
    \item[(c)] $\mathcal{C}_{P'}$ is the class of bipartite graphs \cite{FH.GCF.2021}; forests are bipartite.
    \item[(d)] $\mathcal{C}_{P'}$ is the class of forests \cite{FH.GCF.2021}; $G$ is a forest.
    \item[(e)] $\mathcal{C}_{P'}$ is the class of graphs with vertex deletion distance at most one to bipartite (\Cref{lem:pattern_deletion_to_bipartite}); forests are already bipartite. \qedhere
\end{description}
\end{proof}

\subparagraph*{Patterns $\boldsymbol{P}$ with $\boldsymbol{m_P^+} = 0$.}
Next, we consider the case where $P$ contains no forced edges but arbitrarily many forbidden edges, i.e., $m_P = m_P^-$.
\begin{observation}
    \label{obs:forests-independent-set}
    Let $G$ be a forest on $n$ vertices.
    There exists an independent set $I \subseteq V(G)$ of size at least $\Size{I} \geq \Size{V(G)}/2$.
\end{observation}
\begin{observation}
    \label{obs:forests-independent-set-contains-pattern}
    Let $\Instance = \InstanceLong$ be an instance of \PA where $m_P^+ = 0$ and~$G$ is an independent set on $n_G \geq n_P$ vertices.
    Then, every permutation $\prec_G$ of $V(G)$ contains $P$.
\end{observation}
The following lemma follows from combining \Cref{obs:forests-independent-set,obs:forests-independent-set-contains-pattern}.
\begin{lemma}
    \label{lem:forests-no-forced-edge-no-instance}
    Let $\Instance = \InstanceLong$ be an instance of \PA where $m_P^+ = 0$ and~$G$ is a forest on $n_G \geq 2n_P$ vertices.
    Then, every permutation $\prec_G$ of $V(G)$ contains $P$.
    In particular, \Instance is a negative instance.
\end{lemma}

We now use \Cref{lem:forests-no-forced-edge-no-instance} to devise an algorithm for the case where $m_P^+ = 0$.
\begin{lemma}
    \label{lem:forests-no-forced-edge-algorithm}
    Let $\Instance = \InstanceLong$ be an instance of \PA where $m_P^+ = 0$ and~$G$ is a forest.
    Then, we can decide if $G$ can avoid $P$ in $2^{\BigO{n_P \cdot \log(n_P)}}$ time.
\end{lemma}
\begin{proof}
    If $n_G \geq 2\cdot n_P$ we can return no due to \Cref{lem:forests-no-forced-edge-no-instance}.
    Otherwise, we exhaustively branch into all $\BigO{n_G!}$-many linear orders $\prec_G$ of $V(G)$.
    For every linear order $\prec_G$ we can check in $\BigO{\binom{n_G}{n_P}\cdot {n_P}^2}$ time if it avoids $P$ by examining all possible subsets of $V(G)$ of size~$n_P$.
    Since $n_G < 2 \cdot n_P$, the statement follows.
\end{proof}

\subsection{Patterns with One Forced Edge}
\label{sec:poly-trees-preprocessing}
\Cref{lem:forests-no-forced-edge-no-instance,lem:forests-two-forced-edge-positive-instance} together allow us to focus on the case where $P$ contains a single forced edge, i.e., $m_P^+ = 1$.
Towards this, we first define a special case of \PA\ below; we will later show that solving this special case is sufficient to deal with all non-mixed patterns (cf. Subsection~\ref{sec:poly-trees-extension}).

\probdef{\PAOneFLongUnderline~(\PAOneF)}{A forest $G$ and a pattern $P$ with $E^+(P) = \{uv\}$ such that $u \preceq_P w$ for every $w \in V(P)$.}{Does there exist a total order $\prec_G$ of $V(G)$ that avoids the pattern $P$?}

As our first order of business, we provide a useful observation about detecting patterns and define two notions that will be useful throughout this section.

\begin{observation}
    \label{lem:poly-trees-preprocessing-pattern-on-left}
    Given an instance $\Instance=\InstanceLong$ of \PA and a total order $\prec_G$ of $V(G)$.
    We can find in $\BigO{{n_G}^{n_P} \cdot {n_P}^2}$ time the minimum value $\ell \in [n_G]$ such that $\prec_G\mid_{\ell}$ contains $P$ (if such an $\ell$ exists).
\end{observation}
\begin{proof}
    Let $Y \subseteq V(G)$ be a subset of $G$'s vertices of size $\Size{Y} = n_P$.
    There are $\binom{n_G}{n_P} \in \BigO{{n_G}^{n_P}}$ possible sets $Y$.
    If $(G[Y], \prec_G\mid_Y)$ matches $P$, which we can check in $\BigO{{n_P}^2}$ time, we set $\ell_Y$ to the index of $v$ in $\prec_G$, where $v$ is the rightmost vertex in $\prec_G\mid_Y$.
    Otherwise, we set $\ell_Y$ to infinity.
    The sought-after value $\ell$ is $\ell = \min \{\ell_Y \mid Y \subseteq V(G),\Size{Y} = n_P\}$ (if it is finite). %
\end{proof}

\begin{definition}
    \label{def:poly-trees-preprocessing-induced-pattern}
    Let $P$ be a pattern and let $X \subseteq V(P)$ be a subset of its vertices.
    We define the \emph{induced pattern on $X$} as $P[X] = (X, \prec_P\mid_X, E^+(P) \cap X^2, E^-(P) \cap X^2)$.
\end{definition}
\begin{definition}
    \label{def:poly-trees-preprocessing-left-right-middle}
    Let $P$ be a pattern with $E^+(P) = \{uv\}$, $u \prec_P v$.
    We define $V_L(P) = \{w \in V(P) \mid w \prec_P u\}$, $V_M(P) = \{w \in V(P) \mid u \preceq_P w \preceq_P v\}$, and $V_R(P) = \{w \in V(P) \mid v \prec_P w\}$ to be the vertices left, between (and including), and right of the forced edge, respectively.
    $L(P)\coloneqq P[V_L(P)]$, $M(P) \coloneqq P[V_M(P)]$, and $R(P) \coloneqq P[V_R(P)]$ denote the \emph{left}, \emph{middle}, \emph{right} part of $P$, respectively.
\end{definition}

We also note that high-degree vertices imply the existence of a ``long edge'' in the total order:

\begin{observation}
	\label{obs:poly-trees-preprocessing-high-degree-long-edge}
	Let $v \in V(G)$ be a vertex of degree $x = deg_G(v)$.
	In every linear order $\prec_G$ there exists an edge $uv \in E(G)$ such that the number of vertices $w \in V(G)$ with $u \preceq_G w \preceq_G v$ (or $v \preceq_G w \preceq_G u$) is at least $\lceil x/2 \rceil + 1$.
\end{observation}

\subparagraph*{Hydra Vertices.}
Next, we define the notion of a hydra-vertex, which is a vertex with many non-leaf neighbors in $G$.
As we will see later, these vertices require special treatment in our algorithm and we show that only a bounded number of them can exist in positive instances.
\begin{definition}
    \label{def:poly-trees-preprocessing-hydra-vertex}
    Let $G$ be a forest and let $v \in V(G)$ be a vertex of $G$.
    The \emph{interesting degree} $d_G^{>1}(v)\coloneqq \Size{\{u \in N_G(v) \mid d_G(u) > 1\}}$ is the number of non-leaf neighbors of $v$.
    We call $v$ a \emph{hydra vertex} if $d_G^{>1}(v) \geq 38n_P + 1$.
\end{definition}

As our next step, we show all hydra vertices must be towards the end of every hypothetical solution $\prec_G$.
Eventually, this allows us to branch and fix the last part of $\prec_G$, including the hydra vertices.
Towards establishing this intermediate result, we first show the following lemma.

\begin{lemma}
    \label{lem:poly-trees-preprocessing-matching}
    Let $\Instance= \InstanceLong$ be an instance of \PAOneF.
    Assume there is a solution $\prec_G$ to \Instance and let $M \subseteq E(G)$ be a matching.
    For every vertex $v \in V(G)$ the number of edges $uw \in M$ with $u \prec_G v \prec_G w$ is at most $6n_P$.
\end{lemma}
\begin{proof}
    Assume that there exists a vertex $v \in V(G)$ such that the set $N = \{uw \in M \mid u \prec_G v \prec_G w\}$ is of size $\Size{N} > 6 n_P$.    
    We choose a subset $N' \subseteq N$ such that the graph induced on the endpoints of edges in $N'$ is $N'$ itself (in other words, no two endpoints of $N'$ are connected except through the matching edges $N'$).
    We construct $N'$ from $N$ by contracting the edges from $N$ in the subgraph induced on the endpoints of edges in $N$.
    The resulting graph must be a forest and thus contains an independent set of size $3n_P$.
    We take this independent set and let $N'$ be the edges contracted into them.
    Observe that, therefore, $\Size{N'} \geq 3n_P$.

    Let $V' = \{v\} \cup \bigcup_{uw \in N'} \{u,w\}$ and consider $\prec_G\mid_{V'}$.
    Observe that in $\prec_G\mid_{V'}$ there are at least $3 n_P$ vertices to the left and right of $v$, respectively.
    Let the set $S \subseteq V'$ be composed of the first and last $n_P - 2$ vertices in $\prec_G\mid_{V'}$, respectively, together with the first $\lceil n_P/2\rceil$ vertices that come before and after $v$ in $\prec_G\mid_{V'}$, respectively.
    Observe that $\Size{S} \leq 3n_P - 3$.
    Furthermore, let $M' \subseteq M$ be the edges not incident to a vertex in $S$.
    Observe that $S$ consists of at most $3/2 \cdot n_P - 1$ vertices on each side of $v$ in $\prec_G$.
    Therefore, there exists at least one edge in %
    $N' \cap M'$, i.e., an edge $uw \in E(G)$ such that $u \prec_G v \prec_G w$ and $u,w \notin S$.
    Observe that $S$ contains an independent set $I_L \subseteq S$ to the left of $u$, an independent set $I_R \subseteq S$ right of $w$, and an independent set $I_M \subseteq S$ between $u$ and $w$ such that $\Size{I_L}, \Size{I_R}, \Size{I_M} \geq n_P$; recall the construction of $N'$.
    Moreover, note that $I_L \cup I_R \cup I_M \cup \{u,w\}$ remains an independent set by the construction of $N'$.
    Therefore, $\prec_G$ contains the pattern; recall also \Cref{obs:forests-independent-set-contains-pattern}.
    This is a contradiction to the existence of $\prec_G$.
\end{proof}

We can now establish, via the following two lemmas, that the placement of hydra vertices is nearly fixed and must be towards the right side of $\prec_G$.
To this end, we define for a non-leaf vertex $v \in V(G)$ the set of \emph{critical vertices} $C(v)$ as $C(v) \coloneqq \{u \in V(G) \setminus \{v\} \mid d_G(u) > 1\ \text{or}\ u \notin N_G(v)\}$.
\begin{lemma}
    \label{lem:poly-trees-preprocessing-hydra-not-middle}
    Let $\Instance= \InstanceLong$ be an instance of \PAOneF.
    Assume there is a solution $\prec_G$ to \Instance.
    Let $v \in V(G)$ be a non-leaf vertex of $G$ with $d_G^{>1}(v) \geq 38n_P + 1$.
    Furthermore, let $L_v = \{u \in C(v) \mid u \prec_G v\}$ and $R_v = \{u \in C(v) \mid v \prec_G u\}$.
    Then, $\Size{L_v} \leq 10 n_P$ or $\Size{R_v} \leq 10 n_P$.
\end{lemma}
\begin{proof}

    Assume otherwise and let $M$ be a maximum matching in the graph $G$ induced on the edges $uw \in E(G)$ with $u \prec_G v \prec_G w$.
    By \Cref{lem:poly-trees-preprocessing-matching}, the size of $M$ must be bounded by $6 n_P$.
    Let $G' = G - V(M)$ be the graph obtained from $G$ after removing all endpoints of edges in $M$.
    Moreover, let $\prec_{G'} = \prec_{G}\mid_{V(G')}$ be the solution to \Instance induced on the vertices of $G'$.
    Observe that $v \in V(G')$ and there is no edge $uw \in E(G')$ such that $u \prec_{G'} v \prec_{G'} w$, i.e., no edge that ``spans'' $v$ or, in other words, the left and right side of $v$ with respect to $\prec_{G'}$ are disconnected, due to the maximality of $M$.
    Moreover, observe that $d_{G'}^{>1}(v) \geq 32n_P + 1$ holds, since for every deleted edge, $v$ could only be connected to one endpoint.
    Let $L_v' = L_v \cap V(G')$ and $R_v' = R_v \cap V(G')$ denote the critical vertices left and right of $v$ in $\prec_{G'}$.
    Since $\Size{M} \leq 6 n_P$ and $\Size{R_v},\Size{L_v} > 10 n_P$, we have $\Size{R_v'}, \Size{L_v'} > 4 n_P$. %
    We now show that $\prec_{G'}$ must contain the pattern $P$ and, therefore, also $\prec_G$.
    Without loss of generality, we can assume that there are at least as many non-leaf neighbors of $v$ in $L_v'$ as in $R_v'$, i.e., assume that $\Size{\{u \in L_v'  \cap N_{G'}(v) \mid d_G(u) > 1\}} \geq \Size{\{u \in R_v' \cap N_{G'}(v) \mid d_G(u) > 1\}}$ (otherwise swap the role of $L_v'$ and $R_v'$ in the following arguments).
    Observe that $L_v'$ contains, therefore, at least $16n_P + 1$ non-leaf neighbors of $v$.
    Next, recall that $R_v'$ contains only non-leaf neighbors or non-neighbors of $v$.
    Since no edge spans $v$ in $\prec_{G'}$, we must have $w \in R_v'$ for every neighbor $w \in N_{G'}(u)$ of every non-leaf neighbor $u\in N_{G'}(v)$, $\text{deg}_G(u) > 1$, of $v$.
    Therefore, $R_v'$ contains in particular at least $2n_P$ non-neighbors of $v$, i.e., $\Size{R_v' \setminus N_{G'}(v)} \geq 2n_P$.
    Hence, we can find among $R_v'$ an independent set $I_R$ of size at least $n_P$ such that $I_R \cup \{v\}$ remains an independent set; recall \Cref{obs:forests-independent-set}.

    We now turn our attention to $L_v'$.
    Let $u \in L_v' \cap N_{G'}(v)$ be the $8n_P + 1$-th non-leaf neighbor of $v$ to the left of $v$.
    Moreover, let $N_{u}(v) = \{w \in L_v' \cap N_{G'}(v) \mid w \prec_{G'} u\}$ and $N_{uv}(v) = \{w \in L_v' \cap N_{G'}(v) \mid u \prec_{G'} w\}$ be the non-leaf neighbors of $v$ that are, in $\prec_{G'}$, left of $u$ and between $u$ and $v$, respectively.
    We have $\Size{N_u(v)}, \Size{N_{uv}(v)}\geq 8n_P$ by a similar reason to why \Cref{obs:poly-trees-preprocessing-high-degree-long-edge} holds.
    
    We now aim to construct two further independent sets $I_u$ and $I_{uv}$ consisting of $n_P$ vertices from $N_{u}(v)$ and $N_{uv}(v)$, respectively.
    These independent sets will serve as our witness that $\prec_{G'}$ contains $P$.
    To this end, let $M'$ be a maximum matching in the graph $G'$ induced on the edges $ab \in E(G')$ with $a \prec_{G'} u \prec_{G'} b$.
    As before, the size of $M'$ must be bounded by $6 n_P$ due to \Cref{lem:poly-trees-preprocessing-matching}.
    Let $G'' = G' - V(M')$ be the graph obtained from $G'$ after removing all endpoints of edges in $M'$.
    The number of non-leaf neighbors $N_u'(v)$ and $N_{uv}'(v)$ of $v$ in $G''$ that are in $N_u(v)$ and $N_{uv}(v)$, respectively, remains at least $2n_P$.
    Observe that no more edge spans $u$ or $v$, $N_u'(v)$ and $N_{uv}'(v)$ are non-leaves, and all vertices in $\{u\} \cup N_u'(v) \cup N_{uv}'(v)$ are adjacent to $v$.
    Therefore, we can find two independent sets $I_u \subseteq N_u(v)$ and $I_{uv} \subseteq N_{uv}(v)$ such that $I_R \cup I_u \cup I_{uv} \cup \{u,v\}$ remains an independent set; otherwise $G$ would not be a forest.
    Thus, these vertices witness together with the edge $uv$ that $\prec_{G'}$ (and therefore also $\prec_G$) contains $P$ (recall also \Cref{obs:forests-independent-set-contains-pattern}).
    This is a contradiction to $\prec_G$ being a solution to \Instance.
    Thus, $\Size{L_v} \leq 7n_P$ or $\Size{R_v} \leq 10n_P$ must hold. %
\end{proof}
\begin{lemma}
    \label{lem:poly-trees-preprocessing-hydra-not-left}
    Let $\Instance= \InstanceLong$ be an instance of \PAOneF.
    Assume there is a solution $\prec_G$ to \Instance and let $v \in V(G)$ be a non-leaf vertex with $d_G^{>1}(v) \geq 22n_P + 1$.
    Moreover, let $L_v = \{u \in C(v) \mid u \prec_G v\}$.
    Then $\Size{L_v} > 10 n_P$.
\end{lemma}
\begin{proof}
For the sake of a contradiction, assume otherwise, i.e., $\Size{L_v} \leq 10n_P$.
This implies that with respect to $C(v)$, i.e., ignoring leaf-neighbors of $v$, $v$ is among the $10n_P + 1$ leftmost vertices in $\prec_G$.
Let $G'$ be the subgraph of $G$ obtained as follows.
First, for every neighbor $u \in N_G(v)$ with $d_G(u)>1$, let $w \in N_G(u) \setminus \{v\}$ be an arbitrary neighbor.
Remove $N_G(u) \setminus \{v,w\}$ from $G$, i.e., afterwards $N_{G'}(u) = \{v,w\}$ holds.
Next, remove any vertex leaf-neighbor of $v$, i.e., every vertex $u \in N_G(v)$ with $d_G(u) = 1$.
Observe that afterwards $d_{G'}^{>1}(v) = d_{G'}(v)$ holds.
Finally, remove every remaining vertex $w$ of distance greater than $2$ to $v$, i.e., every vertex $w \notin N_{G'}(v)$ that is not adjacent to a neighbor of $v$.
Note that in the end, $G'$ resembles a star of degree $d_{G}^{>1}(v)$ with center $v$, where every edge is subdivided exactly once, i.e., $G'$ consists of $d_{G}^{>1}(v)$-many ``arms'' $(v,u,w)$ such that $vu,uw \in E(G')$.
Note that this process only shrinks the size of $L_v$.
Let $\prec_{G'} = \prec_G\mid_{V(G')}$ be the resulting linear order.
We now argue that $\prec_{G'}$ (and therefore also $\prec_G$) contains $P$.
Let $(v,u,w)$ be an ``arm'' of $G'$.
We remove $u$ and $w$ from $G'$ (and $\prec_{G'}$) if $u \prec_{G'} v$ or $w \prec_{G'} v$.
Let $G''$ and $\prec_{G''}$ be the resulting graph and its linear order, respectively.
Observe that $v$ is now the leftmost vertex in $\prec_{G''}$.
Let $F = \{uw \in E(G'') \mid u,w \neq v, v \prec_{G''} u,w\}$ denote the set of edges to the right of $v$ that are not incident to $v$.
As $\Size{L_v} \leq 10n_P$, we have that $\deg_{G''}^{>1}(v) \geq \deg_{G}^{>1}(v) - 10n_P \geq 12n_P+1$.
In particular, $\Size{F} \geq 12n_P+1$ edges of the form $uw$, $u,w \neq v$, remain.
In the following, we \emph{mark} vertices and edges in $G''$ based on their position in $\prec_{G''}$.
More concretely, we mark the first $2n_P$ vertices $v \prec_{G''} u$ right of $v$, their neighbor $w$ (if $u$ has two neighbors, then we mark the neighbor $w \neq v$), and its connecting edge $uw$.
Let $M_L$ denote the leftmost $2n_P$ marked vertices.
Likewise, we mark the $2n_P$ rightmost vertices $u$ in $\prec_{G''}$, their neighbor $w \neq v$, and its connecting edge $uw$.
Let $M_R$ denote the $2n_P$ rightmost vertices marked in the above process.
We observe that at least $\deg_{G}^{>1}(v) - 10n_P - 4n_P \geq 8n_P + 1$ unmarked edges in $F$ remain.
Next, we distinguish between the following two cases.

In the first case, we assume that there exists an unmarked edge $uw \in F$ that spans at least $2n_P$ vertices, i.e., we have $V_{uw} = \{x \in V(G'') \mid u \prec_{G''} x \prec_{G''} w\}$ and $\Size{V_{uw}} \geq 2n_P$
As $G''$ is a star with $v$ at its center where every edge is subdivided exactly once, and we have $v \prec_{G''} u,w$, there cannot exist an edge of the form $ux$ or $wx$ for any $x \in V_{uw}$.
The same applies for the $2n_P$ rightmost vertices $x \in M_R$.
Observe that, since $uw$ is unmarked, we have that $u,w \prec_{G''} x$, where $x \in M_R$.
Thus, we can find two independent sets $I_{uw} \subseteq V_{uw}$ and $I_R \subseteq M_R$ of size at least $n_P$; recall \Cref{obs:forests-independent-set}.
Moreover, $I_{uw} \cup I_R \cup \{u,w\}$ remains an independent set.
As the leftmost vertex in $\prec_P$ is incident to the forced edge, these two independent sets, together with the edge $uw$, witness that $\prec_{G''}$ (and, therefore, also $\prec_G$) contain the pattern $P$.

In the second case, we assume that no such unmarked edge exists.
In particular, all unmarked edges $uw \in F$ are short, i.e., span less than $2n_P$ vertices.
Moreover, observe that, since $G''$ is a star with center $v \neq u,w$, they form a matching in $G''$.
We will now show that there exists a neighbor $u^* \in N_{G''}(v)$ such that we can find the pattern $P$ in $\prec_{G''}$ where the edge $u^*v$ takes over the role of the (single) forced edge in $P$.
Towards identifying this vertex $u^*$, we continue with our marking-process.
In particular, we mark the next $4n_P$ vertices $u$ right of $v$ that have not been marked so far.
Let these be the vertices $M_L'$.
As before, we also mark their neighbor $w \neq v$ and their connecting edge $uw$.
Observe that at least $\deg_{G}^{>1}(v) - 10n_P - 4n_P - 4n_P \geq 4n_P + 1$ unmarked edges in $F$ remain.
Since every edge ``spans'' at most $2n_P$ vertices, we have marked at least $2n_P$ vertices with $u \notin N_{G''}(v)$.
Likewise, we mark the $4n_P$ rightmost vertices $u \in V(G'')$ in $\prec_{G''}$ that have not been marked before.
Let these vertices be $M_R'$.
As before, we mark their neighbor $w \neq v$, and their connecting edge $uw$.
We can make a similar calculation as above to observe that $\deg_{G}^{>1}(v) - 10n_P - 4n_P - 4n_P - 4n_P \geq 1$ unmarked edge(s) remain(s) in $F$.

Let $u^*$ be the first unmarked neighbor right of $v$.
As there exists at least one unmarked edge in $F$, the vertex $u^*$ must exist.
Recall that $M_L'$ contains at least $2n_P$ non-neighbors of $v$.
Hence, the edge $vu^*$ ``spans'' over at least $2n_P$ non-neighbors of $v$ and, by similar arguments, of $u^*$ (otherwise $u^*$ would be marked).
Let these vertices be $V_{vu^*}$.
Furthermore, since $u^*$ is unmarked, the set $M_R'$ witnesses that there are at least $2n_P$ vertices $x$ such that $u^* \prec_{G''} x$ and $x \notin N_{G''}(v) \cup N_{G''}(u^*)$.
Let $V_{u^*}$ denote the set of these vertices.
Similar to the first case, we can find two independent sets $I_{vu^*} \subseteq V_{vu^*}$ and $I_{u^*} \subseteq V_{u^*}$ of size at least $n_P$ such that $I_{vu^*} \cup I_{u^*} \cup \{v, u^*\}$ remains an independent set.
Hence, the, sufficiently long, edge $vu^*$, together with these vertices witness that $\prec_{G''}$ (and, therefore, also $\prec_G$) contain the pattern $P$.
As both cases lead to a contradiction, we conclude that $|L_v| > 10n_P$ must hold.
\end{proof}

Note that \Cref{lem:poly-trees-preprocessing-hydra-not-left,lem:poly-trees-preprocessing-hydra-not-middle} apply to all hydra vertices in $G$ since their interesting degree is at least $38n_P + 1$ by definition.
The following corollary summarizes the consequences of \Cref{lem:poly-trees-preprocessing-hydra-not-left,lem:poly-trees-preprocessing-hydra-not-middle} for hydra vertices:
\begin{corollary}
    \label{cor:poly-trees-preprocessing-hydra-vertex-placement}
    Let $\Instance= \InstanceLong$ be an instance of \PAOneF.
    If there exists a hydra vertex $v \in V(G)$, then in every solution $\prec_G$ to \Instance we have $\Size{\{p \in C(v) \mid v \prec_G p\}} \leq 10 n_P$.
\end{corollary}

\subparagraph*{Fixing the Right Side of a Solution.}
In essence, \Cref{cor:poly-trees-preprocessing-hydra-vertex-placement} shows that no hydra vertex can be ``in the left or middle segment'' of a potential solution $\prec_G$.
In particular, every hydra vertex $v$ has to be among the rightmost $10 n_P$ vertices of $\prec_G$, when ignoring non-critical vertices $u \notin C(v)$.
Therefore, we will from now on consider a slightly different problem, called \PAOneFOneSLong in which we force the sought-after linear order $\prec_G$ to end with a particular linear order.
Formally, we define it as follows: %
\probdef{\PAOneFOneSLongUnderline\\(\PAOneFOneS)}{An instance $\InstanceLong$ of \PAOneF, a set $U \subseteq V(G)$, and a total order $\prec_{U}$ on $U$.}{Does there exist a total order $\prec_G$ of $V(G)$ ending with $\prec_{U}$ that avoids the pattern $P$?}
For an instance \InstanceLongOneS of \PAOneFOneS, we also refer to $U$ and $V(G) \setminus U$ as the \emph{fixed} and \emph{free} part (or vertices) of the instance, respectively.
As our next step, we aim at reducing \PAOneF to (boundedly-many) instances of \PAOneFOneS where, due to the above conclusion, the fixed part must contain all hydra vertices of the instance.

Towards this, let us define the type of a vertex in $G$.
\begin{definition}
    \label{lem:poly-trees-preprocessing-type}
    Let $G$ be a forest and $u,v \in V(G)$ be two vertices of $G$.
    We say that $u$ and $v$ are of the same \emph{type} in $G$, denoted as $u \sim_G v$ if and only if
    \begin{itemize}
        \item $u = v$, or
        \item $u$ and $v$ are leaves and $N_G(u) = N_G(v)$, i.e., $u$ and $v$ are siblings.
    \end{itemize}
\end{definition}
Observe that $\sim_G$ defines an equivalence relation on the vertices of $G$.
Recall that $\eqclasses{V(G)}{\sim_G}$ denotes the family of equivalence classes and $[v]_{\sim_G}$ denotes the equivalence class of a vertex $v \in V(G)$.

Next, we define the multisets in the context of this paper.
\begin{definition}
    \label{lem:poly-trees-preprocessing-multiset}
    Let $G$ be a forest.
    A \emph{multiset} $\tau\colon \eqclasses{V(G)}{\sim_G} \to \mathbb{N}_0$ is a mapping from the equivalence classes to the non-negative integers such that $\tau(t) \leq \Size{t}$ for every $t \in \eqclasses{V(G)}{\sim_G}$.
    Given two multisets $\tau_1, \tau_2$, we define
    \begin{itemize}
        \item $\tau = \tau_1 + \tau_2$ as $\tau(t) = \tau_1(t) + \tau_2(t)$ for every $t \in \eqclasses{V(G)}{\sim_G}$,
        \item $\tau = \tau_1 - \tau_2$ as $\tau(t) = \max(\tau_1(t) - \tau_2(t), 0)$ for every $t \in \eqclasses{V(G)}{\sim_G}$,
        \item $\tau_1 \subseteq \tau_2$ if and only if $\tau_1(t) \leq \tau_2(t)$ for every $t \in \eqclasses{V(G)}{\sim_G}$,
        \item $\multisetsize{\tau_1} \coloneqq \Size{\{t \in \eqclasses{V(G)}{\sim_G} \mid \tau_1(t) > 0\}}$, and
        \item $\sum\tau_1 \coloneqq \sum_{t \in \eqclasses{V(G)}{\sim_G}}\tau_1(t)$.
    \end{itemize}
    For a vertex $v \in V(G)$ and multiset $\tau$, we say $v \in \tau$ holds if and only if $\tau([v]_{\sim_G}) > 0$.
\end{definition}

Finally, we define the function $\multiset{\cdot}$, which allows us to obtain a multiset from various representations.
\begin{definition}
    \label{lem:poly-trees-preprocessing-to-multiset}
    Let $G$ be a forest.    
    We define $\multiset{\cdot}$ as follows.
    \begin{itemize}
        \item For a set $U \subseteq V(G)$, $\multiset{U}$ denotes the multiset $\tau$ with $\tau(t) = \Size{U \cap t}$, for every $t \in \eqclasses{V(G)}{\sim_G}$, i.e., $\tau(t)$ corresponds to the number of vertices in $U$ of type $t$.
        \item For a linear order $\prec_U$ on some $U \subseteq V(G)$, $\multiset{\prec_U}$ is  defined as $\multiset{V(\prec)}$.
        \item For a sequence $S$ of types $\eqclasses{V(G)}{\sim_G}$, $\multiset{S}$ denotes the multiset $\tau$  where $\tau(t)$ corresponds to the number of occurrences of $t$ in $S$ for every $t \in \eqclasses{V(G)}{\sim_G}$.
    \end{itemize}    
\end{definition}

Next, we define some operations on sequences $S$ of types from $\eqclasses{V(G)}{\sim_G}$.
\begin{definition}
    \label{lem:poly-trees-preprocessing-type-sequences}
    Let $G$ be a forest and $S$ a sequence of $\eqclasses{V(G)}{\sim_G}$.
    We let $\compress{S}$ denote the number of type-alternations in $S$ plus one.
    Moreover, let $\prec$ be a linear order on $U \subseteq V(G)$.
    We let $S(\prec)$ denote the sequence of $\eqclasses{V(G)}{\sim_G}$ obtained by replacing each vertex in $\prec$ with its type.
    A linear order $\prec$ is \emph{compatible} with a sequence $S'$ if $S(\prec) = S'$.
    Note that two sequences $\prec$ and $\prec'$ with $S(\prec) = S(\prec')$ only differ in (the order of) leaves of the same type. %
\end{definition}
We observe that $\multisetsize{\multiset{S}} \leq \compress{S} \leq \sum \multiset{S}$ holds for every sequence $S$ on $\eqclasses{V(G)}{\sim_G}$.
Moreover, $\compress{S}$ can be interpreted as the size of a run-length encoding of $S$.
In the upcoming description of our DP and its preprocessing steps, we will ensure that for every sequence $S$ that we encounter, we have that $\compress{S}$ is bounded in ${n_G}^{f(n_P)}$ for some function $f(\cdot)$.

With the notion of sequences at hand, we can now finally state the reduction to \PAOneFOneS.
\begin{lemma}
    \label{lem:poly-trees-preprocessing-one-sided-pattern-avoidance}
    Let $\Instance= \InstanceLong$ be an instance of \PAOneF.
    There exists a family $\mathcal{F}$ of instances of \PAOneFOneS such that 
    \begin{enumerate}
    	\item we can construct $\mathcal{F}$ in ${n_G}^{\BigO{n_P}}$ time,
        \item $\Size{\mathcal{F}} \in {n_G}^{\BigO{n_P}}$,
        \item for every hydra vertex $v \in V(G)$ and every $(G, P, U', \prec_{U'}) \in \mathcal{F}$ we have $v \in U'$, 
        \item for every $(G, P, U', \prec_{U'}) \in \mathcal{F}$ we have $\compress{S(\prec_{U'})} = 20n_P + 2$, and
        \item \Instance is a positive instance if and only if some instance $\Instance' \in \mathcal{F}$ is a positive instance.
    \end{enumerate}
\end{lemma}
\begin{proof}
    Let $H \subseteq V(G)$ be the set of hydra vertices in $G$.
    We can compute $H$ in $\BigO{{n_G}^2}$ time.
    From \Cref{cor:poly-trees-preprocessing-hydra-vertex-placement}, we can derive that every vertex $v \in H$ must be among the $10n_P+1$ rightmost vertices in a hypothetical solution $\prec_G$, after ignoring non-critical vertices $u \not\in C(v)$.
    By the definition of critical vertices, in particular due to their high degree, for every two hydra vertices $u,v\in H$, we have that $u \in C(v)$ and $v \in C(u)$.
    Therefore, the number of hydra vertices, i.e., the size of $H$ must be bounded by $10n_P + 1$ in every positive instance.

    Towards constructing the family $\mathcal{F}$ of instances $\Instance'$ of \PAOneFOneS, we generate all sequences $S$ of types $\eqclasses{V(G)}{\sim_G}$ of length $\compress{S} = 20n_P + 2$.
    Our goal is to let $S$ ``represent'' the fixed part of an instance $\Instance'$.
    Thus, if $v\notin \multiset{S}$ for some hydra vertex $v \in H$, the resulting instance would be a negative instance by \Cref{cor:poly-trees-preprocessing-hydra-vertex-placement}.
    To see this, recall that all hydra vertices have to be among the $10n_P+2$ rightmost vertices of a solution (after ignoring the respective non-critical vertices), and only $\compress{S}/2$ entries of $S$ can be a leaf-type of a specific hydra vertex $v \in H$.
    As an immediate consequence, we can discard such sequences and create one instance $\Instance_S = (G, P, V(\prec_S), \prec_{S})$ for the remaining sequences $S$, where $\prec_S$ is a linear order compatible with $S$.

    Finally, we argue that the obtained family $\mathcal{F}$ fulfills the properties of the statement.
    To this end, there are $\BigO{{n_G}^2}$ different tuples $(v,i)$ and $S$ contains $\BigO{n_P}$ of them.
    Thus, there are ${n_G}^{\BigO{n_P}}$ possible sequences $S$, which already gives a bound on $\Size{\mathcal{F}}$.
    For each sequence $S$, the instance $\Instance_S$ can be constructed in time polynomial in $n_G$.
    Thus, we can construct $\mathcal{F}$ in ${n_G}^{\BigO{n_P}}$ time.
    Properties~3 and~4 are true by construction and for Property~5, we make the following observation.
    If \Instance is a positive instance, then the witnessing order $\prec_G$ must have all hydra vertices among its rightmost $10n_P + 1$ vertices; recall \Cref{cor:poly-trees-preprocessing-hydra-vertex-placement}.
    Thus, there exists at least one sequence $S$ such that the ending of $\prec_G$ is compatible with $S$; observe that that leaves of the same type are interchangeable, i.e., their exact representative in $\prec_S$ is irrelevant. 
    Consequently, $\Instance_S \in \mathcal{F}$ is a positive instance.
    For the other direction, observe that every instance $\Instance_S = (G, P, V(\prec_S), \prec_{S})$ is defined on $G$ and $P$.
    Thus, if $\Instance_S$ is a positive instance, so is clearly \Instance.
\end{proof}

With \Cref{lem:poly-trees-preprocessing-one-sided-pattern-avoidance} at hand, we proceed in the following fashion.
First, we branch to determine a sufficient number of vertices that is part of the right fixed side, i.e., in $U$, such that $H \subseteq U$, i.e., it contains every hydra vertex.
Afterwards, we expand upon this initial guess and fix further parts of the instance until for every subset $V'$ of $n_P$ vertices from $V(G) \setminus U$ there exists an independent set $I \subseteq U$ such that there is no edge $uv \in E(G)$ with $u \in V'$ and $v \in I$.
We will ensure that $c(U)$ depends on the size of the pattern, i.e., $n_P$.
To avoid having to deal with a potentially unbounded number of leaves in $G$ (whose exact identity is irrelevant for the existence of a solution), we will not branch upon individual leaves but rather fix a sequence $S$ such that $\prec_U$ is compatible with $S$.
By doing so, we eventually obtain an instance of \PAOneFOneS, for which we will present in \Cref{sec:poly-trees-dp} a dynamic program.

Next, we present some lemmas that can be seen as decomposition procedures which simplify the instance.
\begin{lemma}
    \label{lem:poly-trees-preprocessing-separation-isolated-components}
    Let $\Instance = \InstanceLongOneS$ be an instance of \PAOneFOneS such that $\Size{U} \geq 2 n_P$.
    Moreover, let $\mathcal{T}=\{T_1, \ldots, T_k\}$ be the set of connected components, i.e., trees, of $G$ such that $V(T_i) \cap U = \emptyset$ for all $T_i \in \mathcal{T}$.
    Let $\Instance_{G- \mathcal{T}} = (G - \bigcup_{T \in \mathcal{T}} V(T), P, U, \prec_U)$ and $\Instance_T = (G[T], M(P))$, for $T \in \mathcal{T}$, be instances of \probname{(1S)\PAOneF}.
    
    Then, \Instance is a positive instance if and only if $\Instance_{G- \mathcal{T}}$ is a positive instance and $\Instance_T$ is a positive instance for every $T \in \mathcal{T}$.
    
    In particular, for an instance $(G, P, U, \prec_U)$ of \PAOneFOneS with $U = \emptyset$, then $G$ is a tree and the single forced edge in $P$ is between the leftmost and rightmost vertex of $\prec_P$.
\end{lemma}
\begin{proof}
    We show both directions separately.

    \proofsubparagraph{($\Rightarrow$)}
    Assume that \Instance admits a solution $\prec_G$.
    Clearly, there exists a solution $\prec'$ for $G - \bigcup_{T \in \mathcal{T}} V(T)$, i.e., $\Instance_{G-\mathcal{T}}$ is a positive instance; recall that $V(T) \cap U = \emptyset$ for every $T \in \mathcal{T}$, and $\Instance$ and $\Instance_{G-\mathcal{T}}$ have the same pattern $P$ as input.
    We now show that every instance $\Instance_T$ is also a positive instance.
    In particular, we claim that the linear order $\prec_G\mid_{V(T)}$ is a solution.
    Towards a contradiction, assume that there exists a tree $T \in \mathcal{T}$ such that $\prec_G\mid_{V(T)}$ contains the pattern $M(P)$.
    Recall that $M(P)$ is the pattern $P$ induced on the vertices under (and including the endpoints of) the single forced edge $uv \in E^+(P)$ with $u \prec_P v$.
    Let $(\prec_G, X_T)$ be the realization of $M(P)$ in $\prec_G$ using vertices $X_T \subseteq V(T)$ of $T$ exclusively.
    Note that $(\prec_G, X_T)$ must exist as $\prec_G\mid_{V(T)}$ contains $M(P)$.
    By the definition of $\mathcal{T}$, there is no edge $ab \in E(G)$ with $a \in V(T)$ and $b \in U$. 
    Moreover, $\Size{U} \geq 2n_P$ and, since $G[U]$ is a forest, contains an independent set $I_U$ of size at least $n_P$ by \Cref{obs:forests-independent-set}.
    As $u$ is the leftmost vertex in $\prec_P$ by the definition of \Instance, $P$ and $M(P)$ only differ in the vertices $V_R(P)$, i.e., right of $v \in V(P)$.
    Note that the vertices $V_R(P)$ can only be incident to forbidden edges from $E^-(P)$ since $m_P^+ = 1$.
    Consequently, there exists a set $X \subseteq X_T \cup I_U$ such that $(\prec_G, X)$ is a realization of $P$ in $\prec_G$; recall \Cref{obs:forests-independent-set-contains-pattern}.
    This contradicts the assumption that $\prec_G$ is a solution of \Instance.
    Thus, $\prec_G\mid_{V(T)}$ must be a solution for $(G[T], M(P))$, i.e., $\Instance_T$ is a positive instance.
    
    \proofsubparagraph{($\Leftarrow$)}
    Assume that every instance $\Instance_{T}$ for $T \in \mathcal{T}$ admits a solution $\prec_{T}$ as well as the instance $\Instance_{G- \mathcal{T}}$, where $\prec_{G-\mathcal{T}}$ is a witnessing linear order for the latter instance.
    In the following, we construct a linear order $\prec_G$ that serves as witness that also $\Instance$ is a positive instance.
    To this end, consider an arbitrary indexation of $\mathcal{T}$, i.e., $\mathcal{T} = \{T_1, \ldots, T_k\}$.
    We construct the linear order $\prec_G$ by concatenating the solutions of the created instances, i.e., $\prec_G = \prec_{T_1} \oplus \prec_{T_2} \oplus\ \ldots\ \oplus \prec_{T_k} \oplus \prec_{G-\mathcal{T}}$.
    Observe that $\prec_G$ ends in $\prec_U$ as $\prec_{G-\mathcal{T}}$ ends in $\prec_U$ by the definition of $\Instance_{G- \mathcal{T}}$.
    Therefore, it remains to show that $\prec_G$ avoids the pattern.
    Towards a contradiction, assume that this is not the case.
    This means that there exists a set $X \subseteq V(G)$ such that $(\prec_G, X)$ is a realization of $P$.
    Consider the (single) forced edge $uv \in E^+(P)$ and the corresponding vertices $\delta^{-1}(u), \delta^{-1}(v) \in X$
    As $uv \in E^+(P)$ and $(\prec_G, X)$ is a realization of $P$, we have $\delta^{-1}(u)\delta^{-1}(v) \in E(G)$.
    Thus, we either have $\delta^{-1}(u), \delta^{-1}(v) \in V(T)$, for some tree $T \in \mathcal{T}$, or $\delta^{-1}(u), \delta^{-1}(v) \in V(G) \setminus \bigcup_{T \in \mathcal{T}} V(T)$.
    We first show that the former is not possible and then focus on the latter case.
    Assume the former case, i.e., there exists a tree $T \in \mathcal{T}$ such that $\delta^{-1}(u), \delta^{-1}(v) \in V(T)$.
    This implies that for every $w \in V(P)$ such that $u \preceq_P w \preceq_P v$ we also have $\delta^{-1}(w) \in V(T)$.
    This is because $(\prec_G, X)$ is a realization of $P$ and $\prec_T$ forms a contiguous subsequence of $\prec_G$.
    However, this implies that $(\prec_T, X \cap \{\delta^{-1}(w) \mid w \in M(P)\})$ is a realization of $M(P)$ in $\prec_T$ using only vertices $V(T)$.
    This contradicts the assumption that $\prec_T$ is a solution to $\Instance_T$.
    Therefore, $\delta^{-1}(u), \delta^{-1}(v) \in V(G) \setminus \bigcup_{T \in \mathcal{T}} V(T)$ must hold.
    In particular, $\delta^{-1}(u), \delta^{-1}(v) \in V(T')$ where $T'$ is a connected component, i.e., tree, of $G$ with $T' \notin \mathcal{T}$.
    We claim that also this is not possible.

    To see this, we argue similar to before.
    In particular, since $\delta^{-1}(u) \in V(G) \setminus \bigcup_{T \in \mathcal{T}} V(T)$, we have $\delta^{-1}(w) \in V(G) \setminus \bigcup_{T \in \mathcal{T}} V(T)$ for every $w \in V(P)$ such that $u \prec_P w$.
    To see this, recall that the linear order $\prec_G$ ends with $\prec_{G - \mathcal{T}}$ and $u$ is the leftmost vertex in $\prec_P$.
    Consequently, $(\prec_{G-\mathcal{T}}, X)$ must also be a realization of $P$.
    However, this contradicts the assumption that $(\prec_{G-\mathcal{T}}, X)$ is a solution to $\Instance_{G - \mathcal{T}}$.
    Combining all, we conclude that $(\prec_G, X)$ cannot be a realization of $P$ and, since $X$ was selected arbitrarily, $\prec_G$ avoids the pattern, i.e., is a solution to \Instance.
\end{proof}

\Cref{lem:poly-trees-preprocessing-separation-isolated-components} can be seen as a reduction rule that, in essence, removes trees that do not interact with the fixed right side and treats them separately.
Moreover, we could observe that the number of hydra vertices must be bounded in positive instances (which is also a consequence of \Cref{cor:poly-trees-preprocessing-hydra-vertex-placement}).
However, this does not yet yield a bound on the number of (non-leaf) neighbors of hydra vertices.
As our next goal, our aim is to establish a similar reduction rule with respect to certain non-leaf neighbors of hydra vertices, which we call hydra attachments and are defined below.
\begin{definition}
    \label{def:poly-trees-preprocessing-hydra-attachments}
    Let $G$ be a tree and let $H \subseteq V(G)$ be the hydra vertices of $G$.
    Moreover, let $G' = G - H$ be the forest obtained from $G$ after deleting all hydra vertices.
    An \emph{$H$-bridge} in $G'$ is a connected component $T$ of $G'$ such that there exist vertices $u, u' \in V(T)$ and two vertices $v \neq v' \in H$ with $uv, u'v' \in E(G)$.
    Let $G''$ be the graph obtained from $G'$ after deleting all $H$-bridges.
    The \emph{hydra attachments} of $G$ are the connected components of $G''$ of size at least 2.
    We say that $T$ is a hydra attachment $T$ \emph{of} the hydra vertex $v \in H$ (or it \emph{belongs} to $v$) if there is an edge $uv \in E(G)$ with $u \in V(T)$.
\end{definition}

\begin{definition}
    \label{def:poly-trees-preprocessing-focused}
    Let $\Instance = \InstanceLongOneS$ be an instance of \PAOneFOneS.
    Furthermore, let $v \in H$ be a hydra vertex and let $B$ be a hydra attachment of $v$.
    Let $K = \{k_1, \ldots, k_{n_P}\}$ be $n_P$ new isolated vertices and $\prec_K$ an arbitrary total order of them,
    The \emph{instance focused on $B$} is $\Instance\mid_B = (G', P, U', \prec_U')$, where $G' = G[V(B) \cup U] + K$, $U' = U \cup K$, and $\prec_{U'} =  \prec_K \oplus \prec_U$.
\end{definition}
To put \Cref{def:poly-trees-preprocessing-focused} in other words, the instance $\Instance\mid_B$ consists precisely of the fixed vertices as well as those in the hydra attachment $B$, to which we add $n_P$ isolated vertices.
Thus, to find a solution we only need to order the vertices $V(B)$ and the isolated vertices $K$ ensures that the individual solutions for instances $\Instance\mid_B$ and $\Instance\mid_{B'}$ are ``compatible'', i.e., can be merged. This is captured by the following lemma.

\begin{lemma}
    \label{lem:poly-trees-preprocessing-separation-hydra-attachments}
    Let $\Instance = \InstanceLongOneS$ be an instance of \PAOneFOneS obtained from \Cref{lem:poly-trees-preprocessing-one-sided-pattern-avoidance} after applying \Cref{lem:poly-trees-preprocessing-separation-isolated-components} and let $v \in H$ be a hydra vertex.
    Then \Instance admits a solution if and only if there exists a set $H_v$ of at most $7n_P + \compress{S(\prec_U)}$ hydra attachments of $v$ such that
    \begin{enumerate}
        \item for every hydra attachment $T$ of $v$ for which $T \cap U \neq \emptyset$ we have $T \in H_v$,
        \item for every hydra attachment $B \notin H_v$ of $v$, the instance $\Instance\mid_B$ admits a solution, and
        \item $(G', P, U, \prec_{U})$ admits a solution, where $G'$ is obtained from $G$ after removing all hydra attachments $B \notin H_v$ of $v$. 
    \end{enumerate}
\end{lemma}
\begin{proof}
    We show both directions separately.

    \proofsubparagraph{($\boldsymbol{\Rightarrow}$)}
    Assume that \Instance admits a solution $\prec_G$.
    Since this implies that every induced subgraph of $G'$ also admits a solution, Property~3 holds (for every set $H_v$).
    It remains to show Properties~1 and~2.
    To this end, let $H_v$ denote the rightmost $7n_P + c(S(\prec_U))$ hydra attachments of $v$ in $\prec_G$%
    , i.e., traverse $\prec_G$ from right to left and consider every vertex $w \in V(G)$: if $w \in V(T)$ belongs to a hydra attachment $T$ of $v$%
    , add $T$ to $H_v$ (if it was not already in there); repeat this process until $\Size{H_v} = 7n_P + c(S(\prec_U))$.
    Observe that this, in particular, ensures that for every vertex $u \in U$ %
    from a hydra attachment $B$ of $v$ we have $B \in H_v$.
    Let $B \notin H_v$ be an arbitrary other hydra attachment of $v$ and observe $V(B) \cap U = \emptyset$.

    We now reorder the solution to obtain a different order $\prec_G'$ in which the vertices of $B$ appear contiguously as a sub-sequence.
    More concretely, to obtain $\prec_G'$, we remove all vertices $V(B)$ and append them, in the same order as in $\prec_G$, to the beginning of $\prec_G$, i.e., $\prec_G' = \prec_G\mid_{V(B)} \oplus \prec_G\mid_{V(G) \setminus V(B)}$.

    We now show that $\prec_G'$ is also a solution.
    Since $\prec_G'$ ends in $\prec_U$ as $\prec_G$ ends in $\prec_U$ and $V(B) \cap U = \emptyset$, the only way that $\prec_G'$ cannot be a solution is because it contains $P$.
    So, assume, for the sake of a contradiction, that $\prec_G'$ contains $P$.
    Consider the potential edge $uw \in E(G)$ that could represent the single forced edge in $P$.
    First, assume $uw \in E(B)$.
    Observe that the fixed set $U$ is to the right of $B$.
    Even after ignoring $v \in U$, there remain at least $20n_P$ vertices of different types.
    Among them, we can find a sufficiently large independent set (that has no edge to any vertex of $B$, since $B$ was only incident to $v$) that can take the role of the vertices to the right of the forced edge in the realization of $P$; thus $\prec_G$ would contain $P$.
    Next, assume $u,w \notin V(B)$.
    Since the vertices in $B$ are the leftmost in $\prec_G'$, the realization must also be in $\prec_G$; a contradiction.
    Therefore, the only edge in $G$ that can form the forced edge is the single edge $uv \in E(G)$ with $u \in V(B)$.
    
    Let $(\prec_G', X)$ be a realization of $P$ in $\prec_G'$ and let $G_P = G[X]$ be the graph induced on the vertices used to form the pattern.
    We now partition $V(G_P)$ into three sets based on their position in the solution $\prec_G$ to \Instance (and not the new linear order $\prec_G'$).
    The first set $S_1 \subseteq V(G_P)$ contains all vertices $w \in V(G_P)$ such that $b \prec_G w$ for every $b \in V(B)$, i.e., $S_1$ is completely to the right of $B$.
    The second set $S_2$ contains all vertices $w \in V(G_P) \setminus V(B)$ such that $w \prec_G b$ for some $b \in V(B)$.
    The third set $S_3$ contains all vertices $b \in V(B) \cap V(G_P)$.
    Without loss of generality, assume that the realization is selected such that $\Size{S_2}$ is minimal, i.e., the second set is a smallest set among all possible $(\prec_G', X)$ and, therefore, $G_P$.
    Note that $\Size{S_2} > 0$, as otherwise the pattern would already exist in $\prec_G$.
    To see this, observe that the realization of $P$ would then only use vertices from $S_1$ and $S_3$, whose relative order is identical in $\prec_G$ and $\prec_G'$

    Let $v_b \in V(B)$ denote the rightmost vertex from $B$ in $\prec_G'$.
    Recall \Cref{lem:poly-trees-preprocessing-matching}, which states that every vertex $u \in V(G)$ can only be spanned by a matching of size at most $6n_P$. %
    Thus, the number of hydra attachments of $v$ that are completely to the right of $v_b$, let them be $A$, is at least $\Size{A}\geq n_P + c(S(\prec_U))$.
    At least $n_P$ hydra attachments from $A$ are disjoint from $U$, i.e., consist entirely of free vertices.
    Let them be $A' \subseteq A$.
    We now select from every hydra attachment $T \in A'$ one vertex $v_T \notin N_G(v)$ to obtain an independent set $I_A$; recall that different hydra attachments are disconnected.
    Note that $I_A$ is also independent from every vertex $w \in V(G)$ with $w \prec_G v_b$.
    Recall that the only edge that can form the forced edge in $\prec_G'$ is the edge $uv$ with $u \in V(B)$ and $v \in U$ since $v$ is a hydra vertex.
    However, this implies that 
    all vertices in $I_A$ are between $u$ and $v$.
    Moreover, no vertex in $I_A$ is adjacent to $u$ or a vertex $w \in U \setminus \{v\}$.
    Therefore, we replace the vertices under the forced edge in the realization of $P$ with vertices from $I_A$.
    Observe that $I_A$ belongs to the (updated) set $S_1$.
    Thus, the (updated) set $S_2$ is empty; which is a contradiction to the existence of $\prec_G$ as argued above.
    Thus, we conclude that $\prec_G'$ cannot contain $P$.

    Next, we show that we can obtain from $\prec_G'$ a solution for $\Instance\mid_B$.
    To this end, consider the graph $G_B = G[V(B) \cup I_A \cup U]$ and observe that it is isomorphic to the graph constructed for $\Instance\mid_B$ if we identify the isolated vertices in $\Instance\mid_B$ with $I_A$.
    Furthermore, the linear order $\prec_G'\mid_{V(G_B)}$ is a solution to $\Instance\mid_B$; in particular observe that it places the isolated vertices as dictated by $\prec_{U'}$ in $\Instance\mid_B$.
    Combining all, we conclude that also Property~1 is fulfilled.

    \proofsubparagraph{($\boldsymbol{\Leftarrow}$)}
    Assume there exists a set $H_v$ of at most $7n_P + c(S(\prec_U))$ hydra attachments of $v$ that fulfills the properties of the lemma statement.
    For every hydra attachment $B \notin H_v$ of $v$, let $\prec_B'$ be the order obtained by taking a solution $\prec_B$ to $\Instance\mid_B$ and restricting it to vertices in $B$, i.e., $\prec_B' = \prec_B\mid_{V(B)}$.
    Note that $\prec_B$ exists due to Property~2 of the lemma statement.
    To construct an order $\prec_G$, consider an arbitrary indexation of the hydra attachments $B \notin H_v$ of $v$.
    To construct $\prec_G$, we start with the solution $\prec_{G'}$ for $(G', P, U, \prec_{U})$, which exists due to Property~3, and then append to its left the linear order $\prec_B$ for every $B_i \notin H_v$, i.e., $\prec_G = \prec_{B_1} \oplus \ldots \oplus \prec_{G'}$.
    Observe that $\prec_G$ is a linear order of $V(G)$.
    In particular, $\prec_G$ contains every vertex of $V(G)$.

    We now show that $\prec_G$ is a solution to $\Instance$.
    Assume contrary and consider the edge $uw \in E(G)$ that represents the forced edge in a realization $(\prec_G, X)$ of $P$.
    We show that no edge $uw \in E(G)$ exists that can represent the forced edge in $P$ implying that $(\prec_G, X)$ cannot be a realization.
    First, this must be the case if $uw \in E(G')$ as $(\prec_G', X)$ would then be a realization of $P$ as well, which contradicts the fact that $\prec_G'$ is a solution.
    Second, assume that it would be an edge $bv \in E(G)$ with $b \in V(B)$ for some $B \notin H_v$.
    Since every hydra vertex is contained in the fixed set (recall \Cref{lem:poly-trees-preprocessing-one-sided-pattern-avoidance}), we know that $bv$ also exists in $\Instance\mid_B$.
    Observe that the set $K$ of $n_P$ isolated vertices is between $b$ and $v$ in $\prec_B$, since $K$ is among the fixed vertices in $\Instance\mid_B$. 
    Therefore, if $bv$ takes over the role of the forced edge in $P$, then we can find the pattern using the edge $bv$, the isolated vertices $K$, and the vertices in $X$ right to $v$; recall that the forced edge in $P$ is incident to the leftmost vertex in $\prec_P$.
    Therefore, $\prec_B$ would not be a solution to $\Instance\mid_B$; a contradiction.

    So it remains to consider the third case, in which we assume that an edge $ab \in E(B)$ for some hydra attachment $B \notin H_v$ represents the forced edge in $(\prec_G, X)$. 
    Consider the edge $ab$ in the solution $\prec_B$ for the instance $\Instance\mid_B$.
    Since $\prec_B'$ is just the order $\prec_B$ induced on $V(B)$, the same (non-)edges that exist in $\prec_B'$ also exist in $\prec_B$, and, in particular the respective vertices in $X$ are from $V(B)$.
    Moreover, recall that there exists the set $K$ of $n_P$ isolated vertices in $\Instance\mid_B$, which are to the right of $ab$ as they are part of the fixed vertices. 
    Therefore, we can replace all vertices to the right of $ab$ in $X$ with vertices from $K$ and observe that with the new set $X'$, $(\prec_B, X')$ is a realization of $P$.
    However, this contradicts the assumption that $\prec_B$ is a solution to $\Instance\mid_B$.

    Combining all, we conclude that $\prec_G$ must be a solution to \Instance, which concludes the proof.
\end{proof}

\newcommand{\boundInterestingDegree}{38n_P}
An immediate consequence of \Cref{lem:poly-trees-preprocessing-separation-hydra-attachments} is the following:
\begin{corollary}
    \label{cor:poly-trees-preprocessing-interesting-degree}
    Let $\Instance = \InstanceLongOneS$ be an instance of \PAOneFOneS.
    We can assume that $\deg_G^{>1}(v) \leq \boundInterestingDegree$ for every $v \in V(G)$.
\end{corollary}
\begin{proof}
    For every hydra vertex, we can reduce its interesting degree such that is upper-bounded by $7n_P + \compress{S(\prec_U)}$.
    Replacing the latter part with the value obtained from \Cref{lem:poly-trees-preprocessing-one-sided-pattern-avoidance}, we get an upper bound of $27n_P+2$.
    Observe that this moves the interesting degree below the threshold of $38n_P$ for being a hydra vertex  (which means there are no hydra vertices left). %
\end{proof}

At this point, we make a slight digression towards restricting the structure of the solutions we need to consider. There, the following lemma shows that it is sufficient to focus on solutions with a special structure.

\newcommand{\boundSequenceLeaves}{2{(n_P +1)}^{n_P + 2}}
\begin{lemma}
    \label{lem:bounded-sequence-leaves}
     Let $\Instance = \InstanceLongOneS$ be an instance of \PAOneFOneS,
     and let $\prec_G$ be a solution for \Instance.
     If $S(\prec_G)$ restricted to vertices not in $U$ contains a consecutive subsequence $S'$ of size $\compress{S'} = \boundSequenceLeaves$ consisting only of at most $n_P$ different types, each of which is for a leaf-vertex.
     Then there exists another solution $\prec'$ with $\compress{S(\prec')} < \compress{S(\prec_G)}$. 
\end{lemma}
\begin{proof}
    Let $\prec_G$ be a solution fulfilling the properties of the statement and let $S'$ be such a subsequence.
    Below we show that we can obtain a new solution $\prec_{G'}$ where $S'$ has bounded alternation.
    Repeated application of this procedure yields the result.
    
    We use a similar idea as in \Cref{sec:nd-plus-mp-alternations}.
    In particular, we take $\prec_G$, perform some cuts along it, and show that we can re-order the linear order between two cuts and maintain a solution this way.

    To this end, let $\alpha$ be an assignment of vertices $V(P)$ to types from $(\eqclasses{V(G)}{\sim_G} \cap S') \cup \{\text{nil}\}$, where $\text{nil}$ denotes that some vertex in a realization if $P$ is not of some type from $\eqclasses{V(G)}{\sim_G} \cap S'$.
    Let $V(P) = \{v_1, \ldots, v_{n_P}\}$ be the vertices of $V(P)$ ordered as in $\prec_P$.
    We now traverse the subsequence of $\prec_G$ for entries from $S'$ from left to right until we arrive at the leftmost vertex $u_1 \in V(G)$ with $u_1 \in \alpha(v_1)$.
    We introduce a \emph{cut} before and after $u_1$ in $\prec_G$.
    Next, we continue our traversal of $\prec_G$ until we arrive at the leftmost vertex $u_2 \in V(G)$ right of $u_1$ with $u_2 \in \alpha(v_2)$.
    We repeat this process until we arrive at the end of the subsequence of $\prec_G$; skipping all $v_i$ with $\alpha(v_i) = \text{nil}$.
    There are at most $2n_P$ cuts that we make for $\alpha$.
    Afterwards, we perform similar for all possible functions $\alpha$.

    Note that there are ${(n_P+1)}^{n_P}$ different assignments $\alpha$, since there are only at most $n_P$ different types in $S'$.
    Hence, we make at most $2{(n_P+1)}^{n_P + 1}$ cuts.
    Afterwards, we reorder between two cuts based on the types from $\eqclasses{V(G)}{\sim_G}$, which reduces the alternation in $S'$ to $2{(n_P+1)}^{n_P + 1}$.
    Let $\prec_{G}'$ be a linear order that is consistent with the reordered sequence.
    To see that $\prec_{G}'$ is a solution, first observe that it still ends in $\prec_U$.
    So assume that $\prec_{G}'$ contains $P$ and consider the assignment $\alpha$ of the realization to $(\eqclasses{V(G)}{\sim_G} \cap S') \cup \{\text{nil}\}$.
    In particular, consider the cuts made for $\alpha$. 
    For every vertex from the realization that was re-ordered in order to obtain $\prec_{G'}$, there exists at least one vertex (to the left) of the same type for which for which we introduce the cuts.
    Also these vertices (together with the ones outside $S'$) are also a realization of $P$ in $\prec_G'$.
    As they are surrounded by cuts, they are not re-ordered and exist at the same position in $\prec_G$ and $\prec_G'$, which implies that also $\prec_G$ would contain the pattern; a contradiction to the existence of the solution.
\end{proof}

The reason we made the digression to Lemma~\ref{lem:bounded-sequence-leaves} was for it to serve as a tool for reducing our instance of \PAOneFOneS\ to the case where we are guaranteed that for every subset $A$ of vertices in the non-fixed part of specified size, the fixed part contains a sufficiently large independent set $I_A$ that is also independent to $A$. We first show that this holds for $|I_A|=1$ (in \Cref{lem:poly-trees-preprocessing-independent-neighbor}), and then extend the argument to a sufficient lower bound in \Cref{lem:poly-trees-preprocessing-independent-right}.

\begin{lemma}
    \label{lem:poly-trees-preprocessing-independent-neighbor}
    Let $\Instance = \InstanceLongOneS$ be an instance of \PAOneFOneS.
    Assume that $d_G^{>1}(v) \leq 38n_P$ for all $v \in V(G)$ and set $d_{\max}^{>1} = 38n_P$.
    We can construct a family $\mathcal{F}$ of instances of \PAOneFOneS such that
    \begin{enumerate}
        \item $\Size{\mathcal{F}} \in {n_G}^{{n_P}^{\BigO{n_P}}}$ and can be enumerated in this time,
        \item for every $(G', P', U', \prec_{U'}) \in \mathcal{F}$ we have $U \subseteq U'$, $G'=G$ and $P' = P$,  
        \item for every set $V' \in V(G) \setminus U'$ of size $n_P$ and every $(G', P', U', \prec_{U'}) \in \mathcal{F}$ there exists at least one $u \in U'\setminus U$ such that there is no edge $uv \in E(G)$ with $v \in V'$,
        \item \Instance is a positive instance if and only if at least one instance $\Instance' \in \mathcal{F}$ is a positive instance, and
        \item $\compress{S(\prec_{U'})} \leq \compress{S(\prec_U)} + (2{(n_P +1)}^{n_P + 2} + 1) \cdot (d_{\max}^{>1} \cdot n_P + 1)$
    \end{enumerate}
\end{lemma}
\begin{proof}
    Assume there is a solution $\prec$ for ($G$, $P$, $U$, $\prec_U$), and assume, without loss of generality, that $\prec$ is a solution that minimizes $\compress{S(\prec)}$. Let $\prec_{V(G)\setminus U}$ be the solution restricted to vertices not in $U$.

    Let $S$ be the sequence obtained from the $(2{(n_P +1)}^{n_P + 2} + 1) \cdot (d_{\max}^{>1} \cdot n_P + 1)$ right-most entries in $\prec_{V(G)\setminus U}$.
    Moreover, let $U_\mathrm{new}$ be the vertices used to obtain $S$, and let $U' = U_\mathrm{new} \cup U$.

    We want to show that, for every $V'$ with size $\Size{V'} = n_P$ and $V' \cap U_\mathrm{new} = \varnothing$, there exists a vertex $u$ in $U_\mathrm{new}$ such that $u$ is not adjacent to any vertex in $V'$.

    By assumption, all vertices have interesting degree bounded by $38n_P$.
    If $U_\mathrm{new}$ contains more than $x \coloneqq (38n_P +1) \cdot n_P$ vertices of degree 2 or more, then for every set of vertices $V'$ of size $n_P$, $V' \cap U_\mathrm{new} = \varnothing$ and $V' \cap U = \varnothing $, there exists a vertex in $U_\mathrm{new}$ that is not adjacent to any vertex in $V'$.

    Similarly, if $U_\mathrm{new}$ contains at least $n_P+1$ vertices of degree 1 that have pairwise distinct neighbors, then for every set of vertices $V'$ with $\Size{n_P}$, $V' \cap U_\mathrm{new} = \varnothing$ and $V' \cap U = \varnothing $, the neighbor of at least one of these degree-1 vertices is not present, and this vertex is thus not adjacent to any vertex in $V'$.

    Further, if $U_\mathrm{new}$ contains a vertex $u$ of degree 0, then for every set of vertices $V'$ with $\Size{n_P}$, $V' \cap U_\mathrm{new} = \varnothing$ and $V' \cap U = \varnothing $, $u$ is not adjacent to any vertex in $V'$.

    We now assume towards contradiction that $U_\mathrm{new}$ contains at most $x$ vertices of degree two or more, no vertex of degree zero, and at most $n_P$ different types of leaves.
    The maximum size of a sequence containing only leaves of at most $n_P$ types is bounded by $2{(n_P +1)}^{n_P + 2}$ (recall \Cref{lem:bounded-sequence-leaves}), and we can have at most $38n_P \cdot n_P + 1$ many of those.
    Adding the vertices of degree at least two, we know that the size of $S$ is also bounded by $(2{(n_P +1)}^{n_P + 2} + 1) \cdot (d_{\max}^{>1} \cdot n_P + 1) - 1$.
    This is a contradiction to the chosen length of $S$.
    Hence, the instance $(G, P, U', \prec_{U'})$, where $\prec_U' = \prec\mid_{U'}$, also has $\prec$ as a solution.

    We now construct $\mathcal{F}$ by first generating all sequences $S^*$ of length at most $\compress{S^*} \leq (2{(n_P +1)}^{n_P + 2} + 1) \cdot (d_{\max}^{>1} \cdot n_P + 1)$, containing types from vertices $v \in V(G) \setminus U$.
    Since there are at most ${n_G}$ types, and each type can occur ${n_G}$-times, this is bounded.
    From all generated $S^*$, we eliminate all sequences that do not contain an isolated vertex, or at least $n_P+1$ different types of degree leaves, or $d_{\max}^{>1} \cdot n_P + 1$ vertices of degree 2 or more.
    We denote the resulting set of sequences by $\mathcal{S}$.

    Following the argument from above, the compressed version of $\prec_{U'}$ is contained in $\mathcal{S}$.

    Therefore, the family of instances 
    \begin{align*}
        \mathcal{F} = \{(G, P, \overline{U}, \prec_{\overline{U}}) \mid \exists S \in \mathcal{S}\colon \prec_{\overline{U}}\ =\ \prec_S \oplus \prec_U, \prec_S\ \text{is compatible with}\ S\}    
    \end{align*}
    contains $(G, P, U', \prec_{U'})$.
    Thus, if \Instance has a solution, the family $\mathcal{F}$ has a solution as well. Conversely, if there exists a solution for some instance in $\mathcal{F}$, it is also a solution for \Instance.

    Further, $\Size{\mathcal{F}} \leq \Size{\mathcal{S}} \leq ({n_G}^{2\cdot (2{(n_P +1)}^{n_P + 2} + 1) \cdot (d_{\max}^{>1} \cdot n_P + 1)}) \leq {n_G}^{{n_P}^{\BigO{n_P}}}$.
    Properties 1--5 follow from the discussion above.
\end{proof}
We remark that the size of the set $U'$ in the above lemma need not be bounded by a function of $n_P$.

\newcommand{\boundFixedVertices}{{n_P}^{\BigO{n_P}}}
\begin{lemma}
    \label{lem:poly-trees-preprocessing-independent-right}
    Let $\Instance = \InstanceLongOneS$ be an instance of \PAOneFOneS with $\compress{S_{\prec_{U}}} \leq 20n_P + 2$.
    Assume that $d_G^{>1}(v) \leq 38n_P$ for all $v \in V(G)$ and set $d_{\max}^{>1} = 38n_P$.
    There exists a family $\mathcal{F}$ of instances of \PAOneFOneS such that
    \begin{enumerate}
        \item $\Size{\mathcal{F}} \in {n_G}^{{n_P}^{\BigO{n_P}}}$,
        \item for every $(G, P, U', \prec_{U'}) \in \mathcal{F}$ we have $U \subseteq U'$,
        \item for every set $A \subseteq V(G) \setminus U'$ with $\Size{A} = n_P$, there exists an independent set $I_A \subseteq U'$ of size $n_P$ such that no vertex in $I_A$ is adjacent to one in $A$,
        \item \Instance admits a solution if and only if at least one instance $\Instance' \in \mathcal{F}$ does, and
        \item $\compress{S(\prec_{U'})} \leq (20n_P + 2) + (2{(n_P +1)}^{n_P + 2} + 1) \cdot (d_{\max}^{>1} \cdot n_P + 1) \cdot (2n_P)\in \boundFixedVertices$.
    \end{enumerate}    
\end{lemma}
\begin{proof}
    Our goal is to iteratively apply \Cref{lem:poly-trees-preprocessing-independent-neighbor} $2n_P$ times.    By \Cref{lem:poly-trees-preprocessing-independent-neighbor}, the obtained family $\mathcal{F}_1$ of instances $(G, P, U', \prec_{U'})$ of \PAOneFOneS has $U \subseteq U'$ and for every set $A \subseteq V(G) \setminus U'$ there exists a fixed vertex $u \in U' \setminus U$ that is not adjacent to any vertex in $A$.
    We now again call \Cref{lem:poly-trees-preprocessing-independent-neighbor} with every instance in $\mathcal{F}_1$ to obtain a family $\mathcal{F}_2$ and repeat this process until we eventually obtain the family $\mathcal{F} = \mathcal{F}_{2n_P}$ of instances.
    Properties~1, 2, and 4 follow directly from applying \Cref{lem:poly-trees-preprocessing-independent-neighbor} iteratively $2n_P$ times. 
    By assumption, the size of $\compress{S(\prec_U)}$ was initially bounded by $(20n_P + 2)$. %

    It remains to show Property~3.
    \Cref{lem:poly-trees-preprocessing-independent-neighbor} guarantees us for every $(G, P, U', \prec_{U'}) \in \mathcal{F}_i$ and every $i \in [2n_P]$ the existence of a new vertex $u_i \in U' \setminus U$ that is not adjacent to all sets $A \subseteq V(G) \setminus U'$ of size $n_P$; of course the concrete $u_i$ depends on the instance in question.

    Now consider an instance $(G, P, U', \prec_{U'}) \in \mathcal{F}_{2n_P} = \mathcal{F}$ and let $A \subseteq V(G) \setminus U'$ be a fixed set of free vertices.
    Moreover, let $B = \{u_1, \ldots, u_{2n_P}\}$ be these above-identified vertices, which, since every application of \Cref{lem:poly-trees-preprocessing-independent-neighbor} identifies a previously unfixed vertex, must all be distinct.
    Observe that no $u_i \in B$ is incident to some $v \in A$ due to \Cref{lem:poly-trees-preprocessing-independent-neighbor}.
    Using \Cref{obs:forests-independent-set}, we know that there exists an independent set $I_A \subseteq B$ of size $n_P$ among $B$, which serves as witness for Property~3.
\end{proof}
Note that, similar to \Cref{lem:poly-trees-preprocessing-independent-neighbor}, we did not specify any bound on the size of $U'$ in \Cref{lem:poly-trees-preprocessing-independent-right}.
In fact, there are instances where, due to a high number of leaves, the size of $U'$---but not $\compress{S(\prec_{U'})}$---is unbounded.
\Cref{fig:fixed-side} illustrates the result of \Cref{lem:poly-trees-preprocessing-independent-right}.

Next, we show that the type-alternation of leaves below edges with free endpoints can assumes to be bounded in a solution.
Later on, we use this to ensure that every such edge will appear in the ``sliding window'' of the dynamic program.

\begin{lemma}
    \label{lem:poly-trees-preprocessing-bound-alternation}
    Let $\Instance = \InstanceLongOneS$ be an instance of \PAOneFOneS obtained from \Cref{lem:poly-trees-preprocessing-one-sided-pattern-avoidance} after applying \Cref{lem:poly-trees-preprocessing-separation-hydra-attachments,lem:poly-trees-preprocessing-separation-isolated-components}.
    For an edge $uv \in E(G)$ with $u \prec_G v$, let $\prec_G\mid_{uv}$ be a shorthand for $\prec_G\mid_X$ for $X = \{w \in V(G) \mid u \preceq_G w \preceq_G v\}$.
    If $\Instance$ admits a solution $\prec_G$, then there also exists a solution $\prec_G'$ where for every edge $uv \in E(G)$ with $u,v\notin U$, we have $\compress{S(\prec_G'\mid_{uv})} \leq ((2\cdot (d_{\max}^{>1}) + 2n_P) + 1) \cdot (2n_P + 2)\in\BigO{{n_P}^2}$, where $d_{\max}^{>1} = 38n_P$. %
\end{lemma}
\begin{proof}
    Let $\prec_G$ be a solution to \Instance and assume that there exists an edge $uv$ where the claim from the lemma statement does not apply.
    Without loss of generality, assume $u \prec_G v$.
    In the following, we construct a linear order $\prec_G'$ by re-ordering the vertices $w \in V(G)$ spanned by the edge $uv$ such that $\prec_G'$ remains a solution, $\compress{S(\prec_G'\mid_{uv})} < \compress{S(\prec_G\mid_{uv})}$, and $\compress{S(\prec_G')} \leq \compress{S(\prec_G)}$.
    As a consequence, the lemma statement follows by iteratively applying the procedure below.

    We start by describing how we reorder the vertices.
    Let $N_{uv}^{=1} = \{w \in N_G(u) \cup N_G(v) \mid d_G(w) = 1\}$ be the leaf-neighbors of $u$ or $v$.
    We observe that the number of vertices $w \in V(G) \setminus N_{uv}^{=1}$ with $u \prec_G w \prec_G v$ must be upper-bounded by $x = 2 d_{\max}^{>1} + 2n_P$.
    To see this, assume otherwise.
    Since the interesting degree of $u$ and $v$ is bounded by $d_{\max}^{>1}$, there are at least $2n_P$ such vertices $w$ that are neither adjacent to $u$ or $v$. 
    Among these, we can find an independent set $I_{uv}$ of size $\Size{I_{uv}} = n_P - 2$; recall \Cref{obs:forests-independent-set}.
    In particular, $I_{uv} \cup \{u,v\}$ remains an independent set and we have $\Size{I_{uv} \cup \{u,v\}} = n_P$.
    We can now use \Cref{lem:poly-trees-preprocessing-independent-right} to find an independent set $I_R$ of size $\Size{I_R} = n_P$ to the right of $v$ such that $I_R \cup I_{uv} \cup \{u,v\}$ remains an independent set.
    Together with the edge $uv$, these vertices would therefore be a witness that $\prec_G$ contains $P$ (similar to \Cref{obs:forests-independent-set-contains-pattern}).
    This would contradict the assumption that $\prec_G$ is a solution to \Instance.
    Consequently, there are at most $x$ entries in $S(\prec_G\mid_{uv})$ corresponding to a type that is not a leaf-neighbor of $u$ or $v$.

    Assume now that $S(\prec_G\mid_{ab})$ contains a contiguous maximal subsequence $S'$ such that $\compress{S'} \geq 2n_P + 3$ and for every $t \in S'$ it holds $t \cap N_{uv}^{=1} \neq \emptyset$, i.e., $t$ is the equivalence class of leaf neighbors of $u$ or the  equivalence class of leaf neighbors of $v$.
    Let $S''$ be the sequence obtained from $S'$ by removing entries from the beginning of $S'$ until $\compress{S'} - \compress{S''} = 2n_P$.
    Observe that we have $\compress{S''} \geq 3$.
    Let $t_a\in S''$ be the first type in $S''$ and let $t_b \in S''$ be the next type in $S''$ such that $t_a \neq t_b$, which must exist as $\compress{S''} \geq 3$.
    Let $a\in t_a$ and $b \in t_b$ be two vertices.
    Observe that $a$ and $b$ are two leaves with $a \not\sim_G b$, i.e., $a$ and $b$ are not siblings.
    We now construct a new sequence $R = (t_a, \ldots, t_b, \ldots)$ where we repeat $t_a$ and $t_b$ precisely $\multiset{S''}(t_a)$ and $\multiset(S'')(t_b)$ times, respectively.
    We replace $S''$ with $R$ in $S'$ to obtain $\hat{S'}$ and observe that $\compress{\hat{S'}} = n_P + 2$.
    In the following, we also refer to $\hat{S'}$ as the \emph{compressed} sequence.

    Repeat the above steps for all such contiguous maximal subsequences and let $S^*$ denote the resulting sequence.
    By iterating above arguments, we conclude $\compress{S^*} \leq (x + 1) \cdot (2n_P + 2)$, where we recall that $x$ denotes the number of vertices spanned by $uv$ that are not leaf-neighbors of $u$ or $v$.
    Let $\prec_G'$ be a linear order compatible with $S^*$ such $\prec_G'$ and $\prec_G$ differ only in the order of vertices under the edge $uv$.

    Observe that $\compress{S(\prec_G')} \leq \compress{S(\prec_G)}$ since we only ``merge'' entries of the same type.

    It remains to show that $\prec_G'$ remains a solution.
    Since both $u$ and $v$ are free, $\prec_G'$ still ends in $\prec_U$, so assume that $\prec_G'$ contains the pattern $P$.
    We now argue that we can also find a realization of $P$ in $\prec_G$.
    In particular, we show that there exists a realization that does not use vertices corresponding to the entries in $S^*$ that we re-ordered.
    Let $(\prec_G', X)$ be a realization of $P$ in $\prec_G'$.
    Consider a compressed sequence $S'$.
    We mark every $w \in X$ that is represented in $\prec_G'$ by an entry in $S'$.
    As $\Size{X} = n_P$, we mark at most $n_P$ vertices.
    Recall that $S'$ is preceded by a subsequence $Q$ with $\compress{Q} = 2n_P$ that we did not touch.
    Moreover, since $Q$ only contains types $t$ with $t \cap N_{uv}^{=1} \neq \emptyset$, we have that $Q$ must be an alternation of sequences of $t_a = [a]_{\sim_G}$ and $t_b = [b]_{\sim_G}$ for $a$ and $b$ being two leaf-neighbors of $u$ and $v$, respectively. 
    More concretely, we have $Q=(t_a,\ldots,t_b,\ldots,t_a,\ldots,t_b)$ or $Q=(t_b,\ldots,t_a,\ldots,t_b,\ldots,t_a)$.
    Without loss of generality, assume the former.
    We now re-assign the (at most) $n_P$ marked vertices in $X$ to vertices responsible for the entries $Q$ in $\prec_G'$.
    To this end, consider the leftmost marked vertex $w_1 \in X$.
    If $w_1 \in t_a$, we replace it with the leftmost vertex (in $\prec_G'$) responsible for the first $t_a$ in the first sequence of $t_a$s in $Q$, otherwise (i.e., $w_1 \in t_b$) with the leftmost vertex responsible for the first $t_b$ in the first sequence of $t_b$s in $Q$.
    We replace the next marked vertex $w_2 \in X$ with the leftmost vertex responsible for the first $t_a$ ($t_b$) in the second sequence of $t_a$s ($t_b$s) in $Q$ if $w_2 \in t_a$ ($w_2 \in t_b$).
    Afterwards, we continue in the same fashion until all marked vertices are re-assigned.
    Since $Q$ contains only $t_a$ or $t_b$ and $\compress{Q} = 2n_P$, this process must eventually finish.
    Afterwards, we do the same for all compressed sequences.
    Eventually, no vertex in $X$ is responsible for an entry of a compressed sequence but only for types in sequences that we did not alter.
    This implies that $P$ must already be contained in $\prec_G$, which is a contradiction to the assumption that it is a solution to $\Instance$.
    Therefore, $\prec_G'$ must also be a solution, which completes the proof.
\end{proof}

\NewDocumentCommand\somefunc{O{x}m}{\ensuremath{f_{#1}(#2)}}

\subparagraph*{Wrapping Up.}
Before we continue in \Cref{sec:poly-trees-dp} with describing our dynamic program, let us first summarize the result of this section.
Let $\Instance = \InstanceLong$ be an instance of \PAOneF, i.e., $G$ is a forest and $P$ is a pattern with $E^+(P) = \{uv\}$ and $u$ is the leftmost vertex in $\prec_P$.

We exhaustively apply \Cref{lem:poly-trees-preprocessing-one-sided-pattern-avoidance,lem:poly-trees-preprocessing-separation-isolated-components,lem:poly-trees-preprocessing-separation-hydra-attachments,lem:poly-trees-preprocessing-independent-right} to $\Instance$, some of which replace an instance with a set of (simpler) instances where either at least one or even all need to be positive for the input instance to be positive.
The total number of instances generated in one step is bounded by
${n_G}^{{n_P}^{\BigO{n_P}}}$.
Any resulting instance obtained from this exhaustive application was obtained by at most $\BigO{n_P}$ applications of one of the above lemmas.
Thus, the total number of instances is also bounded by ${n_G}^{{n_P}^{\BigO{n_P}}}$.
Each such instances $\Instance' = (G', P', U', \prec_{U'})$ of \PAOneFOneS has the following properties.
\begin{enumerate}[(N1)]
    \item\label[instprop]{prop:components-contain-fixed-v}
        Either $G'$ is a tree, or for every connected component $T$ of $G'$ we have $V(T) \cap U' \neq \emptyset$, i.e.,
        if $G'$ is a forest, then every connected component has at least one fixed vertex
        (\Cref{lem:poly-trees-preprocessing-separation-isolated-components}).
        Recall that for $U'=\emptyset$, $G'$ is a tree and the forced edge connects the leftmost and rightmost pattern vertex.
    \item\label[instprop]{prop:interesting-degree-bounded}
        For every vertex $v \in V(G')$ we have $d_{G'}^{>1}(v) \leq \boundInterestingDegree$, i.e.,
        the interesting degree of every vertex is bounded in $n_P$
        (\Cref{lem:poly-trees-preprocessing-separation-hydra-attachments} and, in particular, \Cref{cor:poly-trees-preprocessing-interesting-degree}).
    \item\label[instprop]{prop:fixed-independent-set}
        If $U'\neq\emptyset$, for every $A \subseteq V(G') \setminus U'$ of size at most $n_P$, there is an independent set $I_A \subseteq U'$ of size $n_P$ such that there is no edge $uv$ with $u \in A$, $v \in I_A$.
        This means that for every not too large set of free vertices, there is a same-sized independent set among the fixed vertices that are not adjacent to the former
        (\Cref{lem:poly-trees-preprocessing-independent-right}).
    \item\label[instprop]{prop:fixed-part-bounded-alternation}
        We have $\compress{S(\prec_{U'})} \leq \boundFixedVertices$, i.e.,
        there is only a bounded type-alternation in $\prec_{U'}$
        (\Cref{lem:poly-trees-preprocessing-independent-right}).
\end{enumerate}
We call instances $\Instance'$ of \PAOneFOneS that have these properties \PANormalized\ \PAOneFOneS.

The above discussion is summarized in the following lemma that establishes a reduction from \PAOneF to \PAOneFOneS.
\begin{lemma}
    \label{lem:poly-trees-preprocessing-touring-reduction}
    For some computable function $f(\cdot)$, there exists an $\BigO{{n_G}^{f(n_P)}}$-time algorithm that takes as input an instance $\Instance$ of \PAOneF and outputs a family $\mathcal{F}$ of sets $\mathcal{S} \in \mathcal{F}$ of instances of \PANormalized\ \PAOneFOneS such that $\Size{\mathcal{F}}, \Size{\mathcal{S}} \in \BigO{{n_G}^{f({n_P})}}$ and $\Instance$ is a positive instance if and only if there is one $\mathcal{S} \in \mathcal{F}$ where every $\Instance' \in \mathcal{S}$ is a positive instance.
\end{lemma}

Finally, \Cref{lem:poly-trees-preprocessing-bound-alternation,lem:bounded-sequence-leaves,lem:poly-trees-preprocessing-matching} guarantee us the existence of a solution $\prec_G$ to an instance of \PAOneFOneS with the following properties.
\begin{enumerate}[(S1)]
    \item\label[solprop]{prop:sol-bounded-matching}
    The size of a matching over every vertex $v \in V(G)$ is bounded by $6n_P$
    (\Cref{lem:poly-trees-preprocessing-matching}).
    \item\label[solprop]{prop:sol-bounded-alternation-below-edge}
    For every $uv \in E(G)$ with $u,v \notin U$ and $X = \{w \in V(G) \mid u \preceq_G w \preceq_G v\}$, we have $\compress{S(\prec_G\mid_X)}\in\BigO{{n_P}^2}$, i.e.,
    there is only a bounded type-alternation below every edge between free vertices
    (\Cref{lem:poly-trees-preprocessing-bound-alternation}).
    \item\label[solprop]{prop:sol-long-alternation-compressible}
    For every sequence $S$ in $\prec_G\mid_{V(G) \setminus U}$ consisting of at most $n_P$ types that are all leaves, 
    we have $\compress{S} \leq \boundSequenceLeaves$, i.e.,
    every sequence of length at least $\boundSequenceLeaves$ contains not only $n_P$ leaf-types
    (\Cref{lem:bounded-sequence-leaves}).
\end{enumerate}

\subsection{A Dynamic Program for Normalized Instances}
\label{sec:poly-trees-dp}
In this section, we present a dynamic programming (DP) algorithm for  \PANormalized\ \PAOneFOneS.
On a high level, we maintain an ``active window'' of size $\BigO{n_P}$ that we conceptually slide from right to left over our linear order, thus incrementally constructing a solution $\prec_G$ (if it exists) starting from the fixed part on the right.
To ensure that the size of our window can be bounded by $n_P$ and we do not need to store the orders outside the window,
we use (i) the property that the number of connected components in $G$ is bounded, and (ii) few of them can ``span'' over our window, i.e., have vertices to the left and right of it.
This allows us to only keep track of the already processed (and still to process) connected components of $G$. %
We remark that dynamic programming was already used for obtaining a solution for \probname{Bandwidth}~\cite{Saxe80}.

Our DP will rest on exploring various configurations corresponding to partial solutions represented as vertices in a \emph{configuration graph}.
A configuration consists of the actual order of the vertices within the active windows, together with information on which vertices (or rather components) still need to be processed.
Two configurations are adjacent in the graph if one configuration can be obtained from the other by, intuitively, sliding the active window one step further.
An overall solution to the instance (if it exists) thereby corresponds to a path from the unique starting configuration to any final configuration where the window moved over all vertices.

\newcommand{\windowLength}{(\boundSequenceLeaves+1)\cdot(38 {n_P}^2+1)+1}
\subparagraph*{The DP Configuration Graph.}
We start with formally describing the vertices of the \emph{configuration graph} $D$ for a given instance $\Instance = \InstanceLongOneS$ of \PANormalized\ \PAOneFOneS.
A \emph{configuration} is a tuple $(S, L)$, where the
\emph{active window} $S$ is a sequence of vertices with $\compress{S} \leq \windowLength$ and the
\emph{remaining set} $L$ is a vertex-type multiset such that $\multiset{S}+L+\multiset{U}\subseteq \multiset{V(G)}$.
Note that in $L$, the multiplicity of non-leaf vertices is at most 1 while the multiplicity of a leaf vertex-type is at most the number of leaves of said type.
For any configuration $(S, L)$, we define the \emph{already-processed set} $R:=\multiset{V(G)} - (\multiset{S}+L+\multiset{U})$, which will, intuitively, describe the vertices placed between $S$ and $U$ in the partial solution described by the current configuration.
Furthermore, we partition $S$ into sequences $S^L,S^R$ such that $S^L\oplus S^R=S$ with $|S^L|=\lceil\frac{|S|}{2}\rceil$.
The \emph{starting configuration} has $S=\emptyset$ and $L=\multiset{V(G)\setminus U}$, while any configuration with $L=\emptyset$ is considered \emph{final}.
Given a configuration $(S, L)$, we can now obtain another configuration $(S',L')$ such that $D$ contains an edge from the former to the latter by, roughly speaking, taking (a representative of) an arbitrary vertex(-type) from $L$ and prepending it to $S$ while removing the last vertex from $S$ if it thereby became too long.
As considering all configurations obtainable in this way would make $D$ far too large, we will show that we can restrict our attention to only the valid configurations defined as follows.

\begin{definition}
    \label{def:poly-trees-dp-valid-entry}
    We call a configuration $(S, L)$ \emph{valid} if
    \begin{enumerate}[(V1)]
        \item\label[valprop]{prop:val-S-nonempty} it is either the starting configuration or $S\neq\emptyset$,
        \item\label[valprop]{prop:val-no-long-edge} there is no edge $uv$ with $u \in L$ and $v \in R$,
        \item\label[valprop]{prop:val-no-half-edge} there is no edge $uv$ with $u \in S^L$ and $v \in R$ or with $u \in L$ and $v \in S^R$,
        \item\label[valprop]{prop:val-P-avoided} the sequence $S \oplus \prec_U$ avoids $P$, and
        \item\label[valprop]{prop:val-bounded-alternation}  for every subsequence $S'$ of $S$ consisting of at most $n_P$ types that are all leaves, 
        we have $\compress{S'} \leq 2{(n_P +1)}^{n_P + 2}$.
    \end{enumerate}
\end{definition}

Note that we need to only generate valid `successors' of a current configuration, which is why need to be able to check validity efficiently. That is handled by the following lemma.

\begin{lemma}
    \label{lem:poly-trees-dp-valid-entry-check}
    Given a tuple $(S, L)$, we can test in time $\BigO{{n_G}^{n_P}\cdot {n_P}^2}$ whether it corresponds to a valid configuration.
\end{lemma}
\begin{proof}
    \Cref{prop:val-S-nonempty} can trivially be verified in bounded time.
    For the next constraints, note that the only vertices that have a choice of representative (other than themselves) are the leaves, and any choice here yields an equivalent edge to the same neighbor.
    Thus, \Cref{prop:val-no-long-edge} holding is independent of the choice of representatives and we can check in linear time whether there is an edge between the two sets.
    Similarly, \Cref{prop:val-no-half-edge} can be checked in linear time.
    We can use \Cref{lem:poly-trees-preprocessing-pattern-on-left} to check \Cref{prop:val-P-avoided}, that is that $S\ \oplus \prec_U$  avoids $P$, which dominates the running time.
    Finally, \Cref{prop:val-bounded-alternation} can be verified by a linear pass over $S$, storing the counts for the most-recently seen leaf types.
\end{proof}

Next, we will show that the graph $D$ restricted to valid configurations has bounded size.
First, we show that for any given choice of $S$, there are only boundedly many values for $L$ that yield a valid configuration.
Afterwards, we will argue that there are only a bounded number of values for $S$, yielding the desired bounds on $D$.

\begin{lemma}
    \label{lem:poly-trees-dp-valid-entry-count-L}
    Let $\Instance = \InstanceLongOneS$ be an instance of \PANormalized\ \PAOneFOneS and let $S$ be a sequence of vertices.
    The number of different remaining sets $L$ such that $(S, L)$ is a valid configuration is bounded by 
    ${n_G}^{{n_P}^{\BigO{n_P}}}$.
\end{lemma}
\begin{proof}
    First, observe that for $S=\emptyset$ there is only the starting configuration that is valid, which immediately yields the single value we have for $L$.
    For $S$ non-empty, consider the connected components of the graph obtained from $G$ by deleting all vertices in $S$ and in $U$.
    We partition these components into different sets, depending on whether they consists solely of a leaf and depending on their adjacency to $S$ and $U$.
    Recall that if we have a forest, all connected components contain at least one fixed vertex (\Cref{prop:components-contain-fixed-v}), so there are no components adjacent to neither $S$ nor $U$.
    Similarly, if we only have a single tree, then it needs to be adjacent to the non-empty $S$.
    In either case, we will first consider those components that consist of a single leaf.
    Here, observe that the total number of vertex-types for leaves that have their parent in $S$ or $U$ is bounded, as the compressed sizes of $S$ (by definition) and $U$ (by \Cref{prop:fixed-part-bounded-alternation}) are bounded.
    Thus, for each of the at most $\windowLength + \boundFixedVertices$ such vertex-types, we need to decide how many leaves are in $L$.
    As there are at most $n_G$ such leaves, this yields ${n_G}^{\windowLength + \boundFixedVertices}$ choices.

    Now consider the components that not only contain a leaf, i.e. have at least two vertices or two edges to $S\cup U$.
    First observe that, since in a valid configuration no edge is allowed to connect between $R$ and $L$, if one vertex of a component is in $L$ then so is the whole component.
    If such component contains a vertex $v$ that is endpoint of an edge $uv$ with $u\in S$, we can uniquely determine whether it needs to part of $L$ or $R$ as follows.
    Since there are no edges between $S^L$ and $R$ in a valid configuration, we know that the component of $v$ is in $L$ if $u\in S^L$.
    As the same holds for $S^R$ and $L$, this is an `if and only if', yielding a deterministic classification.

    Finally, consider the components that do not have an edge to $S$ but to $U$.
    Recall that \Cref{prop:interesting-degree-bounded} implies that the interesting degree is bounded via $d_{\max}^{>1} \leq \boundInterestingDegree$.
    Together with the above-mentioned bounds on $U$ (\Cref{prop:fixed-part-bounded-alternation}), this allows us to bound the number of components that have an edge to $U$ by $\boundInterestingDegree \cdot \boundFixedVertices$.
    Each such component can either lie in $L$ or $R$, yielding $2^{\boundInterestingDegree \cdot \boundFixedVertices}$ possibilities.
\end{proof}

As the length of the active window $S$ is bounded in $n_P$, and each position within this window can contain one of the $n_G$ vertices, this immediately gives us a bound on the number of values for $S$.

\begin{observation}
    \label{lem:poly-trees-dp-valid-entry-count-S}
    Let $\Instance = \InstanceLongOneS$ be an instance of \PANormalized\ \PAOneFOneS.
    The number of different active windows $S$ such that there exists some $L$ for which $(S, L)$ is a valid configuration is bounded by 
    ${n_G}^{{n_P}^{\BigO{n_P}}}$.
\end{observation}

Together with the above lemma, this yields a bound on the size of $D$, which will allow us to efficiently check for a \emph{solving path} in $D$, i.e, a path from the starting configuration to a final one.
It remains to show that the instance is positive if and only if there is such path.
We start with the forward direction of this.

\begin{lemma}
    \label{lem:poly-trees-dp-solution-implies-dp}
    Let $\Instance = \InstanceLongOneS$ be an instance of \PANormalized\ \PAOneFOneS.
    If $\Instance$ is a positive instance, then there exists a solving path in $D$.
\end{lemma}
\begin{proof}
    If the instance is positive, it by \Cref{prop:sol-long-alternation-compressible} also has a solution $\prec_G$ where each contiguous subsequence of leaves satisfies \Cref{prop:val-bounded-alternation}.
    We now show that we can find a suitable path in $D$ by following edges labeled with the vertex-types of $\prec_G$ in reversed order, starting from the last vertex right before $U$.
    For this initial vertex, it is easy to see that the starting state has an appropriately-labeled edge, as the pattern is not induced and we have no $R$ or $S^R$ yet and the target configuration is thus valid.
    For any later vertex $v$, assume the corresponding configuration $C$ has no appropriately-labeled edge, that is moving $v$ from $L$ to $S$ in $C$ does not yield a valid configuration.
    We now individually check the remaining validity constraints required for a valid configuration.
    \Cref{prop:val-S-nonempty} is trivially satisfied.
    \Cref{prop:val-no-long-edge,prop:val-no-half-edge} cannot be violated, as the solution would otherwise contain an edge that is long enough to violate \Cref{prop:sol-bounded-alternation-below-edge}. %
    Finally, \Cref{prop:val-P-avoided} being violated would mean that the solution does not avoid $P$, a contradiction.
    \Cref{prop:val-bounded-alternation} is respected as we initially chose an appropriate solution.
    Thus, the next state is always valid and can be found via an appropriately-labeled edge, which means that we find a path in $D$ from the starting configuration.
    The path ends at a final configuration because $L$ must be empty after removing $|V(G)\setminus U|$ vertices, corresponding to the length of the path.
\end{proof}

For the converse direction, we obtain a vertex order from a solving path by using the vertex-types annotated to the edges of $D$, translating each vertex-type into any yet-unused vertex from its equivalence class, and prepending this order to $\prec_U$.
When doing so, we never run out of unused vertices and obtain an order containing all vertices of $G$ because the starting state had $L=\multiset{V(G)\setminus U}$ and each edge on the path corresponds to removing exactly one vertex from $L$ until $L=\emptyset$ in a final configuration.
It remains to show that the total order $\prec_G$ obtained in this way does indeed avoid the pattern $P$.
For the sake of contradiction, assume that the pattern is not avoided and first consider the case where the only forced edge $uv$ of the pattern matches two free vertices of $\prec_G$, i.e., that are not in $U$.

\begin{lemma}
    \label{lem:poly-trees-dp-dp-no-free-pattern}
    Let $\Instance = \InstanceLongOneS$ be an instance of \PAOneFOneS and let $\prec$ be an order obtained from a solving path in $D$.
    Let $E^+(P) = \{uv\}$ with $u \prec_P v$ and let $ab \in E(G)$ with $a \prec b$ and $a,b \notin U$.
    If $\prec$ matches $P$, then $\delta^{-1}(u)\delta^{-1}(v) \neq ab$.
\end{lemma}
\begin{proof}
    We first show that the solving path contains at least one configuration that has both $a$ and $b$ in $S$.
    Consider the edge on the path that adds $a$ to $S$.
    At this point, $b$ can no longer be in $L$.
    As the configuration is valid, the edge $ab$ cannot be too long and, as $a\in S^L$, vertex $b$ cannot be in $R$.
    Thus, $b\in S$.

    Now assume that $\delta^{-1}(u)\delta^{-1}(v) = ab$ in the bijection that matches $P$ in $\prec$.
    In case we have $U=\emptyset$, we know that $v$ is the last vertex of the pattern (\Cref{prop:components-contain-fixed-v}), and thereby $S$ alone matches the pattern.
    Otherwise, $U$ contains a large independent set $I_S$ that is also not adjacent to $\Image(\delta)\cap S$ due to \Cref{prop:fixed-independent-set}.
    We obtain a new bijection $\delta'$ from $\delta$ by replacing all vertices in its image that are right of $b$ in $\prec$ with a vertex from $I_S\subset U$.
    Note that this cannot violate any forbidden edge as $I_S$ is an independent set that is also not adjacent to any vertex of $S$.
    As $\Size{I_S}=n_P$, there are also sufficiently many distinct vertices for this replacement.
    Observe that the image of $\delta'$ is now a subset of $S\cup U$ as all vertices between $a$ and $b$ are in $S$ and $u$ with $\delta^{-1}(u)=a$ is the first vertex of the pattern.
    Thus, $\delta'$ matches $P$ in $S\oplus\prec_U$, a contradiction to the configuration being valid.
\end{proof}

The following lemma is a pendant to \Cref{lem:poly-trees-dp-dp-no-free-pattern} but deals with edges that start at a free vertex and end at a fixed vertex.

\begin{lemma}
    \label{lem:poly-trees-dp-dp-no-free-fixed-pattern}
    Let $\Instance = \InstanceLongOneS$ be an instance of \PAOneFOneS and let $\prec$ be an order obtained from a solving path in $D$.
    Let $E^+(P) = \{uv\}$ with $u \prec_P v$ and let $ab \in E(G)$ with $a \prec b$ as well as $a \notin U$ and $b \in U$.
    If $\prec$ contains $P$, then $\delta^{-1}(u)\delta^{-1}(v) \neq ab$.
\end{lemma}
\begin{proof}
    Assume that $\prec$ contains $P$ with a realization with $\delta^{-1}(u) = a$, $\delta^{-1}(v) = b$, and $a \notin U$ while $b \in U$.
    Observe that $a \prec b$ since $\prec$ has to end in $U$.
    The overall argument works similar to the proof above, although we now face the more challenging task to replace the vertices in the image of $\delta$ that are in $R$ with suitable vertices from $S$.
    For this, let $S^P\coloneqq \Image(\delta)\cap S$, $R^P\coloneqq \Image(\delta)\cap R$, and $U^P\coloneqq \Image(\delta)\cap U$.
    First, observe that if $S$ has not yet reached the bound on its length, we have $R^P\subset R=\emptyset$, thus $S\ \oplus \prec_U$ induces the pattern, which contradicts the configuration being valid.

    Observe that \Cref{prop:val-bounded-alternation} ensures the solution we obtain satisfies \Cref{prop:sol-long-alternation-compressible}, %
    and we can thus assume that any subsequence of a solution of length at least $\alpha\coloneqq\boundSequenceLeaves$ contains not only $n_P$ leaf-types.
    Furthermore, \Cref{prop:interesting-degree-bounded} ensures that the interesting degree in the instance is bounded via $d_{\max}^{>1} \leq \boundInterestingDegree$, which especially applies to the vertices in $S^P$ and $U^P$.
    Thereby, we have at most $\beta\coloneqq d_{\max}^{>1} \cdot n_P \leq 38 {n_P}^2$ non-leaf neighbors of $S^P\cup U^P$, which bounds the number of vertices we cannot use to replace $R^P$ because they possibly add a forbidden edge to $S^P$ or $U^P$.
    To find sufficiently many replacement vertices for $R^P$, we proceed as follows.
    We partition the sequence $S$ (excluding its very first vertex $a$) into $2n_P$ subsequences each of length $(\alpha+1)\cdot(\beta+1)=(\boundSequenceLeaves+1)\cdot(38 {n_P}^2+1)$, each of which is thus guaranteed to contain a vertex that is not adjacent to $S^P\cup U^P$, which especially also includes $a$.
    As we are in a forest, a set of $2n_P$ vertices contains an independent set of size $n_P$ (\Cref{obs:forests-independent-set}), which is thereby sufficient to replace $R^P$ in the image of~$\delta$.
\end{proof}

By combining the above lemmas, we directly obtain:

\begin{theorem}
    \label{thm:poly-trees-dp}
    For an instance $\Instance = \InstanceLongOneS$ of \PANormalized\ \PAOneFOneS, we can find a solution $\prec_G$ if it exists in time ${n_G}^{{n_P}^{\BigO{{n_P}}}}$ using dynamic programming.
\end{theorem}
Moreover, by invoking \Cref{lem:poly-trees-preprocessing-touring-reduction} we can establish:
\begin{theorem}
    \label{thm:poly-trees}
    \PAOneF parameterized by the pattern size~$n_P$ is in~\XP.
\end{theorem}

\subsection{An Extension to Left Separated Patterns}
\label{sec:poly-trees-extension} 
\Cref{thm:poly-trees} shows that we can solve \PAOneFOneS on forest in polynomial time for any constant size pattern $P$ where $m_P^+ \neq 1$ or the single forced edge is incident to the leftmost vertex in $\prec_P$.
In this section, we show that we can weaken the latter condition to what we call left-separated patterns, defined as follows.
\begin{definition}
    \label{def:poly-trees-preprocessing-left-separated}
    Let $P$ be a pattern with $E^+(P) = \{uv\}$ and $u \prec_P v$.
    We call $P$ \emph{left-separated} if there does not exist a forbidden edge $xy \in E^-(P)$ with $x \prec_P u \preceq_P y$.
\end{definition}

Recalling Subsection~\ref{sub:forests}, we can now formally define the notion of mixed patterns, which are the only ones for which our algorithm does not apply. To this end, let \emph{right-separated} be the symmetric notion to \Cref{def:poly-trees-preprocessing-left-separated}.
\begin{definition}
    \label{def:mixed-pattern}
    Let $P$ be a pattern.
    We call $P$ \emph{mixed} if $m_P^+ = 1$ and $P$ is neither left-separated nor right-separated.   
\end{definition}

Now, the final step towards our main result is a reduction from instances of \PA on forests with left-separated patterns to \PAOneF; this is handled in the following lemma.
\begin{lemma}
    \label{lem:poly-trees-left-sided}
    Let $\Instance = \InstanceLong$ be an instance of \PA where $G$ is a forest, $P$ is left-separated, and $E^+(P) = \{uv\}$, $u \prec_P v$.
    There exists a family $\mathcal{F}$ of instances of \PAOneF with the following properties:
    \begin{enumerate}
        \item we can construct $\mathcal{F}$ in $\BigO{{n_G}^{5n_P + 2}}$ time,
        \item $\Size{\mathcal{F}} \in \BigO{{n_G}^{2n_P} \cdot (2n_P)!}$,
        \item for every $(G', P') \in \mathcal{F}$ we have $P' = P[V(P) \setminus V_L(P)]$, and
        \item \Instance is a positive instance if and only if some instance $\Instance' \in \mathcal{F}$ is a positive instance.
    \end{enumerate}
\end{lemma}
\begin{proof}
    We enumerate all $\binom{n_G}{2n_{L(P)}}$-many subsets $V' \subseteq V(G)$ of size $\Size{V'} = 2 n_{L(P)}$ and all their potential linear orders $\prec_{V'}$.
    Observe that $m_{L(P)}^+ = 0$ by definition.
    Thus, by \Cref{lem:forests-no-forced-edge-no-instance}, every such $\prec_{V'}$ contains the pattern $L(P)$; observe that $G[V']$ is a forest.
    We apply \Cref{lem:poly-trees-preprocessing-pattern-on-left} to every of the above-considered linear orders $\prec$  to compute the minimum value $\ell_{\prec}$ such that the first $\ell_{\prec}$ vertices in $\prec$ contain $L(P)$.
    All these $\ell_{\prec}$ are well defined as every $\prec$ contains $L(P)$.
    By \Cref{lem:poly-trees-preprocessing-pattern-on-left}, this takes $\BigO{{(2n_{L(P)})}^{n_P} \cdot {n_P}^2}$ time per linear order $\prec$ and there are $\BigO{{n_G}^{2 n_{L(P)}} \cdot (2n_{L(P)})!}$ such linear orders.
    As $n_{L(P)} \leq n_P$, this amounts to $\BigO{{n_G}^{2 n_P} \cdot (2n_P)! \cdot {(2n_P)}^{n_P} \cdot {n_P}^2}$ time overall, which is in $\BigO{{n_G}^{5 n_P + 2}}$.

    We now construct the sought-after family $\mathcal{F}$ by adding, for every linear order $\prec$, a new instance $(G_{\prec}, P')$, where $G_{\prec} = G[V(G) \setminus V(\prec\mid_{\ell_{\prec}})]$, i.e., the graph induced on all but the first $\ell_{\prec}$ vertices of $\prec$, and $P' = P[V(P) \setminus V_L(P)]$.
    Observe that $P'$ corresponds to the pattern $P$ after removing its left part.
    Note that $V_L(P)$ is independent of the remaining vertices $V(P) \setminus V_L(P)$ since $P$ is left-separated.
    The number of instances in $\mathcal{F}$ can be bounded by $\BigO{{n_G}^{2n_P} \cdot (2n_P)!}$ due to above reasoning.

    Up to now, we have established Properties (1.)--(3.) of the statement.
    It remains to show Property~(4.), i.e., that \Instance admits a solution if and only if at least one instance $\Instance' \in \mathcal{F}$ admits a solution.
    \proofsubparagraph{($\boldsymbol{\Rightarrow}$)}
    Assume that \Instance admits a solution $\prec_G$.
    Again, we use \Cref{lem:poly-trees-preprocessing-pattern-on-left} to compute the value $\ell_{\prec_G}$ such that the first $\ell_{\prec_G}$ vertices of $\prec_G$ contain the left part $L(P)$ of $P$.
    Since $m_{L(P)}^+ = 0$, $\ell_{\prec_G}$ must exist and be $\ell_{\prec_G} \leq 2\cdot n_{L(P)}\leq 2\cdot n_P$.
    As $\prec_G$ is a solution, this implies that the remaining part of $\prec_G$, let it be the vertices $X$, cannot contain the remaining part $P'$ of~$P$; recall that $P$ is left-separated.
    Therefore, $\prec_G\mid_X$ witnesses that the instance $(G[X], P')$ is a positive instance.
    Note that $(G[X], P') \in \mathcal{F}$ must hold since we consider all possible subsets of size $2n_{L(P)}$ and their linear orders $\prec$, in particular the one with which $\prec_G$ starts.
    Thus, $\mathcal{F}$ contains at least one positive instance. 

    \proofsubparagraph{($\boldsymbol{\Leftarrow}$)}
    Assume that there exists a positive instance $\Instance' = (G_{\prec}, P') \in \mathcal{F}$.
    Let $\prec'$ be the linear order witnessing that $\Instance'$ is a positive instance.
    Consider the value $\ell_{\prec}$.
    Note that the linear order $\prec$ used to construct $\Instance'$ consists of $2\cdot n_{L(P)}$ vertices instead of all vertices of $G_{\prec}$.
    Moreover, $G_{\prec}$, and thus the solution $\prec'$ of $\Instance'$, contains all but the first $\ell_{\prec}$ vertices of $\prec$.
    We construct a linear order $\prec_G$ of $G$'s vertices by concatenating $\prec\mid_{\ell_{\prec}}$ and $\prec'$.
    By above reasoning, $\prec_G$ is a well-defined total order of $V(G)$.
    We now argue that it is also a solution to \Instance and assume, towards a contradiction, that this is not the case.
    Let $X \subseteq V(G)$ be a subset of vertices such that $(\prec_G, X)$ realizes $P$.
    Moreover, let $w_{\ell}$ be the $\ell_{\prec}$th vertex in $\prec_G$ and $u^* \in X$ be the vertex mapped to the left endpoint of the single forced edge in $P$.
    By construction of $\prec_G$, we have $w_{\ell} \prec_G u^*$.
    Let $X_R$ be the vertices in $X$ mapped to vertices in the part of $P$ that corresponds to $P'$.
    That $w_{\ell} \prec_G u^*$ holds implies that for every $u^* \in X_R$ we have $u^* \in V(G_{\prec})$, i.e., is part of the instance $\Instance'$.
    Moreover, $\prec_G\mid_{X_R}$ occurs as a sub-order of $\prec'$.
    Thus, $\prec'$ would contain $P'$, which is a contradiction.
    Hence, $\prec_G$ cannot contain $P$ and \Instance is, therefore, a positive instance.
    Combining both directions establishes Property~(4.).
\end{proof}

Observe that it was crucial in the proof of \Cref{lem:poly-trees-left-sided} that there is no forbidden edge between vertices from the left and right part of the pattern.
Otherwise, we could not simply separate the pattern, and, consequently, the instances, into two independent sub-patterns.

Combining \Cref{lem:poly-trees-left-sided,thm:poly-trees} yields the main result of this section.
\thmforest*

\subsection{Fixing Two Sides of the Solution Makes the Problem \textsf{NP}-hard}
\label{sec:poly-trees-hardness}
In this section, we show that a straight-forward extension of the approach behind \Cref{thm:non-mixed}, i.e., where we fix vertices at the start and end of the sought-after solution, cannot lead to an efficient algorithm for mixed patterns.
On a technical level, we show that the natural extension of \PAOneFOneS to $2$ sides (defined below) is \NP-hard.
As a byproduct that may be of independent interest, we show that checking whether a forest admits a linear order of bandwidth at most two with a fixed start and end vertex is \NP-hard.

\probdef{\PAOneFTwoSLongUnderline~(\PAOneFTwoS)}{A graph $G$, a pattern $P$, two sets $U_L, U_R \subseteq V(G)$, $U_L \cap U_R = \emptyset$, and two total orders $\prec_{U_L}$ on $U_L$ and $\prec_{U_R}$ on $U_R$.}{Does there exist a total order $\prec_G$ of $V(G)$ that starts with $\prec_{U_L}$, ends with $\prec_{U_R}$, and avoids the pattern $P$?}

\thmonesidedhard*

We reduce from the strongly \NP-hard \textsc{Bin Packing} problem \cite{GareyJohnson}, which asks whether, given a finite set $U$ of items with sizes $s(u) \in \mathbb{Z}^+$ for each $u \in U$, a bin capacity $B \in \mathbb{Z}^+$, and an integer $K > 0$, the set $U$ can be partitioned into disjoint subsets $U_1, U_2, \ldots, U_K$ such that $\sum_{u \in U_i} s(u) \le B$ for all $i \in [K]$.

Fix an instance of \textsc{Bin Packing}.
We construct an instance of \PAOneFTwoSLong 
with input graph $G$,
pattern $P$,
first vertex $U_L = \{s\}$ (with the trivial total order $\prec_{U_L}$),
and final vertex $U_R = \{t\}$ (with the trivial total order $\prec_{U_R}$),
as follows.
Let the forbidden pattern $P \coloneqq ( \{a,b,c,d\}, a \prec b \prec c \prec d, \{ad\}, \{\})$, i.e., $P$ is the pattern characterizing the graphs of bandwidth at most two displayed in \cref{fig:one-sided-hardness}~(a).
To construct the input graph $G$, 
we begin with a path of length $B \cdot K + 2\cdot(K-1)$, called the \emph{spine}, with first vertex $s$ and last vertex $t$.
For $1 \leq i \leq K-1$, call the $i \cdot(B+2)$'th vertex of the path the \emph{separator vertex $s_i$}.
Attach two pendant vertices to each separator vertex.
We call the $K$ subpaths between consecutive vertices in the sequence $s, s_1, s_2, \ldots, s_{K-1}, t$ (including both endpoints of each subpath) the \emph{bin paths}. Here, we call the subpath from $s$ to $s_1$ the first bin path, and number the remaining bin paths in sequence.
Finally, we construct the \emph{item gadgets}. For each item $u \in U$, add a disjoint path of length $s(u)-1$ (i.e., a path with $s(u)$ vertices) to $G$.
See \cref{fig:one-sided-hardness} for an example of the reduction.
As \textsc{Bin Packing} is strongly \NP-hard, the construction is feasible in polynomial time. To establish \cref{thm:one-sided-hard}, it remains to show that the reduction is correct.

\begin{figure}
    \centering
    \includegraphics{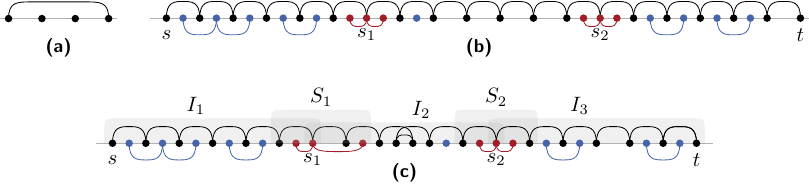}
    \caption{\textbf{\textsf{(a)}} The bandwidth-2 pattern $P$.
    \textbf{\textsf{(b)}} A solution for a \PAOneFTwoS instance constructed as described in the $(\Rightarrow)$-direction of \cref{lem:one-sided-hardness-correctness}, for a \textsc{Bin Packing} instance with $K=3$ bins, per-bin capacity $B=5$, and items of sizes $3,2,1,2,2$, partitioned into the three bins
    $(3,2)$, $(1)$, and $(2,2)$.  Item gadgets are drawn in blue, the spine in black (except separator vertices), and separator vertices and their pendants in red.
    \textbf{\textsf{(c)}} A solution for a \PAOneFTwoS instance corresponding to the same \textsc{Bin Packing} instance as before, with the intervals defined in the $(\Leftarrow)$-direction of \cref{lem:one-sided-hardness-correctness} marked in gray.
    }
    \label{fig:one-sided-hardness}
\end{figure}

\begin{lemma}\label{lem:one-sided-hardness-correctness}
    The \textsc{Bin Packing} instance is positive if and only if the instance of \PAOneFTwoSLong is.
\end{lemma}
\begin{proof}
    $(\Rightarrow):$
    Let $U_1, \dots, U_K$ be a solution to the bin packing instance.
    Let $\prec$ be the total order of $V(G)$ obtained as follows.
    First, insert the vertices of the spine in order of increasing distance from $s$. Second, insert the pendant vertices of each separator vertex immediately before and after the vertex.
    Third, for each $U_i$, take any sequence of the items and, for each item, place its gadget's vertices contiguously in path order; then insert the resulting concatenated sequence into the
     positions in $\prec$ between consecutive vertices of bin path $i$.
    As by construction, each bin path has length $B + 1$ for the first and last path and $B + 2$ otherwise, there are enough free spots (accounting for the pendant vertices) to carry out the construction. See \cref{fig:one-sided-hardness}~(b).
    It is easy to verify that in the constructed order, no edge spans more than one vertex, i.e., $\prec$ certifies a bandwidth of at most two and hence avoids $P$.
    
     $(\Leftarrow):$
     Let $\prec$ be a total order of $V(G)$ avoiding $P$ where $s$ is the first element in $\prec$ and $t$ is the last.
     First, we make some observations about the structure of the solution.
     Let $s_i$ be a separator vertex, and let $a,b$ be its pendant vertices with $a \prec b$.
     Further, let $s_i^{-1}$ be the last vertex of bin path $i$, and let $s_i^{+1}$ be the first vertex of bin path $i+1$.
     Note that $N_G(s_i) = \{ s_i^{-1},s_i^{+1},a,b \}$.
     Let $S_i$ be the interval in $\prec$ consisting of the two direct predecessors of $s_i$, then $s_i$, then the two direct successors of $s_i$, and set $S_i \coloneqq (x,y,z,\alpha,\beta)$.
     We claim that $S_i$ is a permutation of $N_G(s_i) \cup \{s_i\}$ with $\{ s_i^{-1},a\} \prec s_i \prec \{ s_i^{+1},b\}$.
     The first part of the statement (and that $s_i$ is the central element of $S_i$, i.e., $z = s_i$) follows immediately from $(G,\prec)$ avoiding the bandwidth-2 pattern $P$,
     i.e., the endpoints of every edge are at most two positions apart in $\prec$, and from $|N_G(s_i)|=4$.
     For the second part it suffices to derive $s_i^{-1} \prec s_i \prec s_i^{+1}$. Suppose $s_i \prec s_i^{-1}$ (resp.\ $s_i^{+1} \prec s_i$ ).
     Following the vertices of the spine from $s_i^{-1}$, we must eventually reach vertex $s$ occupying the first position (resp. $t$ occupying the last position) in $\prec$. But this means there is a spine edge ``jumping over'' both $x$ and $y$ (resp.\ $\alpha$ and $\beta$), inducing the forbidden pattern $P$, a contradiction.

     Suppose there is a spine edge $xy$ with $x \prec y$ that ``jumps'' over a separator vertex $s_i$, i.e., $x \prec s_i \prec y$.
     If both $x,y \in S_i$, by the above observation the edge is adjacent to $s_i$, a contradiction.
     If otherwise, w.l.o.g., $x \in S_i$ and $y \not\in S_i$, we have $x \prec s_i$, and since by the above observation, $s_i$ is the central element of $S_i$, $x$ and $y$ are at least four positions apart in~$\prec$, a contradiction.
     Similarly, there is no item-gadget edge $xy$ with $x \prec s_i \prec y$:
     From the above, we have
     $\{x,y\} \cap S_i = \emptyset$, hence $xy$ and any two vertices of $S_i$ induce $P$, a contradiction.

     Let $(I_j)_{j\in [K] }$ be the set of intervals in $\prec$ with $I_1 = [s,s_1], I_2=[s_1,s_2], \dots, I_K = [s_{K-1},t]$.
     It follows that for each $i \in [K]$, each vertex of bin path $i$
     is contained in $I_i$,
     and for each item gadget, there is exactly one $I_i$ containing all of the gadget's vertices.
     See \cref{fig:one-sided-hardness}~(c) for an example annotated with both types of intervals.

     Finally, we are able to ``read off'' a bin packing.
     Let $U_1, \dots, U_K$ be the partition of $U$, where $u \in U_i$ if{}f the vertices of $u$'s item gadget are in $I_i$. By the above, the partition is well-defined.
     We claim that $U_1, \dots, U_K$ is a solution for the bin packing instance, i.e., it remains to show $\sum_{u\in U_i} s(u) \leq B$ for each $i \in [K]$.
     Let $i \in [K]$.
     Observe that, since the spine connects $s$ and $t$,
     for each item-gadget or pendant vertex $v$ in $I_i$, there is an edge $xy$ of bin path $i$ ``above'' $v$, i.e., $x \prec v \prec y$.
     Further, observe that bin path $i$ has length $B$ plus the number of pendant vertices in $I_i$,
     and contains $\sum_{u \in U_i}s(u)$ many item-gadget vertices by construction.
     Hence, if $\sum_{u \in U_i}s(u) > B$, by the pigeonhole principle, there is an edge $xy$ ``under'' which there are two vertices $\alpha, \beta$, i.e., $x \prec \alpha \prec \beta \prec y$.
     But this yields precisely the forbidden pattern $P$, a contradiction.
\end{proof}

As a corollary from \Cref{thm:one-sided-hard} we obtain:
\corbandwidthhard*

Finally, we lift \Cref{thm:one-sided-hard} to connected graphs at the cost of a slightly more complex pattern and two fixed vertices on the left side. 
\begin{theorem}\label{thm:one-sided-hard-on-trees}
    \PAOneFTwoSLong is \NP-hard even when the input graph $G$ is a tree, $n_P=5$, $m_P = m_P^+ = 1$, $\Size{U_L} = 2$ and $\Size{U_R} = 1$.
\end{theorem}
\begin{proof}
    Carry out the same construction as before, but add a new vertex $u$ adjacent to one vertex in each component (say to $s$ and to one endpoint of each item gadget), so the resulting graph is a tree. We require that the ordering begins with $u, s$ and must (still) end with $t$. Further, add an isolated vertex before the other vertices in $P$ to obtain the new forbidden pattern $P'$.
    Call the original instance of \PAOneFTwoS $I$ and the new $I'$.
    We claim that $I$ and $I'$ are equivalent.
    Suppose $I$ admits a solution $\prec$.
    Then the order $\prec'$ obtained by placing $u$ in front of $\prec$ is a solution for $I'$: if $\prec'$ contained $P'$, then either the added isolated vertex is $u$, in which case removing it yields a copy of $P$ in $\prec$, or the copy of $P'$ avoids $u$, in which case its four original vertices already form a copy of $P$ in $\prec$.  Either case contradicts that $\prec$ avoids $P$.
    Conversely, suppose $I'$ admits a solution $\prec'$.
    Removing $u$ from $\prec'$ yields an order $\prec$ of $V(G)$ that avoids $P$, since any occurrence of $P$ in $\prec$ together with $u$ would form $P'$ in $\prec'$. Thus $I$ is a yes-instance.
\end{proof}

\section{Concluding Remarks}
\label{sec:conclusion}
In this paper, we initiate a parameterized investigation of the complexity of \PA. Due to the generality of the problem, one cannot hope to achieve tractability under standard graph parameters such as treewidth as these are ruled out already w.r.t.\ very simple patterns. At the same time, even after our most recent results, \PA\ proves to be a challenging problem for which several open questions remain:

\begin{enumerate}
	\item Is \PA\ in \XP\ w.r.t.\ $n_P$ on forests? The results in \Cref{sec:to:poly-trees} answer this affirmatively for all but a single case---however, the lower bound provided in Theorem~\ref{thm:one-sided-hard} can be seen as a strong indication of intractability for the remaining non-left separated case (cf.\ Definition~\ref{def:poly-trees-preprocessing-left-separated}).
	\item What is the complexity of \PA\ w.r.t.\ $n_P$ plus the treedepth of $G$? Here, a natural first step would be to settle the complexity for the rainbow patterns, i.e., to settle the \textsc{Queue Number} problem w.r.t.\ the treedepth of $G$~\cite{bgmn-paql-22}.
	\item What is the complexity of \PA\ w.r.t.\ $m_P$ plus the vertex integrity of $G$? This once again is tied to settling a special case first: the parameterized complexity of bandwidth w.r.t.\ the vertex integrity is explicitly noted as open in the recent work targeting that problem~\cite{GKK+.BPC.2025}.
\end{enumerate}

\bibliography{references}

\end{document}